%% file: main.tex
\documentclass[acmsmall,screen,nonacm]{acmart}
\input{preamble_acm.tex}
\input{preamble.tex}

\input{commands.tex}

\usepackage[normalem]{ulem}
\newcommand{\msout}[1]{\text{\sout{\ensuremath{#1}}}}

\begin{document}

\title{Realizability in Semantics-Guided Synthesis Done Eagerly}

\author{Roland Meyer}
\email{roland.meyer@tu-braunschweig.de}
\orcid{0000-0001-8495-671X}
\affiliation{%
  \institution{TU Braunschweig}
  \country{Germany}
}

\author{Jakob Tepe}
\email{j.tepe@tu-braunschweig.de}
\orcid{0009-0002-8177-4675}
\affiliation{%
  \institution{TU Braunschweig}
  \country{Germany}
}

\author{Sebastian Wolff}
\email{sebastian.wolff@nyu.edu}
\orcid{0000-0002-3974-7713}
\affiliation{%
 \institution{New York University}
 \country{USA}
}


\begin{abstract}
\input{abstract.tex}
\end{abstract}

\maketitle

\input{sections-main.tex}


\bibliographystyle{ACM-Reference-Format}
\bibliography{sample-base,bibliography}

\input{sections-appendix-main.tex}

\end{document}

%% file: preamble_acm.tex
\AtBeginDocument{%
  }

\setcopyright{acmcopyright}
\copyrightyear{2018}
\acmYear{2018}
\acmDOI{XXXXXXX.XXXXXXX}

\acmConference[Conference acronym 'XX]{Make sure to enter the correct
  conference title from your rights confirmation emai}{June 03--05,
  2018}{Woodstock, NY}
\acmPrice{15.00}
\acmISBN{978-1-4503-XXXX-X/18/06}



%% file: preamble.tex
\usepackage[utf8]{inputenc}
\usepackage[T1]{fontenc}
\usepackage{microtype}

\usepackage{tikz}
\usetikzlibrary{automata, arrows.meta, positioning, shapes}

\usepackage{mathtools}
\usepackage{multirow}
\usepackage{stackengine}

\usepackage{fontawesome5}

\usepackage{colortbl} 

\usepackage[nameinlink]{cleveref}
\usepackage{xcolor}
\definecolor{materialPurple}{RGB}{98,0,238}
\definecolor{materialRed}{RGB}{176,0,32}
\usepackage{mathtools}
\usepackage{syntax} 
\usepackage{galois}

\usepackage{listings,multicol}
\usepackage{wrapfig}
\usepackage{xspace}

\usepackage{mathpartir}

\lstdefinestyle{condensed}{
  language=Java,
  basicstyle=\footnotesize\ttfamily,
  keywordstyle=\color{teal}\ttfamily,
  commentstyle=\color{gray}\ttfamily,
  tabsize=2,
  morekeywords={NULL,CAS,shared,threadlocal,struct,atomic,delete,inline,CAS2,2CAS,retry,bool,assert,beginAtomic,endAtomic,assume},
  moredelim=**[is][\color{materialPurple}]{~}{~},
  moredelim=**[is][\color{materialRed}]{\#}{\#},
  moredelim=**[is][\color{gray}]{@}{@},
  belowskip=0pt,
  keepspaces=true,
  numbers=left,
  numberstyle=\color{gray}\scriptsize\ttfamily,
  numbersep=6pt,
  mathescape=true,
  escapebegin=\color{gray},
  xleftmargin=10pt,
  literate={\\@}{@}1,
  escapechar=|
}
\newtheorem{corollary}{Corollary}
\newtheorem{definition}{Definition}

\DeclareMathSymbol{\dv}{\mathbin}{operators}{"3A}

\makeatletter
\newcommand{\leqnomode}{\tagsleft@true}
\newcommand{\reqnomode}{\tagsleft@false}
\makeatother

%% file: commands.tex
\newcommand{\baseCase}[1]{\textbf{Base Case:}}
\newcommand{\inductionStep}[1]{\textbf{Induction Step:}}

\newcommand{\definingEquals}{:=}

\newcommand{\postfunction}[1]{\mathit{post}}

\newcommand{\prefunction}[1]{\mathit{pre}}

\newcommand{\skipCommand}[1]{\code{skip}}

\newcommand{\leftHoareParenthesis}[1]{\boldsymbol{\langle}}
\newcommand{\rightHoareParenthesis}[1]{\boldsymbol{\rangle}}
\newcommand{\inHoareParenthesis}[1]{\leftHoareParenthesis{}#1\rightHoareParenthesis{}}
\newcommand{\hoaretriplet}[3]{\inHoareParenthesis{#1} #2 \inHoareParenthesis{#3}}

\newcommand{\leftHoareParenthesisDemonic}[1]{\{}
\newcommand{\rightHoareParenthesisDemonic}[1]{\}}
\newcommand{\inHoareParenthesisDemonic}[1]{\leftHoareParenthesisDemonic{}#1\rightHoareParenthesisDemonic{}}
\newcommand{\hoaretripletDemonic}[3]{\inHoareParenthesisDemonic{#1} #2 \inHoareParenthesisDemonic{#3}}

\newcommand{\code}[1]{\texttt{#1}}

\newcommand{\sizeof}[1]{|#1|}
\newcommand{\powerset}[1]{\mathcal{P}}
\newcommand{\powersetOf}[1]{\powerset{}(#1)}
\newcommand{\vars}[1]{\textit{Vars}}
\newcommand{\join}[1]{\sqcup}
\newcommand{\joinOf}[1]{\sqcup(#1)}
\newcommand{\meet}[1]{\sqcap}
\newcommand{\meetOf}[1]{\sqcap(#1)}
\newcommand{\bottom}[1]{\bot}

\newcommand{\fail}[1]{\mathit{fail}}
\newcommand{\typeEnvironments}[1]{\mathit{TypeEnvironments}}
\newcommand{\typeEnvironmentsP}[1]{\mathit{TypeEnvironments_f}}

\newcommand{\strongestPostSem}[1]{\mathit{sp}}
\newcommand{\strongestPostSemOf}[2]{\strongestPostSem{}(#1, #2)}

\newcommand{\leftSemBracket}[1]{\llbracket}
\newcommand{\rightSemBracket}[1]{\rrbracket}

\newcommand{\synToSem}[1]{\leftSemBracket{} #1 \rightSemBracket{}}

\newcommand{\set}[1]{\lbrace #1 \rbrace}
\newcommand{\setCond}[2]{\set{#1 \mid #2}}

\newcommand{\angelicPrograms}[1]{\mathit{AngelicPrograms}}
\newcommand{\concFunc}[1]{\mathit{conc}}

\newcommand{\demonConcFunc}[1]{\mathit{exec}}
\newcommand{\demonConcFuncOf}[1]{\demonConcFunc{}(#1)}
\newcommand{\angelConcFunc}[1]{\mathit{drv}}
\newcommand{\angelConcFuncOf}[1]{\angelConcFunc{}(#1)}

\newcommand{\states}[1]{\mathit{States}}
\newcommand{\funcDef}[2]{#1: #2}

\newcommand{\demonicAssertions}[1]{\mathit{Predicates}}
\newcommand{\angelicAssertions}[1]{\mathit{Selections}}

\newcommand{\straightLineProgramLang}[1]{\code{Execs}}
\newcommand{\demonicProgramLang}[1]{\code{Programs}}
\newcommand{\angelicProgramLang}[1]{\code{Sketches}}
\newcommand{\annotatedAngelicProgramLang}[1]{\angelicProgramLang{}}

\newcommand{\typeEnvStronger}[1]{\preceq}

\newcommand{\demonicAssertionStronger}[1]{\typeEnvStronger{}_{d}}
\newcommand{\demonicAssertionStrongerOf}[2]{#1 \demonicAssertionStronger{} #2}
\newcommand{\angelicAssertionStronger}[1]{\typeEnvStronger{}_{a}}
\newcommand{\angelicAssertionStrongerOf}[2]{#1 \angelicAssertionStronger{} #2}

\newcommand{\angelicAssertionAbsStronger}[1]{\angelicAssertionStronger{}^{\#}}
\newcommand{\angelicAssertionAbsStrongerOf}[2]{#1 \angelicAssertionAbsStronger{} #2}

\newcommand{\demonicAssertionAbsStronger}[1]{\demonicAssertionStronger{}^{\#}}
\newcommand{\demonicAssertionAbsStrongerOf}[2]{#1 \demonicAssertionAbsStronger{} #2}

\newcommand{\angelicAssertionFont}[1]{#1}

\newcommand{\progSemFunc}[1]{\synToSem{#1}}
\newcommand{\progSemFuncOf}[2]{\progSemFunc{#1}(#2)}

\newcommand{\progSemFuncAbstract}[1]{\progSemFunc{#1}^{\#}}
\newcommand{\progSemFuncAbstractOf}[2]{\progSemFuncAbstract{#1}(#2)}

\newcommand{\progSemFuncOrig}[1]{\progSemFunc{#1}}
\newcommand{\progSemFuncOrigOf}[2]{\progSemFuncOrig{#1}(#2)}

\newcommand{\progSemFuncTypes}[1]{\progSemFunc{#1}_{t}}
\newcommand{\progSemFuncTypesOf}[2]{\progSemFuncTypes{#1}(#2)}

\newcommand{\ruleLabelSmall}[1]{{\scriptsize(\textbf{#1})}}
\newcommand{\ruleLabel}[1]{{\footnotesize(\textbf{#1})}}

\newcommand{\demonicHoareTripletHolds}[3]{\vdash_{d} \hoaretripletDemonic{#1}{#2}{#3}}
\newcommand{\angelicHoareTripletHolds}[3]{\vdash_{a} \hoaretriplet{#1}{#2}{#3}}
\newcommand{\demonicHoareTripletHoldsSemantically}[3]{\models_{d} \hoaretripletDemonic{#1}{#2}{#3}}
\newcommand{\angelicHoareTripletHoldsSemantically}[3]{\models_{a} \hoaretriplet{#1}{#2}{#3}}

\newcommand{\verificationConditionsFunc}[1]{\mathit{vc}}
\newcommand{\verificationConditionsFuncOf}[1]{\verificationConditionsFunc{}(#1)}

\newcommand{\assertionSet}[1]{\mathit{VerifConds}}
\newcommand{\assertionSetAbs}[1]{\assertionSet^{\#}}
\newcommand{\vcAngelicHoareTripletSet}[1]{\mathit{annRealizabilityTriples}}
\newcommand{\vcAngelicHoareTripletSetAbs}[1]{\mathit{annRealizabilityTriples}^{\#}}

\newcommand{\vcStrongestPostSem}[1]{\strongestPostSem{}}
\newcommand{\vcStrongestPostSemOf}[2]{\vcStrongestPostSem{}(#1, #2)}

\newcommand{\demonicAssertionsAbs}[1]{\demonicAssertions{}^{\#}}
\newcommand{\angelicAssertionsAbs}[1]{\angelicAssertions{}^{\#}}

\newcommand{\demonicAbsFunc}[1]{\alpha_{d}}
\newcommand{\demonicAbsFuncOf}[1]{\demonicAbsFunc{}(#1)}

\newcommand{\demonicConcFunc}[1]{\gamma_{d}}
\newcommand{\demonicConcFuncOf}[1]{\demonicConcFunc{}(#1)}

\newcommand{\angelicAbsFunc}[1]{\alpha_{a}}
\newcommand{\angelicAbsFuncOf}[1]{\angelicAbsFunc{}(#1)}

\newcommand{\angelicConcFunc}[1]{\gamma_{a}}
\newcommand{\angelicConcFuncOf}[1]{\angelicConcFunc{}(#1)}

\newcommand{\strongestPostSemAbs}[1]{\strongestPostSem{}^{\#}}
\newcommand{\strongestPostSemAbsOf}[2]{\strongestPostSemAbs{}(#1, #2)}

\newcommand{\absAssertionSet}[1]{\assertionSet{}^{\#}}

\newcommand{\vcStrongestPostSemAbs}[1]{\vcStrongestPostSem{}^{\#}}
\newcommand{\vcStrongestPostSemAbsOf}[2]{\vcStrongestPostSemAbs{}(#1, #2)}

\newcommand{\absAnnotatedAngelicProgramLang}[1]{\mathit{annAngelProgs}^{\#}}

\newcommand{\absVerificationConditionsFunc}[1]{\verificationConditionsFunc{}^{\#}}
\newcommand{\absVerificationConditionsFuncOf}[1]{\absVerificationConditionsFunc{}(#1)}

\newcommand{\rewriteFunc}[1]{\mathit{sc}}

\newcommand{\absAnnotatedDemonicProgramLang}[1]{\mathit{annDemonProgs}^{\#}}

\newcommand{\toSebastianFunc}[1]{\mathit{toSeb}}

\newcommand{\angelProgramProofLang}[1]{\mathit{prfOutls}}
\newcommand{\angelProgramProofLangAbs}[1]{\angelProgramProofLang{}^{\#}}

\newcommand{\proofToProgFunc}[1]{\mathit{sketch}}
\newcommand{\proofToProgFuncOf}[1]{\proofToProgFunc{}(#1)}

\newcommand{\gatherFunc}[1]{\mathit{gather}}
\newcommand{\gatherFuncOf}[2]{\gatherFunc{}(#1, #2)}

\newcommand{\angelicProofHoareTripletHolds}[3]{\vdash_{p} \hoaretriplet{#1}{#2}{#3}}

\newcommand{\vsim}{\mathrel{\scalebox{1}[1.5]{$\shortmid$}\mkern-3.1mu\raisebox{0.15ex}{$\sim$}}}

\newcommand{\angelicProofRewriteHolds}[2]{#1 \vsim #2}

\newcommand{\proofStronger}[1]{\typeEnvStronger{}_p}
\newcommand{\proofStrongerOf}[2]{#1 \proofStronger{} #2}

\newcommand{\partialFuncArrow}[1]{\rightharpoonup}

\newcommand{\predicates}[1]{\mathit{Preds}}

\newcommand{\choice}[1]{+}
\newcommand{\choiceOf}[2]{#1 \choice{} #2}

\newcommand{\angelicChoice}[1]{\;|\;}
\newcommand{\angelicChoiceOf}[2]{#1 \angelicChoice{} #2}

\newcommand{\reduceStrengthFunc}[1]{\mathit{jl}}

\newcommand{\restrictStrengthFunc}[1]{\mathit{jr}}

\newcommand{\rewriteToProgFunc}[1]{\mathit{sc}}

\newcommand{\rewriteToDemonProgFunc}[1]{\rewriteToProgFunc{}}

\newcommand{\automaton}[1]{\mathcal{O}}

\newcommand{\activeGuarantee}[1]{\mathbb{A}}
\newcommand{\localGuarantee}[1]{\mathbb{L}}
\newcommand{\safeGuarantee}[1]{\mathbb{S}}
\newcommand{\ELGuarantee}[1]{\mathbb{E}_L}
\newcommand{\EisuGuarantee}[1]{\mathbb{E}_{\mathit{isu}}}
\newcommand{\EinvGuarantee}[1]{\mathbb{E}_{\mathit{inv}}}
\newcommand{\nothingGuarantee}[1]{\automaton{}}

\newcommand{\sebHoareTripletDash}[1]{\vdash_{t}}
\newcommand{\sebHoareTripletHolds}[3]{\sebHoareTripletDash{}\hoaretriplet{#1}{#2}{#3}}

\newcommand{\killb}[1]{\mathit{kill}}
\newcommand{\gen}[1]{\mathit{gen}}

\newcommand{\makeBlue}[1]{{\color{blue}#1}}

\newcommand{\unroll}[1]{\mathit{ro}}

\newcommand{\strongestPostDem}[1]{\mathit{sp}_d}

%% file: abstract.tex
We present realizability and realization logic, two program logics that jointly address the problem of finding solutions in semantics-guided synthesis. 
What is new is that we proceed eagerly and not only analyze a single candidate program but a whole set.
Realizability logic computes information about the set of candidate programs in a forward fashion. 
Realization logic uses this information as guidance to identify a suitable candidate in a backward fashion.
Realizability logic is able to analyze a set of programs due to a new form of assertions that tracks synthesis alternatives. 
Realizability logic then picks alternatives to arrive at a program, and we give the guarantee that this process will not need backtracking. 
We show how to implement the program logics using verification conditions, and report on experiments with a prototype in the context of safe memory reclamation for lock-free data structures.

%% file: sections-main.tex
\input{sections-intro-intro.tex}
\input{sections-newgroundtruth-main.tex}
\input{sections-rewriteproofsystem-main.tex}
\input{sections-newgroundtruthalgorithms-main.tex}
\input{sections-rewriteprooftoprog-main.tex}
\input{sections-application-main.tex}
\input{sections-application-evaluation-main.tex}
\input{sections-discussion-main.tex}

%% file: sections-intro-intro.tex

\newcommand{\semgus}{\textsf{SemGuS}\xspace}
\newcommand{\sygus}{\textsf{SyGuS}\xspace}
\newcommand{\rosette}{\textsf{Rosette}\xspace}
\newcommand{\sketch}{\textsf{Sketch}\xspace}
\newcommand{\messy}{\textsf{Messy}\xspace}
\newcommand{\messyenum}{\textsf{Messy-Enum}\xspace}

\newcommand{\hoareof}[3]{\{#1\}#2\{#3\}}
\newcommand{\apred}{\mathit{r}}
\newcommand{\apredp}{\mathit{s}}
\newcommand{\apredpp}{\mathit{t}}
\newcommand{\anonterm}{\code{N}}
\newcommand{\anontermp}{\code{M}}
\newcommand{\xzero}{\code{x=0}}
\newcommand{\yzero}{\code{y=0}}
\newcommand{\xone}{\code{x=1}}
\newcommand{\yone}{\code{y=1}}
\newcommand{\xvar}{\code{x}}
\newcommand{\yvar}{\code{y}}
\newcommand{\apredset}{\mathit{R}}
\newcommand{\apredsetp}{\mathit{S}}
\newcommand{\apredsetpp}{\mathit{T}}
\newcommand{\aprogset}{\mathit{P}}
\newcommand{\aprogsetp}{\mathit{Q}}
\newcommand{\aprog}{\code{prog}}
\newcommand{\aprogp}{\aprog'}
\newcommand{\anexec}{\code{ex}}
\newcommand{\anexecp}{\anexec'}
\newcommand{\concat}{;}
\newcommand{\concatof}[2]{#1\concat#2}
\newcommand{\kleeneof}[1]{#1^{*}}

\newcommand{\bnf}{\; | \;}

\newcommand{\apo}{\code{po}}
\newcommand{\apop}{\apo'}

\newcommand{\mytrue}{\mathit{true}}

\newcommand{\productions}{\mathit{prod}}
\newcommand{\productionsof}[1]{\productions(#1)}

\section{Introduction}\label{Section:Introduction}
The syntax-guided synthesis (\sygus) initiative~\cite{sygus,SearchBased18} has been instrumental in pushing the development of program synthesis technology. 
Key to this success has been the definition of a standardized input format that is \emph{solver independent}: the format only refers to the synthesis task but does not constrain the synthesis technology. 
This means every \sygus solver can, in principle, be applied to the entire \sygus benchmark set, and the community can focus on comparing and improving the synthesis technology.  
A \sygus task consists of a specification of the desired program behavior and a grammar for the programs that may be chosen. 
The inherent limitation of \sygus is that the grammar should refer to SMT expressions, expressions from logical theories that are supported by SMT solvers. 
\sygus cannot handle expressions that fall outside the known logical theories, and for the expressions it can handle it assumes the standard semantics.

Semantics-guided synthesis (\semgus)~\cite{dantoni:semgus} has recently been proposed as a successor of \sygus that is meant to overcome the dependence on SMT expressions and describe synthesis problems in a \emph{domain-independent} way. 
The key to domain independence is to add a third parameter to the definition of the synthesis problem, namely a definition of the program semantics. 
Giving the user the ability to define the program semantics dramatically increases the reach of \semgus over \sygus.  
The user can now define loop constructs and synthesize programs with complex control flow. 
In fact, synthesizing full programs rather than logical expressions has been one of the goals behind \semgus. 
It is a goal that \semgus shares with solver-aided languages like \rosette~\cite{torlak:rosette:onward} or \sketch~\cite{sketch:phd}. 
Note, however, that \rosette and \sketch are neither solver nor domain independent. 
\rosette will not be able to read \sketch input and vice versa, and both can only draw conclusions about the standard semantics of the language.
\semgus allows the user to provide approximate semantics, and applications abound.

Giving the user the ability to define the program semantics also dramatically increases the computational effort of solving \semgus tasks over \sygus tasks.  
To come up with efficient synthesis algorithms, it has turned out helpful to decompose the problem and develop algorithms dedicated to proving synthesis tasks unrealizable and algorithms for synthesizing solutions. 
For unrealizability, the state-of-the-art approach, implemented in the tool \messy~\cite{dantoni:semgus} and its \sygus\ precursors \textsf{nay}~\cite{dantoni:nay} and \textsf{nope}~\cite{dantoni:nope}, is  \emph{unrealizability logic}~\cite{dantoni:UnrealizabilityLogic}. 
Unrealizability logic tries to prove Hoare triples of the form $\hoareof{\apred}{\anonterm}{\apredp}$. 
The novelty is that $\anonterm$ denotes a set of programs, and by proving the triple one shows that all programs in this set satisfy the pre-post specification. 
The pre-post specification is set-up in a way that validity of the triple means no program in the set can solve the \semgus task at hand, hence the name unrealizability.

Our contribution is a new algorithm for realizability, for synthesizing solutions to \semgus tasks. 
The state-of-the-art tool for realizability, called \messyenum~\cite{dantoni:semgus}, implements the CEGIS algorithm~\cite{cegis} algorithm for search-based synthesis~\cite{SearchBased18}. 
Guided by knowledge about failed synthesis attempts, \messyenum iteratively constructs candidate programs and checks every single one of them for validity, i.e., for whether it solves the synthesis task. 
There is, however, an important difference between \sygus and \semgus. 
In \sygus, checking the validity of a candidate program is cheap, it is an SMT query. 
In \semgus, checking the validity of a candidate program is expensive: \messyenum reduces it to a constraint Horn clause (CHC) query, and solving them may turn out as hard as conducting a full verification run.
This is not unexpected, after all \semgus aims to synthesize full programs. 
But it means that search-based solvers have no chance to scale without new strategies to reduce (i) the number of iterations and (ii) the cost of each iteration.

\looseness=-1
Our algorithm reduces the number of iterations by analyzing sets of candidate programs for validity, in an \emph{eager} fashion, rather than a single one. 
What makes this scale is compositionality.   
We construct the sets of programs bottom-up, starting from sets of commands to sets of more and more complex programs. 
The key idea is to abstract the sets of programs to their input-output behavior. 
Abstracting single programs to their input-output behavior to compute over equivalence classes is an idea that has been widely used in synthesis~\cite{FlashFill11,AGK13,Transit13,FlashMETA15,dillig:synthAbstractionRefinement}.
The novelty in our work is eagerness: we abstract sets of programs whose input-output behavior does not match. 
Compositionality requires that the input-output is rich enough (i) to only reason with this abstraction, without having to resort to the underlying programs, and (ii) to answer the realizability problem.

\looseness=-1
We develop this idea in a new program logic~\cite{Hoare69}. 
Our \emph{realizability logic} reasons over triples of the form $\hoaretriplet
{\apredset}{\aprogset}{\apredsetp}$, where $\aprogset$ is a set of programs whose input-output behavior is given in the form of a precondition $\apredset$ and a postcondition $\apredsetp$. 
What is new is that the pre- and postcondition are not single predicates, but sets of predicates. 
This reflects the fact that the programs in $\aprogset$ are synthesis alternatives. 
To achieve compositionality, it is important to get the notion of validity right. 
Given $\hoaretriplet
{\apredset}{\aprogset}{\apredsetp}$ and $\hoaretriplet
{\apredsetp}{\aprogsetp}{\apredsetpp}$, we want to be able to conclude $\hoaretriplet
{\apredset}{\aprogset;\aprogsetp}{\apredsetpp}$, where $\aprogset;\aprogsetp$ contains all programs $\aprog;\aprogp$ with $\aprog\in\aprogset$ and $\aprogp\in\aprogsetp$. 
The right choice is to reason backwards and define validity 
\begin{align*}
\models_a \hoaretriplet
{\apredset}{\aprogset}{\apredsetp}\quad\text{by}\quad \forall \apredp\in\apredsetp.\ \exists\apred\in\apredset.\ \exists\aprog\in\aprogset.\ \models_d \hoareof{\apred}{\aprog}{\apredp}\ . 
\end{align*}
It is worth contrasting this definition  with the notion of validity in the recent unrealizability logic~\cite{dantoni:UnrealizabilityLogic}.
Their triples $\hoareof{\apred}{\aprogset}{\apredp}$ are valid, if $\models_d \hoareof{\apred}{\aprog}{\apredp}$ holds for all programs $\aprog\in\aprogset$. 
Moreover, note that $\apred$ and $\apredp$ are single predicates, and not sets of predicates as in our case.
In fact, the paper explicitly asks for a realizability analogue of their unrealizability logic, and we named our program logic after that proposal.

The relation $\models_d\hoareof{\apred}{\aprog}{\apredp}$ is validity in classical Hoare logic. 
The index $d$ stands for \emph{demonic}, and indicates that the choice of the initial state in $\apred$ as well as the choice of the execution in program $\aprog$ are made demonically, and the postcondition $\apredp$ has to over-approximate all the resulting states.   
The index $a$ for validity $\models_a \hoaretriplet
{\apredset}{\aprogset}{\apredsetp}$ in our synthesis logic stands for \emph{angelic}, and indicates that the choice of the predicate $\apred$ in the precondition $\apredset$ as well as the choice of the program $\aprog\in\aprogset$ are made \emph{angelically}, and the postcondition $\apredsetp$ is an under-approximation of all predicates $\apredp$ that can be guaranteed. 
Assertions in our program logic are thus angelic choices over predicates, and these predicates are demonic choices over states. 
To the best of our knowledge, there is no program logic that would reason over such alternations in related work.

Consider the program proof given in \Cref{code:introa}. 
The program has two Boolean variables $\xvar$ and~$\yvar$ and two non-terminals, $\anontermp$ that may be rewritten to $\xzero$ or $\xone$ and $\anonterm$ that may be rewritten to $\yzero$ or $\yone$. 
Starting from the set of all states, represented by the singleton predicate $\mytrue$, the two programs represented by $\anontermp$ give us two synthesis alternatives: we can guarantee $\xvar=0$ or we can guarantee $\xvar = 1$. 
In both cases, we still do not have any knowledge about the value of $\yvar$.
We do not yet make a decision, but use the two alternatives to consider the programs that can be derived from $\anonterm$.
This yields four synthesis alternatives, namely $\xzero\wedge\yzero$ to $\xone\wedge \yone$. 
Only the last of these alternatives passes the assertion, the others lead to a failure of the execution.

\begin{figure}
    \begin{minipage}[t]{0.4\textwidth}
\begin{lstlisting}
$\inHoareParenthesis{\mytrue}$
M($\inHoareParenthesis{\mytrue}$x = 0$\inHoareParenthesis{\xvar = 0}$ $\color{black}\angelicChoice{}$ 
  $\inHoareParenthesis{\mytrue}$x = 1$\inHoareParenthesis{\xvar = 1}$);
$\inHoareParenthesis{\xvar = 0, \xvar = 1}$
N($\inHoareParenthesis{\xvar = 0, \xvar = 1}$y = 0$\inHoareParenthesis{\xvar = 0\wedge \yvar=0, \xvar = 1\wedge \yvar=0}$ $\color{black}\angelicChoice{}$ 
  $\inHoareParenthesis{\xvar = 0, \xvar = 1}$y = 1$\inHoareParenthesis{\xvar = 0\wedge \yvar=1, \xvar=1\wedge \yvar= 1}$);
$\inHoareParenthesis{x = 0\wedge \yvar=0,\ldots, \xvar = 1\wedge \yvar=1}$ |\label{line:introa:assertionAfterNonterm}|
assert($\xvar=1\wedge \yvar=1$)
$\inHoareParenthesis{\xvar = 1\wedge \yvar=1, \fail{}}$
\end{lstlisting}
    \caption{Proof outline in realizability logic.}
    \label{code:introa}
    \end{minipage}\hfill
    \begin{minipage}[t]{0.4\textwidth}
\begin{lstlisting}
$\inHoareParenthesis{\mytrue}$
M($\inHoareParenthesis{\mytrue}$x = 0$\inHoareParenthesis{\xvar = 0}$ $\color{black}\angelicChoice{}$ 
  $\inHoareParenthesis{\mytrue}$x = 1$\inHoareParenthesis{\xvar = 1}$);
$\inHoareParenthesis{\xvar = 1}$
N($\msout{\inHoareParenthesis{\xvar = 0, \xvar = 1}\code{y = 0}\inHoareParenthesis{\xvar = 1\wedge \yvar=1}}$ $\color{black}\angelicChoice{}$ 
  $\inHoareParenthesis{\xvar = 1}$y = 1$\inHoareParenthesis{\xvar=1\wedge \yvar= 1}$);
$\inHoareParenthesis{\xvar = 1\wedge \yvar=1}$
assert($\xvar=1\wedge \yvar=1$)
$\inHoareParenthesis{\xvar = 1\wedge \yvar=1}$
\end{lstlisting}
    \caption{Proof outline derived from \Cref{code:introa}.}
    \label{code:introb}
    \end{minipage}
\end{figure}

A first point of criticism one may have about realizability logic is that the alternation between angelic and demonic choice will explode too quickly. 
In the end, all we have done is to replace the choice of a program by a powerset construction over predicates, a technique pioneered in automata theory~\cite{RabinScott59}. 
This merely translates a complex search into complex assertions. 
We argue that, for \semgus where validity checks are expensive, the assertions may be the right place to keep the complexity. 
Formal methods has developed expressive logical languages and abstract domains that can denote complex sets of states with very concise assertions. 
In our experiments, we have worked with Cartesian abstraction~\cite[Chapter 9]{cousotBook}.

A more severe point of criticism is that our realizability triples $\hoaretriplet
{\apredset}{\aprogset}{\apredsetp}$ drop the relationship between the single programs $\aprog\in \aprogset$ and the pre- and postconditions $\apred\in\apredset$ and $\apredp\in\apredsetp$ that this program can achieve.
Hence, when we have given a proof in realizability logic, like the one in \Cref{code:introa}, all we know is that the program sketch $\anontermp;\anonterm;\code{assert($\xvar=1\wedge \yvar=1$)}$ can be completed to a program that passes the assertion, but \emph{we do not know the program}. 
What we have, however, is a proof outline in realizability logic that annotates the program with intermediary assertions.  
The idea is to use this proof outline as guidance of how to instantiate the non-terminals.

Our second contribution is \emph{realization logic},  
a program logic to derive rewriting steps $\apo \vsim \apop$ between proof outlines in realizability logic. 
The guarantee given by realization logic is that valid proof outlines are rewritten to valid proof outlines. 
The rewriting process will not only eliminate alternatives from the definition of non-terminals until a program that satisfies the specification has been found.
An equally important step is to eliminate predicates that will not be helpful to satisfy the desired postcondition, or may even fail in the future.
The elimination of predicates is done backwards, starting from the postcondition, and we illustrate it on our example:
\begin{align}
&\hoaretriplet
{\xvar=0\wedge\yvar=0,\ldots, \xvar=1\wedge \yvar=1}{\code{assert($\xvar=1\wedge\yvar=1$)}}{\xvar=1\wedge\yvar=1, \fail{}} \label{eq:ruleApplicationCsq} \\
\vsim & \hoaretriplet
{\xvar=0\wedge\yvar=0,\ldots, \xvar=1\wedge \yvar=1}{\code{assert($\xvar=1\wedge\yvar=1$)}}{\xvar=1\wedge\yvar=1} \label{eq:ruleApplicationBase}\\
\vsim& \hoaretriplet
{\xvar=1\wedge \yvar=1}{\code{assert($\xvar=1\wedge\yvar=1$)}}{\xvar=1\wedge\yvar=1}\ . \nonumber
\end{align}
The first step drops the alternative $\fail{}$.
As there are less synthesis options to choose from, this weakens the postcondition and the step is thus sound by the standard rule of consequence in Hoare logic. 
The second step propagates this elimination backwards by \emph{weakening} the precondition, which is rather uncommon in program logics. 
We can weaken the precondition, because $\xvar=1\wedge \yvar=1$ is sufficient to obtain the postcondition.
For the following rewriting steps, it will be helpful to consider \Cref{code:introb}. 
We propagate the precondition $\xvar=1\wedge \yvar=1$ to the postconditions of the alternatives for non-terminal~$\anonterm$. 
Note that the proof outline remains valid although the alternative $\hoareof{\xvar=0, \xvar=1}{\yzero}{\xvar=1\wedge \yvar=1}$ is invalid. 
The point is that there is still the alternative~$\yone$.   
We eliminate the incorrect alternative in the next rewriting step. 
The backwards reasoning thus revealed that the alternative $\yzero$ will not help to satisfy the specification. 
We then weaken the precondition of the alternative $\yone$ and obtain the proof outline shown in the figure.
The logic will proceed along the same lines and identify the program $\xone;\yone;\code{assert($\xvar=1\wedge\yvar=1$)}$ as a solution to the synthesis task.

Our third contribution is an algorithm to construct valid proof outlines in realizability logic. 
The ambition is to avoid expensive CHC queries and get away with standard SMT queries, very much like \sygus.  
To achieve this, we are willing to rely on user guidance, namely the loop invariants that are known to be crucial for verification. 
Our algorithm is then a generalization of deductive verification to realizability logic. 
We start from a realizability triple whose program sketch comes with invariant annotations. 
We show how to compute the missing intermediary assertions with the help of strongest postconditions. 
From the resulting proof outline, we compute verification conditions: a set of constraints whose validity entails the validity of the proof outline (and the annotations).
Given that our assertions are new, we had to adapt the strongest postconditions and the verification conditions.  
The new definitions are made such that they can be handled by standard technology: the postconditions can be computed in a symbolic way  and the verification conditions can be  discharged by an SMT solver.

\looseness=-1
Our next contribution is to automate realization logic: we give an algorithm that rewrites proof outlines until a solution to the synthesis problem has been found. 
Our algorithm is again based on deductive verification and solver support. 
The idea is to derive from the given proof outline (ordinary) verification conditions (over predicates rather than sets of predicates) whose validity proves that a certain program is a solution to the synthesis problem.
What is new is that the verification condition checks have to be interleaved with their construction, because they control the choice of the program and thus the future verification conditions. 
The choice of verification conditions is not unique but depends on the program as much as the argument why this program solves the synthesis problem. 
Interestingly, we can show that there is a \emph{backtracking freedom} guarantee: no matter the choice, the synthesis is guaranteed to succeed.
At a high level, this is a consequence of the notion of validity in realzability logic.

We implemented our algorithms for realizability and realization logic in a new tool. 
Given that we have not yet developed generic assertion languages to deal with sets of sets of predicates, our implementation is tied to one application domain: memory management in lock-free data structures with the help of a safe memory reclamation algorithm. 
Due to the lock-free processing, protecting a memory cell is not a mere call to the safe memory reclamation algorithm,  
but a complicated sequence of calls followed by checks of global invariants that indicate whether a call has been successful.
The data structures include tricky ones like the ORVYY set and the DGLM queue, the challenging protection is with ordered hazard pointers (the first has priority over the second), but we can also handle epochs.
This set of case studies is interesting in several respects. 
It works over an abstract data domain, namely the SMR types introduced in~\cite{POPL2020}, and therefore cannot be handled by solver-aided languages out-of-the-box (an encoding would be possible).  
The programs contain unbounded loops, and therefore cannot be handled by \sygus solvers.
The places in which to synthesize information are far apart, which means we have to capture the influence of complex code on the to-be-synthesized information.
Our approach handles all instances in a matter of seconds. 
Moreover, we did not need any user annotations but were able to infer the required loop invariants automatically.

It may have become clear that we develop our algorithms in the context of an imperative programming language that is parameterized in the data domain, the set of commands, and the semantics of commands. 
We fix the semantics of choice, sequential composition, and Kleene star. 
In its original formulation, \semgus is more liberal and would also take the operators for building programs as a parameter. 
To generalize our work to that setting, the user would have to come up with appropriate proof rules for the new operators.
This is also the approach taken by unrealizability logic~\cite{dantoni:UnrealizabilityLogic}.
We believe, however, that this should only be a second step. 
The parameterization we work with is the standard assumption in program logics~\cite{AbstractSL07}, and it has proven rich enough to handle a large variety of benchmarks. 
In fact, it is rich enough to handle all examples that are typically given to motivate \semgus (C, Python, regular expressions, and bounded loops). 

To sum up, our contributions are the following.
\begin{itemize}
\item Realizability logic (Section~\ref{ch:pls}), the first program logic to reason about realizability in \semgus. The new idea is to have sets of predicates as assertions to represent synthesis alternatives.
\item Realization logic (Section~\ref{ch:proofsystem}), the first program logic to compute a solution to a \semgus problem from a proof outline in realizability logic. 
The key idea of realization logic is to propagate failing synthesis attempts backwards, and eliminate unsuitable alternatives.
\item An algorithm to find proofs in realizability logic (Section~\ref{ch:vc}). 
We use deductive verification and develop appropriate strongest postconditions and verification conditions.
\item An algorithm to find program derivations in realization logic (Section~\ref{ch:adAlgo}).
Also here we rely on deductive verification, and we can give a backtracking freedom guarantee. 
\item We implemented our algorithms and applied them to synthesize the protection of memory accesses in a lock-free data structure by a safe memory reclamation algorithm (Section~\ref{sec:evaluation}). 
The synthesis works over an abstract semantics of protection types, and the problem cannot be handled by state-of-the-art synthesizers.
\end{itemize}

We start with a recapitulation of program logic, which forms the foundation of our work. 

%% file: sections-newgroundtruth-main.tex
\section{Preliminaries}
We introduce an imperative programming language that is parameterized in the data domain and in the set of commands, and give a Hoare logic for it that will form the basis of our development. 
\input{sections-newgroundtruth-preliminaries.tex}

\input{sections-newgroundtruth-hoare.tex}
\section{Realizability Logic}\label{ch:pls}
We extend programs to program sketches that contain yet to be resolved non-terminals, and present our program logic to reason about realizability. 
\input{sections-newgroundtruth-syntaxextension.tex}
\input{sections-newgroundtruth-realizabilitylogic.tex}

%% file: sections-newgroundtruth-preliminaries.tex
\subsection{Parameters}
The development in this paper is parameterized in a triple $(\states{}, \code{COM}, \progSemFuncOrig{-})$ which is given as part of the synthesis task. 
The set $\states{}$ is the data domain, the set of states programs operate over.
There are no requirements, the set may be finite or infinite. 
The set $\code{COM}$ contains the commands $\code{com}$ that can be used in programs. 
The function $\progSemFuncOrig{-}:\code{COM}\rightarrow \states{} \rightarrow \demonicAssertions{}$ assigns to each command a function from states to the set of predicates we define in a moment.

\subsection{Programs}
We use a classical while-language whose commands stem from the given set $\code{COM}$. Apart from that, we have sequential composition, choice, and Kleene star:
\begin{align*}
    \aprog\quad ::=\quad 
    \code{com} 
    \bnf \concatof{\aprog}{\aprog}
    \bnf \choiceOf{\aprog}{\aprog} 
    \bnf \kleeneof{\aprog}\ .
\end{align*}
We use $\demonicProgramLang{}$ for the set of all programs.  
We use $\straightLineProgramLang{}$ for the set of all executions $\anexec$, programs that neither contain choices nor loops.  
The function $\demonConcFunc{}$ yields all executions that can originate from a given program. 
For example, $\demonConcFuncOf{\concatof{(\choiceOf{\code{com}_1}{\code{com}_2})}{\code{com}}}$ is the set $\set{\concatof{\code{com}_1}{\code{com}}, \concatof{\code{com}_2}{\code{com}}}$.  
Loops are unrolled. 
Formally, we see the program as a regular expression and take the language. 

To give the program semantics, we first define $\demonicAssertions{} = \powersetOf{\states{}} \cup \set{\fail{}}$. 
Predicates are sets of states or a distinguished element $\fail{}$, and we typically use $\apred, \apredp, \apredpp$ for predicates. 
We have sets of states to support non-determinism, and failure to track program crashes, due to problems like segfaults or failing assertions. 
The predicates form a complete lattice with inclusion as the ordering and $\fail{}$ as the top element:  $\demonicAssertionStrongerOf{r}{s}$ is defined by $s = \fail{} \vee (r \neq \fail{} \wedge r \subseteq s)$, and we say that $r$ is \emph{more precise} than~$s$. 
We also write $\bigsqcup_d$ for the join. 
The index~$d$ reminds us that the non-determinism in programs is resolved demonically. 
Note that we formulated predicates in a semantic way, there is no assertion language to denote sets of states. 
Nevertheless, we emphasize that the inclusion in the formulation of $\demonicAssertionStrongerOf{r}{s}$ corresponds to implication in logic, meaning the definition can be implemented with solver technology right away. 

To give a semantics to executions, we first lift the semantics of commands from states to predicates. 
We
define $\progSemFuncOrigOf{\code{com}}{\fail{}} = \fail{}$ and $\progSemFuncOrigOf{\code{com}}{r} = \bigsqcup_d\setCond{\progSemFuncOrigOf{\code{com}}{s}}{s\in r}$, where $r\neq\fail{}$. 
We then set $\progSemFuncOf{\concatof{\anexec_1}{\anexec_2}}{r} = \progSemFuncOf{\anexec_2}{\progSemFuncOf{\anexec_1}{r}}$. 
A consequence is that programs cannot recover from failure.

%% file: sections-newgroundtruth-hoare.tex
\subsection{Hoare Logic}
We specify the correctness of programs by Hoare triples of the form $\hoareof{r}{\aprog}{s}$. 
The triple is valid, if every possible execution of $\aprog$ that starts in a state from $r$ ends in a state from $s$. 
We define 
    \begin{align*}
        \models_d \hoareof{r}{\aprog}{s} 
        \quad \text{by} \quad 
        \forall \anexec \in \demonConcFuncOf{\aprog}.\ \demonicAssertionStrongerOf{\progSemFuncOf{\code{ex}}{r}}{s}\ . 
    \end{align*}

To derive valid triples in a compositional way, Hoare logic offers proof rules similar to the black ones in Figure~\ref{def:pls}, we just have to replace all angle brackets by set brackets, the indices $a$ by $d$, and assume $R, S, T$ are predicates as we have defined them above. 
Rule~\ruleLabel{COM} requires that the postcondition over-approximates the effect of the command on the precondition. 
Rule~\ruleLabel{SEQ} composes programs sequentially if the intermediary assertion matches. 
Rule~\ruleLabel{LOOP} checks that the predicate is invariant under the loop body. 
Rule~\ruleLabel{DEM} is commonly called choice and lets the demon choose the alternative with which to continue the execution.
The rule of consequence \ruleLabel{CSQ} lets us strengthen the precondition and weaken the postcondition.
We use $\vdash_d \hoareof{r}{\aprog}{s}$ to indicate that a Hoare triple can be derived with these proof rules. 
It is readily checked that these rules are sound in that  $\vdash_d \hoareof{\apred}{\aprog}{\apredp}$  implies $\models_d \hoareof{\apred}{\aprog}{\apredp}$. 

One can also show that the rules are complete, $\models_d \hoareof{\apred}{\aprog}{\apredp}$ implies $\vdash_d \hoareof{\apred}{\aprog}{\apredp}$. 
As we will need it later, we explain the proof. 
The idea is to consider the full state space of program $\aprog$ from states in $\apred$. 
Now for each line of code, we collect all states from the state space that decorate this line, and let them form a predicate. 
We construct a proof outline in which every line of code carries the predicate we have just constructed as a precondition.
One can now check that this proof outline can be derived with the above rules. 

%% file: sections-newgroundtruth-syntaxextension.tex
\newcommand{\asketch}{\code{sketch}}
\newcommand{\asketchp}{\asketch'}
\newcommand{\anontermset}{\mathcal{N}}
\newcommand{\aproductionset}{\mathcal{P}}

\subsection{Sketches}
We extend the grammar of $\demonicProgramLang{}$ by non-terminals $\code{N}$ from a finite set $\anontermset$. 
The resulting \emph{sketches} are given by the grammar:
\begin{align*}
    \asketch\quad ::=\quad 
    \code{com} 
    \bnf \concatof{\asketch}{\asketch}
    \bnf \choiceOf{\asketch}{\asketch} 
    \bnf \kleeneof{\asketch}
    \bnf \anonterm 
    \ .
\end{align*}
We use $\angelicProgramLang{}$ for the set of all sketches. 
We assume each non-terminal $\code{N}\in\anontermset$ comes with a set of productions $\emptyset\neq \productionsof{\code{N}}\subseteq\angelicProgramLang{}$ and usually write $\code{N} ::=  \angelicChoiceOf{\code{sketch}_1}{\code{sketch}_2}$ rather than $\productionsof{\code{N}}=\set{\code{sketch}_1, \code{sketch}_2}$. 
There may be more than two alternatives, and the sketches may themselves contain non-terminals so that we have proper recursion. 
The pair  $(\anontermset, \productions)$ of non-terminals and their productions belong to the synthesis task and are given by the user. 

To lift the semantics from programs to sketches, we let the function $\angelConcFunc{}$ determine all programs (without non-terminals) that can be derived from a sketch by rewriting the non-terminals. 
For example, with $\code{N} ::=  \angelicChoiceOf{\code{com}_1}{\code{com}_2}$ we have $\angelConcFuncOf{\concatof{\code{N}}{\code{com}}}=\set{\concatof{\code{com}_1}{\code{com}}, \concatof{\code{com}_2}{\code{com}}}$. 
Formally, we see the sketch as a sentential form in a context-free grammar and take the language. 

%% file: sections-newgroundtruth-realizabilitylogic.tex
\subsection{Realizability Logic}
Realizability logic reasons over correctness specification of the form $\hoaretriplet{\apredset}{\asketch}{\apredsetp}$. 
Here, $\apredset, \apredsetp$ are so-called \emph{selections} from the set $\angelicAssertions{} = \powerset{}_{\mathit{fin}}(\demonicAssertions{})$. 
A selection is a finite set of predicates as we have defined them above.
The idea is this.  
When we encounter a non-terminal in a sketch, we do not yet determine how to resolve it to a program with the function $\angelConcFunc{}$. 
Instead, we collect in a selection the predicates that these programs may justify as a postcondition. 
Since also the precondition is a selection, we can not only choose the program but also the precondition to justify the postcondition.
We thus define validity of realizability triples via validity in Hoare logic:
    \begin{equation*}
        \angelicHoareTripletHoldsSemantically{\angelicAssertionFont{\apredset}}{\asketch}{\angelicAssertionFont{\apredsetp}} 
        \qquad \text{by}\qquad 
        \forall \apredp \in \angelicAssertionFont{\apredsetp}.\ 
        \exists \apred \in \angelicAssertionFont{\apredset}. \
        \exists \aprog \in \angelConcFuncOf{\asketch}. \ 
        \demonicHoareTripletHoldsSemantically{\apred}{\aprog}{\apredp} \ .
    \end{equation*}
\begin{figure}
    \begin{mathpar}
        \footnotesize
        \inferrule[\ruleLabelSmall{COM}]{
\demonicAssertionStrongerOf{\progSemFuncOf{\code{com}}{\apred}}{\apredp}
        }{
\angelicHoareTripletHolds{\set{\apred}}{\code{com}}{\set{\apredp}}
}
\and
\inferrule[\ruleLabelSmall{SEQ}]{
\angelicHoareTripletHolds{\apredset}{\asketch_1}{\apredsetp}
\\\\
\, \angelicHoareTripletHolds{\apredsetp}{\asketch_2}{\apredsetpp}
        }{
\angelicHoareTripletHolds{\apredset}{\concatof{\asketch_1}{\asketch_2}}{\apredsetpp}}
\and
\inferrule[\ruleLabelSmall{LOOP}]{
\makeBlue{\apredset = \set{\apred}}
\\\\
\angelicHoareTripletHolds{\apredset}{\asketch}{\apredset}
        }{
\angelicHoareTripletHolds{\apredset}{\kleeneof{\asketch}}{\apredset}
}
\and
\inferrule[\ruleLabelSmall{CSQ}]{
\angelicAssertionStrongerOf{\apredset}{\apredset'}
\\
\angelicAssertionStrongerOf{\apredsetp'}{\apredsetp}
\\\\
\angelicHoareTripletHolds{\apredset'}{\asketch}{\apredsetp'}
        }{
\angelicHoareTripletHolds{\apredset}{\asketch}{\apredsetp}
}
\and
\inferrule[\ruleLabelSmall{DEM}]{
\makeBlue{\apredset = \set{\apred}}
\\\\
\angelicHoareTripletHolds{\apredset}{\asketch_1}{\apredsetp}
\\\\
\angelicHoareTripletHolds{\apredset}{\asketch_2}{\apredsetp}
        }{
\angelicHoareTripletHolds{\apredset}{\choiceOf{\asketch_1}{\asketch_2}}{\apredsetp}
}
\and
\makeBlue{
\inferrule[\ruleLabelSmall{ANG}]{
\anonterm ::= \angelicChoiceOf{\asketch}{\dots}
\\\\
\angelicHoareTripletHolds{\apredset}{\asketch}{\apredsetp}
        }{
\angelicHoareTripletHolds{\apredset}{\anonterm}{\apredsetp}
}
}
\and
\makeBlue{
\inferrule[\ruleLabelSmall{GATHER}]{
\angelicHoareTripletHolds{\apredset_1}{\asketch}{\apredsetp_1}
\\\\
\angelicHoareTripletHolds{\apredset_2}{\asketch}{\apredsetp_2}
        }{
\angelicHoareTripletHolds{\apredset_1 \cup \apredset_2}{\asketch}{\apredsetp_1 \cup \apredsetp_2}
}
}
    \end{mathpar}
    \caption{Rules of realizability logic. Extensions of Hoare logic are highlighted in \makeBlue{blue}.}
    \label{def:pls}
\end{figure}

\begin{figure}
 \begin{minipage}[b]{0.45\textwidth}
\begin{lstlisting}[belowskip=0pt]
$\inHoareParenthesis{\mytrue}$  N(x = 0 $\color{black}\angelicChoice{}$ x = 1);  $\inHoareParenthesis{x = 0, x = 1}$
(assert(x==0) $\color{black}\choice{}$ assert(x==1))
$\inHoareParenthesis{\fail{}}$
\end{lstlisting}
    \caption{Demonic choice.}
    \label{code:exampleTwo}
\end{minipage}
\hfill
    \begin{minipage}[b]{0.45\textwidth}
        \centering
    \begin{mathpar}
        \footnotesize
\makeBlue{
\inferrule[\ruleLabelSmall{SHARE}]{
\angelicHoareTripletHolds{\angelicAssertionFont{T}}{\anonterm}{\angelicAssertionFont{U}}
\\
\angelicHoareTripletHolds{\apredset}{\asketch[\anonterm/\inHoareParenthesis{T}\inHoareParenthesis{U}]}{\apredsetp} }{
\angelicHoareTripletHolds{\apredset}{\asketch}{\apredsetp}
}
}
    \end{mathpar}
    \caption{Subproof sharing.}
    \label{fig:nonterminalsIntuition}
    \end{minipage}
\end{figure}

To derive valid realizability triples, we use the proof rules in \Cref{def:pls}. 
We write $\angelicHoareTripletHolds{\apredset}{\asketch}{\apredsetp}$ if the realizability triple can be derived with the help of these rules.
Rule~\ruleLabel{COM} checks whether the effect of a command $\code{com}$ on a predicate $\apred$ is captured by $\apredp$. 
The rule allows us to derive the realizability triple $\hoaretriplet{\set{\apred}}{\code{com}}{\set{\apredp}}$ in which the selections are singleton sets. 
The rule plays together with~\ruleLabel{GATHER}: if we can also derive $\hoaretriplet{\set{\apred'}}{\code{com}}{\set{\apredp'}}$ for another pre- and postcondition, then we can join the selections and obtain $\hoaretriplet{\set{\apred, \apred'}}{\code{com}}{\set{\apredp, \apredp'}}$.    
Rule~\ruleLabel{GATHER} is also important to join the pre- and postconditions for the various programs that can be derived from a non-terminal. 
To obtain a realizability triple $\hoaretriplet{\apredset}{\anonterm}{\apredsetp}$, Rule~\ruleLabel{ANG} unwinds the non-terminal to a sketch $\asketch$ that is given by a production and then proves $\hoaretriplet{\apredset}{\asketch}{\apredsetp}$. 
Note that the sketch may again contain $\anonterm$, and may contain more than one occurrence.  
We explain at the end of the section how to avoid the repetition of proof trees and handle non-terminals in an efficient way.   

\looseness=-1
Rule \ruleLabel{DEM} looks like the standard rule for choice operators in Hoare logic. 
In realizability logic, the precondition has to be a singleton selection,  and the rule is unsound without this side condition. 
If there were multiple predicates in $\apredset$, then the synthesis from $\asketch_1$
and the synthesis from $\asketch_2$ could use different predicates. 
This would be incorrect, because the synthesis is done at compile-time while the branch is chosen at runtime, and therefore the synthesis is not allowed to react on the branch.
To see this, consider the sketch in \Cref{code:exampleTwo} 
and let $\code{choice}$ stand for $\code{assert(x==0)} \choice{} \code{assert(x==1)}$. 
Without the side condition, we could derive the realizability triple $\angelicHoareTripletHolds{x=0, x=1}{\code{choice}}{\xvar=0, \xvar=1}$, because we could enter the left branch with the predicate $x=0$ and the right branch with the predicate $x=1$. 
However, the programs we can synthesize, namely $\concatof{\xzero}{\code{choice}}$ and $\concatof{\xone}{\code{choice}}$, have to commit to one of the assignments, and therefore  $\code{choice}$ is doomed to fail. 
A similar reasoning applies for Rule \ruleLabel{LOOP}. 
The way to handle multiple predicates in a precondition is to consider them one by one and join the results, again using \ruleLabel{GATHER}. 

Rule~\ruleLabel{SEQ} reinforces the definition of validity for realizability triples.   
Consider $\angelicHoareTripletHoldsSemantically{\apredset}{\asketch_1}{\apredsetp}$ and $\angelicHoareTripletHoldsSemantically{\apredsetp}{\asketch_2}{\apredsetpp}$. 
If validity was defined by an existential instead of a universal quantifier over the postcondition, then it would not be sound to derive $\angelicHoareTripletHoldsSemantically{\apredset}{\concatof{\asketch_1}{\asketch_2}}{\apredsetpp}$. 
The point is that the existentially quantified predicate $\apredp_1\in\apredsetp$ that is chosen to justify the validity of the first triple could differ from the predicate $\apredp_2\in\apredsetp$ needed for the second.
There is an alternative definition of validity that admits compositionality. 
We discuss it at the end of the following section when we have developed the overall synthesis technique.  

It remains to discuss \ruleLabel{CSQ}. 
The rule of consequence needs an ordering on selections. 
We say that selection $\apredset$ is \emph{more versatile} than selection $\apredsetp$, if for every predicate in $\apredsetp$ there is a more precise predicate in $\apredset$, so we define
\begin{align*}
\apredset\angelicAssertionStronger{}\apredsetp\qquad\text{by}\qquad \forall \apredp\in\apredsetp.\ \exists \apred\in\apredset.\ \demonicAssertionStrongerOf{\apred}{\apredp}\ .
\end{align*}
Having more predicates available makes a selection more versatile and having less predicates makes a selection less versatile. 
It is also possible to have a more versatile selection with less predicates if they are more precise.  
With this, $\emptyset$ is the least versatile and $\set{\emptyset}$ is the most versatile selection.    

The rules are sound and, together with a rule for the edge case of an empty precondition that we give in the appendix, also complete. 
\begin{theorem}[Sound-And-Complete]\label{th:angelicHoareSoundnessCompleteness}
$\angelicHoareTripletHolds{\apredset}{\asketch}{\apredsetp}$ if and only if  
$\angelicHoareTripletHoldsSemantically{\apredset}{\asketch}{\apredsetp}$\ .  
\end{theorem}
\begin{proof}
Soundness holds by an induction on the height of the proof tree.  
We argue for completeness and consider $\angelicHoareTripletHoldsSemantically{\apredset}{\asketch}{\apredsetp}$. 
By the definition of validity, this means for every predicate $\apredp\in\apredsetp$ there is a predicate $\apred\in\apredset$ and a program $\aprog\in\angelConcFuncOf{\asketch}$ so that $\demonicHoareTripletHoldsSemantically{\apred}{\aprog}{\apredp}$ holds. 
We invoke the completeness of Hoare logic and obtain $\demonicHoareTripletHolds{\apred}{\aprog}{\apredp}$. 
The derivation can be mimicked in realizability logic and yields $\angelicHoareTripletHolds{\apred}{\aprog}{\apredp}$. 
Since the program has been derived from $\asketch$, we can also obtain $\angelicHoareTripletHolds{\apred}{\asketch}{\apredp}$ with finitely many applications of \ruleLabel{ANG}. 
Rule \ruleLabel{GATHER} allows us to join the triples $\hoaretriplet{\apred}{\asketch}{\apredp}$ and obtain the desired $\angelicHoareTripletHolds{\apredset}{\asketch}{\apredsetp}$. 
\end{proof}
\subsection{Discussion}
\paragraph*{On \semgus}
The synthesis tasks we consider have the following input: 
the states, the commands, and the semantics $(\states{}, \code{COM}, \progSemFuncOrig{-})$ of the programming language, the non-terminals with their production rules $(\anontermset, \productions)$, and the realizability triple $\hoaretriplet{\apredset}{\asketch}{\apredsetp}$ of interest. 
In its original formulation~\cite{dantoni:semgus}, \semgus\ would be more liberal and allow the user to also define the operators (their syntax and their semantics) from which programs can be built.  
We fix those operators to concatenation, choice, and Kleene star, instead. 
This allows us to work with an extension of Hoare logic.
To lift our approach to the more general setting, the user would have to specify the proof rules that are sound for the new operators, which would be in-line with the approach in unrealizability logic~\cite{dantoni:UnrealizabilityLogic}. 

The reader may also note that we have not made assumptions about the correctness specification. 
In synthesis, it is common to work with sets of examples. 
The work on unrealizability logic has shown  that such sets can be captured by so-called vector assertions~\cite{dantoni:UnrealizabilityLogic}. 
Vector assertions are predicates of a particular form, and our construction of selections readily applies to them. 
\paragraph*{On search-based synthesis}
Realizability logic is formulated such that it can be readily combined with search-based synthesis. 
The point is that Rule~\ruleLabel{ANG} does not force us to consider all programs a non-terminal can be rewritten to. 
Instead, we can consider a set of programs that appear most promising, as can be judged from a probabilistic grammar~\cite{ProbGrammar18}. 
\paragraph*{On the implementation of non-terminals}
A non-terminal may occur multiple times in a proof in realizability logic, it may even occur recursively.  
Rather than duplicating the subproofs for this non-terminal, we would like to share them in the various places the non-terminal is used. 
This is made possible by Rule \ruleLabel{SHARE} given in \Cref{fig:nonterminalsIntuition}. 
If we have already derived the triple $\hoaretriplet{\angelicAssertionFont{T}}{\anonterm}{\angelicAssertionFont{U}}$, then we can eliminate the non-terminal from a program sketch and replace it by the pair $\inHoareParenthesis{\angelicAssertionFont{T}}\inHoareParenthesis{\angelicAssertionFont{U}}$.  
The pair stands for the fact that we can synthesize a program that transforms $\angelicAssertionFont{T}$ to $\angelicAssertionFont{U}$.  
Rule~\ruleLabel{SHARE} is readily checked to be sound. 
However, it is a derived rule whose application can easily be mimicked by Rule~\ruleLabel{ANG}, at the cost of blowing-up the proof. 
An implementation would refine the rule so that different occurrences of the non-terminal can be replaced in different ways. 
Moreover, one would maintain a pointer to the subproof that should be inserted, to be able to derive the program with the technique given next. 

%% file: sections-rewriteproofsystem-main.tex
\section{Realization Logic} \label{ch:proofsystem}
We define realization logic, a program logic to derive programs from sketches.
The key insight is that a proof outline in realizability logic for the sketch of interest is helpful guidance to find a program that solves the synthesis task. 
To make use of this guidance, realization logic rewrites the given proof outline until a suitable program is found. 
Realization logic thus reasons over rewriting steps of the form $\apo \vsim \apop$ between proof outlines. 
The strategy behind the rewriting is to propagate information about failed synthesis attempts backwards, and thereby iteratively eliminate predicates from selections and productions from the definition of non-terminals.

As we need a precise understanding of proof outlines, we give the definition:
\begin{align*}
    \apo\quad ::=\quad 
    \inHoareParenthesis{\apredset} \code{com} \inHoareParenthesis{\apredsetp} 
    \bnf \concatof{\apo}{\apo}
    \bnf \choiceOf{\apo}{\apo}
    \bnf \kleeneof{\apo}
    \bnf \anonterm(\apo)
    \bnf (\apo \angelicChoice{} \apo)
    \ .
\end{align*}
We write $\vdash_{a} \apo$ to indicate that the proof outline has been derived with  realizability logic. 
The assertions we track are the ones that have been used in the application of \ruleLabel{SEQ}, \ruleLabel{LOOP}, and \ruleLabel{ANG}.  
The formal definition is in the appendix. 
We abuse the notation a bit and write $\hoaretriplet{\apredset}{\apo}{\apredsetp}$ to indicate that $\apredset$ is the precondition and $\apredsetp$ is the postcondition of the proof outline $\apo$.
This is well-defined, because $\vdash_a \choiceOf{\apo_1}{\apo_2}$ implies that  $\apo_1$ and $\apo_2$ have the same pre- and postcondition.

\input{sections-rewriteproofsystem-rewritingcalculus.tex}

%% file: sections-rewriteproofsystem-rewritingcalculus.tex
\newcommand{\weakerof}[1]{#1^{\mathsf{w}}}
\newcommand{\wweakerof}[1]{#1^{\mathsf{w}\mathsf{w}}}
The proof rules of realization logic are listed in \Cref{def:adc}, we discuss them below. 
The rules give the following soundness guarantee: if we start from a proof outline in realizability logic and rewrite it to another proof outline, then also this proof outline can be derived with realizability logic. 
\begin{theorem}[Soundness]\label{th:decisionCalculusSoundenss}
$\vdash_a \apo$ and $\angelicProofRewriteHolds{\apo}{\apop}$ together imply 
$\vdash_a \apop$. 
\end{theorem}
\noindent This allows us to invoke the soundness of realizability logic from~\Cref{th:angelicHoareSoundnessCompleteness} and obtain the desired semantic guarantee for the program we have derived. 

A feature all rules share is that they are guided by the original proof outline: the selections in the proof outline that results from rewriting are limited to the selections in the original proof outline. 
To ease the notation, we use $\weakerof{\apredset}$ to denote a selection that is known to be less versatile than $\apredset$, meaning $\angelicAssertionStrongerOf{\apredset}{\weakerof{\apredset}}$ holds. 
This may be used repeatedly, so $\apredset\typeEnvStronger{}_{a}\weakerof{\apredset}\typeEnvStronger{}_{a}\wweakerof{\apredset}$. 

The rules in \Cref{def:adc} rewrite the given proof outline in a compositional fashion. 
The base case \ruleLabel{RCOM} allows us to pick a pre- and postcondition from a selection. 
This is the moment we realize that predicates from the precondition can be dropped because they are not needed to obtain the postcondition, as shown in the introduction. 
We have to explicitly check the over-approximation because realizability triples do not maintain the interplay between the predicates in the pre- and in the postcondition. 
Using \ruleLabel{RANG}, we can rewrite proofs for the right-hand sides of non-terminals. 
The heart of the calculus is \ruleLabel{RSELECT}, which allows us to replace non-terminals by their right-hand sides. 
The Rules~\ruleLabel{RSEQL} and \ruleLabel{RSEQR} rewrite the left resp. the right part of a sequence.
In \ruleLabel{RSEQL}, it is important that the new intermediary assertion matches the original intermediary assertion.
This does not have to be the case in \ruleLabel{RSEQR}, which uses the rule of consequence to automatically weaken the postcondition on the left.
As in realizability logic, Rule~\ruleLabel{RDEM} forces us to commit to a predicate when reasoning about demonic choices. 
It also produces the same postcondition~$\weakerof{S}$ in both branches. 
Rule \ruleLabel{RLOOP} for rewriting loop bodies has a similar behavior.  
The rule of consequence \ruleLabel{RCSQ} is used like in realizability logic.  
The same is true for  \ruleLabel{RGATHER}, which uses function $\gatherFunc{}$ to recursively join the intermediary assertions in the proof outlines.
The definition is in the appendix. 
The need for this function stems from the fact that \ruleLabel{RDEM} can only work with single predicates.  
Rule \ruleLabel{TRANS} is the transitivity of rewriting. 
This rule is needed to combine rewrites, for example when rewriting the left and the right part of a sequence.
\begin{figure}
\begin{mathpar}
\footnotesize
\stackMath
\inferrule[\ruleLabelSmall{RCOM}]{
\demonicAssertionStrongerOf{\progSemFuncOf{\code{com}}{\apred}}{\apredp}
\\
\angelicAssertionStrongerOf{\apredset}{\set{\apred}}
\\
\angelicAssertionStrongerOf{\apredsetp}{\set{\apredp}}
		}{
\angelicProofRewriteHolds{\hoaretriplet{\apredset}{\code{com}}{\apredsetp}}
{\hoaretriplet{\set{\apred}}{\code{com}}{\set{\apredp}}}
}
\and
\inferrule[\ruleLabelSmall{RANG}]{
\angelicProofRewriteHolds{\apo}{\apop}
        }{
        \angelicProofRewriteHolds{\anonterm(\apo_1\bnf\apo\bnf\apo_2)}{\anonterm(\apo_1\bnf\apop\bnf\apo_2)}
}
\and
\inferrule[\ruleLabelSmall{RSELECT}]{
\phantom{i}
        }{
\angelicProofRewriteHolds{\anonterm(\apo_1\bnf\apo\bnf\apo_2)}{\apo}
}
\and
\inferrule[\ruleLabelSmall{RSEQL}]{
\angelicProofRewriteHolds{\hoaretriplet{\apredset}{\apo_1}{\apredsetp}}{\hoaretriplet{\weakerof{\apredset}}{\apop_1}{\apredsetp}}
		}{
\angelicProofRewriteHolds{\concatof{\hoaretriplet{\apredset}{\apo_1}{\apredsetp}}{\apo_2}}{\concatof{\hoaretriplet{\weakerof{\apredset}}{\apop_1}{\apredsetp}}}{\apo_2}
}
\and
\inferrule[\ruleLabelSmall{RSEQR}]{
\angelicProofRewriteHolds{\hoaretriplet{\apredsetp}{\apo_2}{\apredsetpp}}{\hoaretriplet{\weakerof{\apredsetp}}{\apo_2'}{\weakerof{\apredsetpp}}}
		}{
\angelicProofRewriteHolds{
\concatof{\hoaretriplet{\apredset}{\apo_1}{\apredsetp}}
{\hoaretriplet{\apredsetp}{\apo_2}{\apredsetpp}}
}{\concatof{\hoaretriplet{\apredset}{\apo_1}{\weakerof{\apredsetp}}}
{\hoaretriplet{\weakerof{\apredsetp}}{\apo_2'}{\weakerof{\apredsetpp}}}}
}
\and
\inferrule[\ruleLabelSmall{RDEM}]{
\weakerof{\apredset} = \set{\apred} 
\\
\angelicProofRewriteHolds{
\hoaretriplet{\apredset}{\apo_i}{\apredsetp}}{
\hoaretriplet{\weakerof{\apredset}}{\apo_i'}{\weakerof{\apredsetp}}}
\quad\text{for }i=1, 2
        }{
\angelicProofRewriteHolds{
\choiceOf{\hoaretriplet{\apredset}{\apo_1}{\apredsetp}}
{\hoaretriplet{\apredset}{\apo_2}{\apredsetp}}
}{\choiceOf{\hoaretriplet{\weakerof{\apredset}}{\apo_1'}{\weakerof{\apredsetp}}}
{\hoaretriplet{\weakerof{\apredset}}{\apo_2'}{\weakerof{\apredsetp}}}}
}
\and
\inferrule[\ruleLabelSmall{RLOOP}]{
\angelicProofRewriteHolds{\hoaretriplet{I}{\apo}{I}}{\hoaretriplet{\weakerof{I}}{\apop}{\weakerof{I}}}
\\
\weakerof{I} = \set{i}
        }{
\angelicProofRewriteHolds{
        \hoaretriplet{I}{\kleeneof{\apo}}{I}
}{
        \hoaretriplet{\weakerof{I}}{\kleeneof{\apo}}{\weakerof{I}}
}
}
\and
\inferrule[\ruleLabelSmall{RTRANS}]{
\angelicProofRewriteHolds{\apo_1}{\apo_2}
\\\\
\angelicProofRewriteHolds{\apo_2}{\apo_3}
        }{
\angelicProofRewriteHolds{\apo_1}{\apo_3}
}
\and
\inferrule[\ruleLabelSmall{RCSQ}]{
\angelicProofRewriteHolds{\hoaretriplet{\apredset}{\apo}{\apredsetp}}{\hoaretriplet{\wweakerof{\apredset}}{\apop}{\weakerof{\apredsetp}}}
        }{
\angelicProofRewriteHolds{\hoaretriplet{\apredset}{\apo}{\apredsetp}}{\hoaretriplet{\weakerof{\apredset}}{\apop}{\wweakerof{\apredsetp}}}
}
\and
\inferrule[\ruleLabelSmall{RGATHER}]{
\angelicAssertionStrongerOf{\apredset}{(\apredset_1 \cup \apredset_2)}
\\
\angelicAssertionStrongerOf{\apredsetp}{(\apredsetp_1 \cup \apredsetp_2)}
\\
\proofToProgFuncOf{\apop_1} = \proofToProgFuncOf{\apop_2}
\\\\
\angelicProofRewriteHolds{\hoaretriplet{\apredset}{\apo}{\apredsetp}}{\hoaretriplet{\apredset_1}{\apop_1}{\apredsetp_1}}
\\
\angelicProofRewriteHolds{\hoaretriplet{\apredset}{\apo}{\apredsetp}}{\hoaretriplet{\apredset_2}{\apop_2}{\apredsetp_2}}
		}{
\angelicProofRewriteHolds{\hoaretriplet{\apredset}{\apo}{\apredsetp}}
{\hoaretriplet{\apredset_1 \cup \apredset_2}{\gatherFuncOf{\apop_1}{\apop_2}}{\apredsetp_1 \cup \apredsetp_2}}
}
\end{mathpar}
        \caption{Rules of realization logic.
        }
        \label{def:adc}
\end{figure}

Realization logic always rewrites a proof outline into one that is weaker in the sense that it has less synthesis options or, phrased differently, is closer to a program. 
\begin{lemma} \label{lem:rewriteWeaker}
        $\angelicProofRewriteHolds{\apo}{\apop}$ 
        implies
        $\proofStrongerOf{\apo}{\apop}$\ .
\end{lemma} 
\noindent The ordering $\proofStrongerOf{\apo}{\apop}$ states that $\apop$ has less versatile assertions than $\apo$ and a sketch or even program that can be derived from the sketch in $\apo$.  
The definition is by induction on the structure of proof outlines.
The base case is $\proofStrongerOf{\hoaretriplet{\apredset}{\code{com}}{\apredsetp}}{\hoaretriplet{\weakerof{\apredset}}{\code{com}}{\weakerof{\apredsetp}}}$. 
The step cases are similar to the one for choice:
$\proofStrongerOf{
        \choiceOf{\apo_1}{\apo_2}
}{\choiceOf{\apop_1}{\apop_2}}
$ 
if $\proofStrongerOf{\apo_1}{\apop_1}$ and $\proofStrongerOf{\apo_2}{\apop_2}$.
Non-terminals are an exception, where we use \(
        \proofStrongerOf{
                \anonterm(\apo_1\bnf\apo\bnf\apo_2)}
                {
                \anonterm(\apo_1\bnf\apop\bnf\apo_2)} 
\) if \(
        \proofStrongerOf{\apo}{\apop}\,
\).
A proof outline over a non-terminal is also stronger than a proof outline over a derivative, $\proofStrongerOf{\anonterm(\apo_1\bnf\apo\bnf\apo_2)}{\apo}$.

Realization logic also has a completeness property: 
proof outlines can be rewritten into all weaker proof outlines that are derivable in realizability logic. 
\begin{theorem}[Completeness]\label{th:angelicDecisionCalcComplete}
$\vdash_a \apo$ and $\vdash_a \apop$ and $\proofStrongerOf{\apo}{\apop}$ together imply $\angelicProofRewriteHolds{\apo}{\apop}$. 
\end{theorem}
A final guarantee that is interesting from an algorithmic point of view is that realization logic provably does not need backtracking: from every proof outline and for every predicate in the postcondition, one can derive a program that satisfies this postcondition.
By \Cref{th:decisionCalculusSoundenss}, the guarantee continues to hold after a rewriting step (that should not remove the predicate of interest), hence the name. 
\begin{theorem}[Backtracking Freedom]
    \label{th:angelicChoiceElimination}
    Let $\angelicHoareTripletHolds{\apredset}{\apo}{\apredsetp}$
    and $\apredp \in \apredsetp$. 
    Then there are $\apred$ and $\apop$ so that 
    $\angelicProofRewriteHolds{
        \hoaretriplet{\apredset}{\apo}{\apredsetp}
    }{
        \hoaretriplet{\set{\apred}}{\apop}{\set{\apredp}}
    }$, 
    $\apred \in \apredset$, 
    and $\proofToProgFuncOf{\apop} \in \angelConcFuncOf{\proofToProgFuncOf{\apo}}$.
\end{theorem}

It is tempting to try to prove backtracking freedom with the help of completeness in \Cref{th:angelicHoareSoundnessCompleteness}. 
To understand why this does not work, consider the valid proof outline given in \Cref{code:unsynthExample} and suppose both commands have the same semantics, namely the identity function. 
When we want to derive a program for the postcondition $\set{\xvar = 0}$, realization logic can only propose~$\code{skip}_1$. 
The notion of validity in realizability logic, however, may justify the postcondition by $\demonicHoareTripletHoldsSemantically{\xvar = 0}{\code{skip}_2}{\xvar = 0}$.
Completeness in  \Cref{th:angelicHoareSoundnessCompleteness} now shows that there is a proof outline for $\angelicHoareTripletHolds{\xvar=0}{\code{skip}_2}{\xvar=0}$. 
The point is that this proof outline is not weaker than the one in \Cref{code:unsynthExample}. 
This also means that there is no contradiction to \Cref{th:angelicDecisionCalcComplete}. 
In short, \Cref{th:angelicHoareSoundnessCompleteness} reasons about sketches and \Cref{th:angelicChoiceElimination} reasons about proof outlines, and the two are incomparable.

\begin{wrapfigure}[7]{O}{0.32\textwidth}
\vspace{-5mm}
\begin{lstlisting}
$\inHoareParenthesis{x = 0}$
N(($\inHoareParenthesis{x = 0}$ skip$\color{black}{}_1$ $\inHoareParenthesis{x = 0}$) $\color{black}\angelicChoice{}$
  ($\inHoareParenthesis{x = 0}$ skip$\color{black}{}_2$ $\inHoareParenthesis{\emptyset}$))
$\inHoareParenthesis{x = 0}$
\end{lstlisting}
    \caption{Guidance.}
    \label{code:unsynthExample}
\end{wrapfigure}

For an example application of realization logic, we refer to \Cref{Section:Introduction}.
Eliminating the predicate $\fail{}$ in \Cref{eq:ruleApplicationCsq} is an application of Rule \ruleLabel{RCSQ}.
Weakening the precondition in \Cref{eq:ruleApplicationBase} is an application of Rule \ruleLabel{RCOM}.
Resolving the non-terminal to the left production is an application of rule \ruleLabel{RSELECT}.
The resulting proof outline fragments are put together using Rules \ruleLabel{RSEQL}, \ruleLabel{RSEQR}, and \ruleLabel{RTRANS}.

\paragraph*{Overall approach and the notion of validity}
The overall approach to synthesizing a program from a sketch $\asketch$, a precondition~$\apredset$, and postcondition~$\apredsetp$ is this. 
We prove the triple $\hoaretriplet{\apredset}{\asketch}{\apredsetp}$ in realizability logic. 
One may see this proof as a symbolic execution that annotates the code by assertions in a \emph{forward} fashion, starting from the precondition and ending at the postcondition.  
When we have proven the triple, we in particular have a proof outline at hand. 
We rewrite this proof outline to a program using the realization logic we have just developed. 
The rewriting proceeds \emph{backwards}: starting from failures and predicates that do not guarantee the postcondition, we eliminate predicates from selections and productions from the definition of non-terminals.

We need the rewriting phase because realizability logic abstracts away the link between the programs that can be synthesized from the sketch and the pre- and postconditions that these programs can guarantee. 
This abstraction is precisely what makes realizability logic scale, and the link is precisely what realization logic recovers. 
What makes realization logic efficient is that it is guided by the proof outline in realizability logic.
Moreover, realizability logic provably does not need backtracking (the rewriting cannot go wrong).  

It would be possible to define the notion of validity in realizability logic differently, namely by universally quantifying over the precondition and existentially quantifying over the postcondition. 
Our synthesis approach can be adapted to this definition as follows. 
We imagine the proof in realizability logic as being constructed \emph{backwards}, the symbolic execution starts from the postcondition and proceeds to the precondition.
Realization logic would then start from unsuitable preconditions and proceed \emph{forwards}. 
Our approach has the advantage that forward reasoning tends to be more deterministic than backwards reasoning, which leads to smaller assertions in realizability logic. 

%% file: sections-newgroundtruthalgorithms-main.tex
\section{Implementing Realizability Logic} \label{ch:vc}
We give an algorithm that checks the validity of triples in realizability logic. 
The algorithm also yields a proof outline that will be handed over to the implementation of realization logic developed in the next section. 
The algorithm is a deductive verification: we compute a set of verification conditions whose validity entails the validity of the given realizability triple. 
Importantly, although the assertions in realizability logic are selections rather than mere predicates, the validity check can be discharged with an off-the-shelf SMT solver. 
At the technical level, our contribution lies in the definition of suitable verification conditions. 
We proceed forwards, using a tailor-made strongest post. 
Deductive verification needs loop invariants, and this is no different in our case.
We assume these invariants are given by the user. 
Moreover, we assume the user gives us information about the non-terminals, either in the form of an annotated assertion transformer, or in the form of a recursion depth to which the non-terminal should be unwound in the search for a program. 
Taking these user annotations as an input is a compromise: it is a burden on the user, but it improves the applicability of the method.
Having said that, we remark that in our experiments we were able to automatically determine the loop invariants and handle the non-terminals.
\input{sections-newgroundtruthalgorithms-verificationconditions.tex}

\input{sections-newgroundtruthalgorithms-example.tex}

%% file: sections-newgroundtruthalgorithms-verificationconditions.tex
\newcommand{\checkAssertionFunc}{\mathit{ca}}
\newcommand{\checkAssertionFuncOfL}[2]{\checkAssertionFunc(#1, #2)}
\newcommand{\checkAssertionFuncOfN}[3]{\checkAssertionFunc(#1, #2, #3)}
\newcommand{\oracleFunc}{\mathbb{O}}
\newcommand{\oracleFuncOf}[2]{\oracleFunc(#1,#2)}
\newcommand{\ashortsketch}{\code{sk}}
\newcommand{\ashortprog}{\code{p}}
\newcommand{\ashortprogp}{\code{p}'}
\subsection{Verification Conditions and Strongest Post}

In \Cref{def:vc}, we define the function  $\funcDef{\verificationConditionsFunc{}}{\vcAngelicHoareTripletSet{} \rightarrow \powersetOf{\assertionSet{}}}$ which takes an annotated realizability triple and produces a set of verification conditions. 
The annotation we assume is an invariant $\angelicAssertionFont{I}$ in the case of loops $\kleeneof{\asketch}$, denoted by $\kleeneof{\asketch}[\angelicAssertionFont{I}]$, and a selection transformer $\Gamma:\angelicAssertions{}\rightarrow\angelicAssertions{}$ in the case of non-terminals~$\anonterm$, written as~$\anonterm[\Gamma]$. 
A verification condition is an inequality between selections from the set $\assertionSet{} = \angelicAssertions{} \times \angelicAssertions{}$.

We rely on the correctness of the annotations. 
A loop invariant is sound, if $\angelicHoareTripletHoldsSemantically{\set{i}}{\asketch}{\set{i}}$ holds for all $i\in I$. 
It is complete if $\angelicHoareTripletHoldsSemantically{\set{j}}{\asketch}{\set{j}}$ implies $j\in\angelicAssertionFont{I}$. 
A selection transformer is sound, if $\Gamma(\apredset)=\apredsetp$ implies $\angelicHoareTripletHoldsSemantically{\angelicAssertionFont{\apredset}}{\anonterm}{\apredsetp}$. 
It is complete, if $\angelicHoareTripletHoldsSemantically{\angelicAssertionFont{\apredset}}{\anonterm}{\apredsetp}$ implies $\angelicAssertionStrongerOf{\Gamma(\apredset)}{\apredsetp}$.

Function $\verificationConditionsFunc{}$ is sound: when all verification conditions hold, then the realizability triple is valid. 
Note that we do not need to assume the soundness of the annotations, but will check this as part of the verification conditions. 
For completeness, however, we need to make this assumption. 
\begin{theorem}[VC-Sound-And-Complete]\label{th:vcSoundness}
$\models \verificationConditionsFuncOf{\hoaretriplet{\apredset}{\asketch}{\apredsetp}}$ implies $\angelicHoareTripletHoldsSemantically{\apredset}{\asketch}{\apredsetp}$.
If all annotations are complete, $\angelicHoareTripletHoldsSemantically{\apredset}{\asketch}{\apredsetp}$ implies $\models \verificationConditionsFuncOf{\hoaretriplet{\apredset}{\asketch}{\apredsetp}}$.
\end{theorem}
The benefit of using verification conditions $\angelicAssertionStrongerOf{\apredset}{\apredsetp}$ as the proof obligation is that they can be discharged with off-the-shelf SMT technology. 
For every predicate $\apredp\in \apredsetp$, we have to go through the predicates $\apred\in\apredset$ until we find $\apred\subseteq\apredp$. 
This means at most quadratically many SMT queries. 
Future assertion languages for selections may suggest a different strategy. 

Function $\verificationConditionsFunc{}$ makes use of the strongest post function $\strongestPostSem{}$ which is also defined in \Cref{def:vc}. 
It takes as input a selection and a sketch and outputs a selection, 
$\funcDef{\strongestPostSem{}}{(\angelicAssertions{} \times \angelicProgramLang{}) \rightarrow \angelicAssertions{}}$. 
The output is a selection that can be guaranteed when running the sketch on the input. 
Moreover, it is the most versatile selection, if the annotations are complete.
Since the result does not refer to the verification conditions, we have to require soundness. 
\begin{theorem}[SP-Sound-And-Complete]\label{th:spComplete}
If the annotations are sound, $\angelicAssertionStrongerOf{\strongestPostSemOf{\apredset}{\asketch}}{\apredsetp}$ implies $\angelicHoareTripletHoldsSemantically{\apredset}{\code{\asketch}}{\apredsetp}$. 
If the annotations are sound and complete, $\angelicHoareTripletHoldsSemantically{\apredset}{\code{\asketch}}{\apredsetp}$ implies $\angelicAssertionStrongerOf{\strongestPostSemOf{\apredset}{\asketch}}{\apredsetp}$. 
\end{theorem}

    \begin{figure}
        \footnotesize
    \begin{align*}
        \strongestPostSemOf{\apredset}{\code{com}} &\definingEquals 
        \setCond{\progSemFuncOf{\code{com}}{\apred}}{\apred \in \apredset} 
        &
        \strongestPostSemOf{\apredset}{\code{\anonterm}[\Gamma]} &\definingEquals 
        \Gamma(\apredset)
        \\
        \strongestPostSemOf{\apredset}{\concatof{\asketch_1}{\asketch_2}} & \definingEquals 
        \strongestPostSemOf{\strongestPostSemOf{\apredset}{\asketch_1}}{\asketch_2}
        &
        \vcStrongestPostSemOf{\apredset}{\kleeneof{\asketch}[\angelicAssertionFont{I}]} & \definingEquals
        \setCond{i \in \angelicAssertionFont{I}}{\angelicAssertionStrongerOf{\angelicAssertionFont{\apredset}}{\set{i}}}
    \end{align*}
    \vspace{-0.5cm}
    \begin{align*}
       \strongestPostSemOf{\apredset}{\choiceOf{\asketch_1}{\asketch_2}} & \definingEquals 
       \setCond{\apredp_1 \join{}_d \apredp_2}{ \apredp_k \in \strongestPostSemOf{\set{\apred}}{\asketch_k} 
       \wedge 
       \apred \in \apredset
       \wedge k \in \set{1,2}
       }
    \end{align*}
    \begin{align*}
        \verificationConditionsFuncOf{
            \hoaretriplet{\apredset}{\code{com}}{\apredsetp}} 
             \definingEquals& \, \set{\angelicAssertionStrongerOf{\vcStrongestPostSemOf{\apredset}{\code{com}}}{\apredsetp}}
            &
            \verificationConditionsFuncOf{
            \hoaretriplet{\apredset}{\anonterm[\Gamma]}{\apredsetp}} 
             \definingEquals
            & 
            \, 
            \set{\angelicAssertionStrongerOf{\vcStrongestPostSemOf{\apredset}{\anonterm[\Gamma]}}{\apredsetp}} 
            \cup
            \checkAssertionFuncOfN{\apredset}{\anonterm}{\Gamma}
    \end{align*}
    \vspace{-0.5cm}
    \begin{align*}
        \verificationConditionsFuncOf{
            \hoaretriplet{\apredset}{\asketch_1\code{;}\asketch_2}{\angelicAssertionFont{S}}} 
             \definingEquals& \, \verificationConditionsFuncOf{            \hoaretriplet{\angelicAssertionFont{\apredset}}{\asketch_1}{\vcStrongestPostSemOf{\apredset}{\asketch_1}}} \cup
            \verificationConditionsFuncOf{\hoaretriplet{\vcStrongestPostSemOf{\apredset}{\asketch_1}}{\asketch_2}{\apredsetp}}
            \\
    \verificationConditionsFuncOf{
            \hoaretriplet{\apredset}{\kleeneof{\asketch}[\angelicAssertionFont{I}]}{\apredsetp}} 
             \definingEquals            
            & \,
            \set{\angelicAssertionStrongerOf{\vcStrongestPostSemOf{\apredset}{\kleeneof{\asketch}[\angelicAssertionFont{I}]}}{\angelicAssertionFont{S}}}
            \cup
            \checkAssertionFuncOfL{\angelicAssertionFont{I}}{\kleeneof{\asketch}}
            \\
        \verificationConditionsFuncOf{
            \hoaretriplet{\apredset}{\choiceOf{\asketch_1}{\asketch_2}}{\apredsetp}} 
             \definingEquals& \, 
            (
            \cup
            \setCond{
            \verificationConditionsFuncOf{            \hoaretriplet{\set{\apred}}{\asketch_k}{\vcStrongestPostSemOf{\set{\apred}}{\asketch_k}}}              
            }{
             \apred \in \apredset               
             \wedge
             k \in \set{1, 2}
            }
            )
            \cup
            \\ 
            & \, 
            \set{\angelicAssertionStrongerOf{\vcStrongestPostSemOf{\apredset}{\choiceOf{\asketch_1}{\asketch_2}}}{\apredsetp}}
    \end{align*}
    \begin{align*}
        \checkAssertionFuncOfN{\apredset}{\anonterm}{\Gamma}  \definingEquals
        & \,
        \cup
        \setCond{\verificationConditionsFuncOf{\hoaretriplet{\apredset}{\aprog}{\vcStrongestPostSemOf{\apredset}{\aprog}}}}
        {\min j \dv 
            \angelicAssertionStrongerOf{(\cup \setCond{\vcStrongestPostSemOf{\apredset}{\aprogp}}{\aprogp \in \oracleFuncOf{\anonterm}{j}})}{\Gamma(\apredset)}
            \wedge
            \aprog \in \oracleFuncOf{\anonterm}{j}
        }
    \end{align*}
    \vspace{-0.5cm}
    \begin{align*}
        \checkAssertionFuncOfL{\angelicAssertionFont{I}}{\kleeneof{\asketch}}  \definingEquals
        & \,
        \setCond{\verificationConditionsFuncOf{\hoaretriplet{\set{i}}{\asketch}{\set{i}}}}{i \in \angelicAssertionFont{I}}        
    \end{align*}
        \caption{Definitions of strongest post, verification conditions, and check annotation functions.}
        \label{def:vc}
\end{figure}

\looseness=-1
The functions $\verificationConditionsFunc{}$ and $\strongestPostSem{}$ are defined inductively over the structure of sketches.
The novelty is in the cases for demonic choices, non-terminals, and loops.
For demonic choices, recall that Rule~\ruleLabel{DEM} reasons over single predicates.
The strongest post thus iterates through the predicates in the given precondition. 
For each $\apred\in\apredset$ , it computes the strongest post for both branches.
If we have $\apredp_1\in\strongestPostSemOf{\set{\apred}}{\asketch_1}$, then we can guarantee the predicates $\apredp_1$ when running $\asketch_1$ on $\apred$.
The same holds for $\apredp_2\in \strongestPostSemOf{\set{\apred}}{\asketch_2}$.
But as the branch will be chosen demonically, we can only enforce the least upper bound of $\apredp_1$ and $\apredp_2$.
However, we are free to combine all $\apredp_1$ with all $\apredp_2$. 
This explains the definition. 
The verification conditions check each branch individually, again with a singleton precondition.
Moreover, they check whether the strongest post is sufficient to prove the postcondition.

For non-terminals, the strongest post is given by the annotated transformer. 
The verification conditions check whether the strongest post entails the postcondition.
Moreover, they check the soundness of the annotation with the function $\checkAssertionFuncOfN{\apredset}{\anonterm}{\Gamma}$. 
We elaborate on it in a moment.
For loops, we rely on an annotated invariant. 
The strongest post keeps only those predicates from the invariant that follow from the precondition. 
The verification conditions again check whether the strongest post entails the postcondition.
Moreover, they check whether the invariant annotation is sound via $\checkAssertionFuncOfL{\angelicAssertionFont{I}}{\kleeneof{\asketch}}$.

Soundness of a loop invariant is easy to check. 
Since Rule \ruleLabel{LOOP} can only enter a loop with a singleton, we check each predicate in $\angelicAssertionFont{I}$ for being an invariant. 
The soundness of transformers is more difficult to check. 
Function $\checkAssertionFuncOfN{\apredset}{\anonterm}{\Gamma}$ tries to derive a set of programs from $\anonterm$ which justifies the postcondition $\Gamma(\apredset)$.
To this end, it repeatedly invokes an oracle $\oracleFunc$ which returns programs from $\angelConcFuncOf{N}$.
The function $\oracleFuncOf{\anonterm}{j}$ returns the first $j$ programs the oracle produces for the non-terminal $\anonterm$.
If the set of programs that have been returned so far is able to produce the postcondition, 
 $\checkAssertionFuncOfN{\apredset}{\anonterm}{\Gamma}$ collects the corresponding verification conditions and terminates. 
The only guarantee we need about the oracle is that every program will eventually be returned.
If the transformer is sound, we will eventually obtain a large enough set of programs. 
Otherwise, the verification condition generation will not terminate.
A simple implementation for an oracle is to return the programs that can be obtained by unrolling each non-terminal at most $k$ times, for larger and larger $k$.

%% file: sections-newgroundtruthalgorithms-example.tex
\subsection{Example: Checking Realizability of the Factorial Function}
We illustrate the verification condition computation on an example.
Consider the sketch depicted in \Cref{code:exampleVCFactorial}.
We check whether we can derive from it a program that computes the factorial of~42:  
the precondition of interest is  $\apredset=\set{\xvar=42 \wedge \yvar = 1}$ and the postcondition is $\apredsetp=\set{\yvar = 42!}$. 
The sketch is already annotated by the results of the strongest post in grey, the loop is annotated by the invariant $I=\set{\yvar * \xvar! = 42! \wedge \xvar \geq 0}$, and the non-terminal $\code{\anonterm ::= y=y+x} \mid \code{y=y*x}$ is annotated by the  transformer $\Gamma$ with
$\Gamma(\apredset) = \strongestPostSemOf{\apredset}{\code{y=y+x}} \cup \strongestPostSemOf{\apredset}{\code{y=y*x}}$. 
We refer to the sketch as $\asketch$, the loop body as $\code{body}$, and the commands by their line number, so $\code{x-{}-}$ is $\code{17}$.
We abbreviate the non-terminal including Lines~\ref{line:fac:nontermStart} to \ref{line:fac:nontermEnd} by $\anonterm$.

We explain the computation of the verification conditions, which returns
\begin{align}
  \verificationConditionsFuncOf{\hoaretriplet{\apredset}{\asketch}{\apredsetp}} =
\{
  &\angelicAssertionStrongerOf{\vcStrongestPostSemOf{I}{\code{22}}}{\apredsetp},
    &\angelicAssertionStrongerOf{I}{I} \ \  (I =  \vcStrongestPostSemOf{\apredset}{\code{body}^*[I]}), \label{eq:vcExample2}
    \\
    &\angelicAssertionStrongerOf{\vcStrongestPostSemOf{\set{i}}{\code{4;\anonterm;17}}}{\set{i}},
    &\angelicAssertionStrongerOf{\vcStrongestPostSemOf{\set{i}}{\code{4}}}{\vcStrongestPostSemOf{\set{i}}{\code{4}}}, \label{eq:vcExample4}
    \\
    &\angelicAssertionStrongerOf{\vcStrongestPostSemOf{\set{i}}{\code{4;\anonterm}}}{\vcStrongestPostSemOf{\set{i}}{\code{4;\anonterm}}},
    &\angelicAssertionStrongerOf{\vcStrongestPostSemOf{\set{i}}{\code{4;8}}}{\vcStrongestPostSemOf{\set{i}}{\code{4;8}}}, \label{eq:vcExample6}
    \\
    &\angelicAssertionStrongerOf{\vcStrongestPostSemOf{\set{i}}{\code{4;12}}}{\vcStrongestPostSemOf{\set{i}}{\code{4;12}}}\label{eq:vcExample7}
\}
\ .
\end{align}
Executing $\verificationConditionsFuncOf{\hoaretriplet{\apredset}{\asketch}{\apredsetp}}$ yields two recursive calls: 
\begin{equation*}
  \verificationConditionsFuncOf{\hoaretriplet{\apredset}{\asketch}{\apredsetp}} = 
  \verificationConditionsFuncOf{\hoaretriplet{\apredset}{\kleeneof{\code{body}}[I]}{\vcStrongestPostSemOf{\apredset}{\kleeneof{\code{body}}[I]}}}
\cup
  \verificationConditionsFuncOf{\hoaretriplet{\vcStrongestPostSemOf{\apredset}{\kleeneof{\code{body}}[I]}}{\code{22}}{\apredsetp}}
\ .
\end{equation*}
To calculate the strongest post of the loop, $\apredset$ and each individual element of the loop invariant have to be compared.
Since $I=\set{i}$ is a singleton and $\angelicAssertionStrongerOf{\apredset}{\set{i}}$, we have that $I$ itself is the strongest post of the loop.
Therefore, the second recursive call boils down to checking 
$\angelicAssertionStrongerOf{\vcStrongestPostSemOf{I}{\code{assume x = 0}}}{\apredsetp}$, Inequality~\eqref{eq:vcExample2}(left).
Next, we invoke
$\verificationConditionsFuncOf{\hoaretriplet{\apredset}{\kleeneof{\code{body}}[I]}{\vcStrongestPostSemOf{\apredset}{\code{body}^*[I]}}}$.
We already argued that the strongest post of the loop is the loop invariant. 
Thus, $\angelicAssertionStrongerOf{I}{I}$ is added as Inequality~\eqref{eq:vcExample2}(right). 
Checking the loop invariant annotation with $\checkAssertionFuncOfL{I}{\kleeneof{\code{body}}}$
returns $\verificationConditionsFuncOf{\hoaretriplet{\set{i}}{\code{body}}{\set{i}}}$.
This, in turn, yields two calls to the verification conditions function:
\begin{equation*}
  \verificationConditionsFuncOf{\hoaretriplet{\set{i}}{\code{body}}{\set{i}}} = 
    \verificationConditionsFuncOf{\hoaretriplet{i}{\code{4;\anonterm}}{\vcStrongestPostSemOf{\set{i}}{\code{4;\anonterm}}}}
    \cup
    \verificationConditionsFuncOf{\hoaretriplet{\vcStrongestPostSemOf{\set{i}}{\code{4;\anonterm}}}{\code{17}}{i}}
    \ .
\end{equation*}
For the second recursive call, the base case of function $\verificationConditionsFuncOf{-}$ applies and adds $
\angelicAssertionStrongerOf{\vcStrongestPostSemOf{\set{i}}{\code{4;\anonterm;17}}}{\set{i}}$ to the output, Inequality~\eqref{eq:vcExample4}(left). 
The first recursive call yields
\begin{equation*}
  \verificationConditionsFuncOf{\hoaretriplet{i}{\code{4;\anonterm}}{\vcStrongestPostSemOf{\set{i}}{\code{4;\anonterm}}}} = 
    \verificationConditionsFuncOf{\hoaretriplet{i}{\code{4}}{\vcStrongestPostSemOf{\set{i}}{\code{4}}}}
    \cup
    \verificationConditionsFuncOf{\hoaretriplet{\vcStrongestPostSemOf{\set{i}}{\code{4}}}{\code{\anonterm}}{\vcStrongestPostSemOf{\set{i}}{\code{4;\anonterm}}}}
    \ .
\end{equation*}
The first of these calls adds $\angelicAssertionStrongerOf{\vcStrongestPostSemOf{\set{i}}{\code{4}}}{\vcStrongestPostSemOf{\set{i}}{\code{4}}}$,  Inequality~\eqref{eq:vcExample4}(right). 
The second recursive call is for the non-terminal.
Inequality $\angelicAssertionStrongerOf{\vcStrongestPostSemOf{\set{i}}{\code{4;\anonterm}}}{\vcStrongestPostSemOf{\set{i}}{\code{4;\anonterm}}}$ checks entailment of the postcondition, Inequality~\eqref{eq:vcExample6}(left).
For soundness of the annotation, $\checkAssertionFuncOfN{\vcStrongestPostSemOf{\set{i}}{\code{4}}}{\anonterm}{\Gamma}$ is called.
For $j = 2$,  the oracle proposes the programs $\code{y=y+x}$ and $\code{y=y*x}$.
See that by our definition of the selection transformer $\Gamma$,
the inequality
$\angelicAssertionStrongerOf{\strongestPostSemOf{\apredset}{\code{y=y+x}} \cup \strongestPostSemOf{\apredset}{\code{y=y*x}}}{\Gamma(\apredset)}$ holds.
Thus, the function $\checkAssertionFunc$ outputs
two more inequalities (after resolving the recursive calls to the base case of $\verificationConditionsFunc{}$):
the inequality $\angelicAssertionStrongerOf{\vcStrongestPostSemOf{\set{i}}{4;8}}{\vcStrongestPostSemOf{\set{i}}{4;8}}$
considers the first production resolving the nonterminal
and
the inequality 
$\angelicAssertionStrongerOf{\vcStrongestPostSemOf{\set{i}}{4;12}}{\vcStrongestPostSemOf{\set{i}}{4;12}}$
considers the second production.
Both inequalities are added to the output, Inequality~\eqref{eq:vcExample6}(right) and Inequality~\eqref{eq:vcExample7}.

All inequalities of 
$\verificationConditionsFuncOf{\hoaretriplet{\apredset}{\asketch}{\apredsetp}}$ hold and so $\angelicHoareTripletHoldsSemantically{\apredset}{\asketch}{\apredsetp}$ by \Cref{th:vcSoundness}.
This means, the non-terminal can be resolved such that the resulting program will compute the factorial of $42$. 
We present next an algorithm to automatically choose the correct branch.

\begin{figure}
\begin{lstlisting}[multicols=2]
$\inHoareParenthesis{\xvar = 42 \wedge \yvar = 1}$
(
  $\inHoareParenthesis{\yvar * \xvar! = 42! \wedge \xvar \geq 0}$
  assume(x > 0);
  $\inHoareParenthesis{\yvar * \xvar! = 42! \wedge \xvar > 0}$
  N$\color{black}[\Gamma]$( |\label{line:fac:nontermStart}|
    $\inHoareParenthesis{\yvar * \xvar! = 42! \wedge \xvar > 0}$
    y = y + x;
    $\inHoareParenthesis{(\yvar - \xvar) * \xvar! = 42! \wedge \xvar > 0}$
    $\color{black}\angelicChoice{}$ 
    $\inHoareParenthesis{\yvar * \xvar! = 42! \wedge \xvar > 0}$
    y = y * x; 
    $\inHoareParenthesis{(\yvar / \xvar) * \xvar! = 42! \wedge \xvar > 0}$
  )|\label{line:fac:nontermEnd}|
  $\leftHoareParenthesis{}(\yvar - \xvar) * \xvar! = 42! \wedge \xvar > 0,$
    $(\yvar / \xvar) * \xvar! = 42! \wedge \xvar > 0\rightHoareParenthesis{}$
  x--; |\label{line:fac:xDec}|
  $\leftHoareParenthesis{}(\yvar - (\xvar + 1)) * (\xvar + 1)! = 42! \wedge \xvar + 1 > 0,$
    $(\yvar / (\xvar + 1)) * (\xvar + 1)! = 42! \wedge \xvar + 1 > 0\rightHoareParenthesis{}$
)$\color{black}^*[\set{\yvar * \xvar! = 42! \wedge \xvar \geq 0}]$
$\inHoareParenthesis{\yvar * \xvar! = 42! \wedge \xvar \geq 0}$
assume(x = 0);
$\inHoareParenthesis{\yvar * \xvar! = 42! \wedge \xvar = 0}$
\end{lstlisting}
\vspace{-2mm}
    \caption{Proof outline of a sketch for computing the factorial function.}
    \label{code:exampleVCFactorial}
\end{figure}

%% file: sections-rewriteprooftoprog-main.tex
\newcommand{\setofineq}{\mathit{VC}}
\newcommand{\synthconds}{\mathit{SynthConds}}
\newcommand{\syn}{\mathit{syn}}
\newcommand{\synof}[2]{\syn(#1, #2)}
\newcommand{\myin}{\;\mathop{\texttt{\textcolor{blue}{in}}}\;}
\newcommand{\mylet}{\;\mathop{\texttt{\textcolor{blue}{let}}}\;}
\newcommand{\myand}{\;\mathop{\texttt{\textcolor{blue}{and}}}\;}
\newcommand{\myfor}{\;\mathop{\texttt{\textcolor{blue}{for}}}\;}
\newcommand{\ite}[3]{#1\; \mathop{\texttt{\bfseries \textcolor{blue}?}}\; #2 \;\mathop{\texttt{\bfseries \textcolor{blue}:}}\; #3}

\section{Implementing Realization Logic}\label{ch:adAlgo}
We give an algorithm to compute a program from a proof outline in realizability logic. 
It is a deductive verification that collects a set of verification conditions whose validity shows that the program is a solution to the synthesis task. 
What is unconventional is that the validity has to be checked in the course of the verification condition computation. 
The reason is that the validity checks steer the program construction, and also the verification conditions that will be collected in the future. 
Despite these dynamics, we can show that a small number of verification condition checks will be sufficient to derive the program.
Moreover, the verification conditions compare ordinary predicates (rather than selections in the previous section), and therefore can be implemented as single SMT solver queries.
All this makes our algorithm efficient in practice.

Our algorithm implements the proof rules in realization logic, more precisely the non-backtracking strategy stated in \Cref{th:angelicChoiceElimination}. 
Given a proof outline and a predicate in the postcondition, it traverses the proof backwards to determine suitable preconditions and programs that justify the postcondition.
Our algorithm is thus a function of type $\syn:\angelProgramProofLang{}\times \demonicAssertions{}\ \rightarrow\ \demonicAssertions{}\times\demonicProgramLang{}$.  
It is defined in \Cref{def:add} and we discuss it in a moment. 
An important feature of the function is to discard predicates from selections, in which case it may return $(\fail{}, -)$ and an arbitrary program.

The function gives strong guarantees: when we start from a sketch that has been derived in realizability logic and a predicate in the postcondition, then the function is guaranteed not to fail but return a precondition and a program that, together with the given postcondition, will form a valid Hoare triple. 
It is also very efficient: the number of verification conditions that have to be checked is linear, actually bounded by the size of the proof outline.
\begin{theorem}[$\syn$-Sound-And-Complete]\label{Theorem:syn}
Consider $\vdash_{a}\hoaretriplet{\apredset}{\apo}{\apredsetp}$ and $\apredp\in\apredsetp$ with $\apredp\neq\fail{}$. 
Then $\synof{\hoaretriplet{\apredset}{\apo}{\apredsetp}}{\apredp}=(\apred, \aprog)$ 
with $\apred\in\apredset$, $\apred\neq\fail{}$, $\aprog\in\angelConcFuncOf{\proofToProgFuncOf{\apo}}$, and $\demonicHoareTripletHoldsSemantically{\apred}{\aprog}{\apredp}$. 
The number of SMT solver calls is at most $\sizeof{\apo}$. 
\end{theorem}
\looseness=-1
Function $\syn$ is given in \Cref{def:add}. 
If several cases apply, the topmost one will be taken.
This means the invocation will always return $(\fail{}, -)$ on an empty precondition, and the $\fail{}$ predicate in a precondition will always be skipped. 
When we have a command, we go through the predicates $\apred$ in the precondition 
until we find one that justifies the given postcondition $\apredp$. 
Since $\apredp\in\apredsetp$ and we started from a proof outline in realizability logic, the search for $\apred$ is guaranteed to be successful.
Moreover, for the guarantees given by \Cref{Theorem:syn} it does not matter which predicate $\apred$ we take. 
The notion of soundness in realizability logic gives the guarantee that the earlier (in the program text) assertions can handle every $\apred$. 
But of course the choice has an influence on the shape of the program. 

For a sequential composition $\concatof{\apo_1}{\apo_2}$, we take the given predicates $\apredp$, propagate it through $\apo_2$ to obtain $(\apredpp, \aprog_2)$, propagate $\apredpp$ through $\apo_1$ to get $(\apred, \aprog_1)$, and return $(\apred, \concatof{\aprog_1}{\aprog_2})$. 
Importantly, also here we have the guarantee that both calls will be successful. 
When we have a non-terminal, we consult the proof outlines for the right-hand sides.
For an angelic choice among two right-hand sides, we first try to synthesize a program from the left proof outline, and if we fail we know that we will be successful on the right.
Whether the synthesis with the left proof outline will succeed can be foreseen by checking if the target post condition $\apredp$ is in the post condition of the left proof outline.

The involved cases are for demonic choices and loops. 
These are the cases in which the program logics had to consider single predicates. 
This is mimicked here, and we first discuss the demonic choice.  
We go through all predicates $\apred$ in the given precondition.
We try to construct proof outlines for both branches in which $\apred$ is the precondition and $\apredp$ is the postcondition.
This is the task of the calls $\mathit{outl}(\hoaretriplet{\set{\apred}}{\proofToProgFuncOf{\apo_1}}{\set{\apredp}})$ and $\mathit{outl}(\hoaretriplet{\set{\apred}}{\proofToProgFuncOf{\apo_2}}{\set{\apredp}})$. 
We argue in a moment why these calls can actually be looked up in the given proof outline, and therefore mean no effort. 
If one of these calls aborts, we continue with the next $\apred$.
Otherwise, we obtain the proof outlines $\apo_1'$ and $\apo_2'$. 
We invoke $\synof{\apo_1'}{\apredp}$ and $\synof{\apo_2'}{\apredp}$. 
They will for sure return $\apred$ as it is the only precondition, but we need the programs to be able to return $(\apred, \choiceOf{\aprog_1}{\aprog_2})$. 

For loops, we have to make an assumption on the form of the proof outline: the loop invariant should not be lost through weakening. 
If assertions have to be weakened, we assume the proof outline takes the form $\hoaretriplet{\apredset}{\hoaretriplet{I}{\kleeneof{\apo}}{I}}{\apredsetp}$. 
Like in the previous case, we go through the predicates $i\in I$ to find one that is stronger than the given postcondition $\apredp$.
If it has been found, it remains to synthesize the program for the loop body. 
We reconstruct the proof outline with $i$ as the pre and postcondition, and rely on $\syn$ to determine the program.
If the proof outline takes the form 
$\hoaretriplet{\apredset}{\hoaretriplet{I}{\kleeneof{\apo}}{I}}{\apredsetp}$,
instead of returning an $i$ of $I$, we return an $\apred$ of $\apredset$ with $\demonicAssertionStrongerOf{\apred}{i}$.

We elaborate on why the function calls $\mathit{outl}(\hoaretriplet{\set{\apred}}{\proofToProgFuncOf{\apo}}{\set{\apredp}})$ mean no overhead. 
The point is that the Rules \ruleLabel{DEM} and \ruleLabel{LOOP} can only deal with single predicates in the precondition. 
This means when we constructed the proof outline of interest, 
we have constructed the proof outlines that we now need as a byproduct.
We just have to store them explicitly to be able to look them up now.  
This explains the case of loops.

In the case of demonic choices, the proof outline construction may abort, and we explain how to handle this with hashing. 
Since Rule \ruleLabel{DEM} requires a singleton as the precondition, we must have done proofs
of the form
$\hoaretriplet{\set{\apred_i}}{\choiceOf{\proofToProgFuncOf{\apop_1}}{\proofToProgFuncOf{\apop_2}}}{\apredsetp_i}$.
From the requirements of Rule \ruleLabel{DEM}, we also have the proof outlines
$\hoaretriplet{\set{\apred_i}}{\apop_j}{\apredsetp_i}$ for $j = 1,2$.
The proofs might have been gathered to form the proof outline
$\hoaretriplet{\apredset}{\choiceOf{\apo_1}{\apo_2}}{\apredsetp}$.
Since the predicate $\apredp$ is in $\apredsetp$, there must be a selection $\apredsetp_i$ with $\apredp \in \apredsetp_i$.
We can reuse the corresponding proof outlines
$\hoaretriplet{\set{\apred_i}}{\apop_j}{\apredsetp_i}$
by weakening the postcondition $\apredsetp_i$ to the singleton $\set{\apredp}$.

\begin{figure}
    \footnotesize
    \begin{align*}
        \synof{\hoaretriplet{\emptyset}{\apo}{\apredsetp}}{\apredp}\;\definingEquals\;\;&(\fail{}, -)\\
        \synof{\hoaretriplet{\set{\fail{}}\cup\apredset}{\apo}{\apredsetp}}{\apredp}\;\definingEquals\;\;&\synof{\hoaretriplet{\apredset}{\apo}{\apredsetp}}{\apredp}\\
        \synof{\hoaretriplet{\set{\apred}\cup \apredset}{\code{com}}{\apredsetp}}{\apredp} 
        \;\definingEquals\;\; &  
        \ite{\demonicAssertionStrongerOf{\progSemFuncOrigOf{\code{com}}{\apred}}{\apredp}}
        {(\apred, \code{com})}
        {\synof{\hoaretriplet{\apredset}{\code{com}}{\apredsetp}}{\apredp}}\\
        \synof{\concatof{\apo_1}{\apo_2}}{\apredp} 
        \;\definingEquals\;\; &  \mylet\ (\apredpp, \aprog_2) = \synof{\apo_2}{\apredp}\myand (\apred, \aprog_1) = \synof{\apo_1}{\apredpp}\myin (\apred, \concatof{\aprog_1}{\aprog_2}) 
        \\
         \synof{\anonterm(\apo)}{\apredp} 
        \;\definingEquals\;\; & \synof{\apo}{\apredp} 
        \\
        \synof{\apo_1\bnf\apo_2}{\apredp} 
        \;\definingEquals\;\; & \mylet (\apred, \aprog)\in \synof{\apo_1}{\apredp}\myin 
        \ite{\apred\neq \fail{}}{(\apred, \aprog)}{\synof{\apo_2}{\apredp}}
        \\
    \end{align*}
\vspace{-0.9cm}
    \begin{align*}
            \synof{\choiceOf{\hoaretriplet{\set{\apred}\cup \apredset}{\apo_1}{\apredsetp}}{\hoaretriplet{\set{\apred}\cup\apredset}}{\apo_2}{\apredsetp}}{\apredp} 
        \;\definingEquals\;\; & \ite{\mathit{outl}(\hoaretriplet{\set{\apred}}{\proofToProgFuncOf{\apo_i}}{\set{\apredp}})=\mathsf{abort}}{\synof{\choiceOf{\hoaretriplet{\apredset}{\apo_1}{\apredsetp}}{\hoaretriplet{\apredset}}{\apo_2}{\apredsetp}}{\apredp}}{}\\
        &\hspace{-3cm}\mylet\ \apo_i' = \mathit{outl}(\hoaretriplet{\set{\apred}}{\proofToProgFuncOf{\apo_i}}{\set{\apredp}})
\myand (\apred, \aprog_i) = \synof{\apo_i'}{\apredp} \myin\ (\apred, \choiceOf{\aprog_1}{\aprog_2})\\
        \synof{\hoaretriplet{\set{i}\cup I}{\kleeneof{\apo}}{\set{i} \cup I}}{\apredp} 
        \;\definingEquals\;\; &  \ite{i\not \demonicAssertionStronger{}\apredp}{\synof{\hoaretriplet{I}{\kleeneof{\apo}}{I}}{\apredp}}{}\\
        &\hspace{-3cm}\mylet\ \apo' = \mathit{outl}(\hoaretriplet{\set{i}}{\proofToProgFuncOf{\apo}}{\set{i}})\myand (i, \aprog) = \synof{\apo'}{i}\myin(i, \kleeneof{\aprog})
    \end{align*}
    \caption{Definition of the synthesis function.}
    \label{def:add}
\end{figure}

\input{sections-rewriteprooftoprog-example.tex}

%% file: sections-rewriteprooftoprog-example.tex
\subsection{Example: Synthesizing Factorial Function}
\looseness=-1
Recall the sketch $\code{sketch}$ from last section's example, \Cref{code:exampleVCFactorial}.
In the last section, we have proven $\angelicHoareTripletHoldsSemantically{\apredset}{\asketch}{\apredsetp}$ with $\apredset = \set{\xvar = 42 \wedge \yvar = 1}$ and $\apredsetp = \set{\yvar = 42!}$.
In this example, we will concretize the sketch to a program $\aprog$ such that $\demonicHoareTripletHoldsSemantically{\apred}{\aprog}{\apredp}$ holds.
To do so, we use the function $\syn$.
The whole proof outline is called $\apo$.
We use the same  abbreviations as in the last example.
Additionally, with $\code{T8}$ we abbreviate the full realizability triple of Line $8$.
That means $\code{T8}$ stands for Lines $7$ through $9$.
Other realizability triples are abbreviated similarly.
The realizability triple around the nonterminal is abbreviated by $\code{TN}$, i.e.\ $\code{TN}$ stands for Lines $5$ to $16$.
We abbreviate the selection of Lines $15$ and $16$ with $A_{15}$.
Other selections are abbreviated similarly.
Let $\set{\apred} = \apredset$ and $\set{\apredp} = \apredsetp$.
We start by calling 
$\synof{\hoaretriplet{\apredset}{\asketch}{\apredsetp}}{\apredp}$.
More specifically, the call is 
$\synof{\hoaretriplet{\apredset}{\kleeneof{\code{body}}}{I}\code{;T22}}{\apredp}$.
The selection $I$ is the singleton $\set{i}$.
The function invokes the call $\synof{\code{T22}}{\apredp}$ which returns $(i, \code{assume(x = 0)})$.
Next, it invokes $\synof{\hoaretriplet{\apredset}{\kleeneof{\code{body}}}{I}}{i}$ which, in turn, invokes $\synof{\hoaretriplet{\set{i}}{\code{body}}{\set{i}}}{i}$ because $\demonicAssertionStrongerOf{i}{i}$ holds trivially.
We can reuse the already done proof $\hoaretriplet{\set{i}}{\code{body}}{\set{i}}$ as $\apop$.
The call $\synof{\hoaretriplet{\set{i}}{\code{body}}{\set{i}}}{i}$
invokes $\synof{\hoaretriplet{A_{15}}{\code{x-{}-}}{\set{i}}}{i}$.
The call $\synof{\hoaretriplet{\set{A_{15}[0]} \cup \set{A_{15}[1]}}{\code{x-{}-}}{\set{i}}}{i}$ will invoke 
$\synof{\hoaretriplet{\set{A_{15}[1]}}{\code{x-{}-}}{\set{i}}}{i}$
because the result of $\progSemFuncOf{\code{x-{}-}}{A_{15}[0]}$ is not more precise than $i$.
However, the call 
$\synof{\hoaretriplet{\set{A_{15}[1]}}{\code{x-{}-}}{\set{i}}}{i}$
returns $(A_{15}[1], \code{x-{}-})$.
Here, we witness a foreshadowing on how the wrong angelic choice will be eliminated.
The eliminated predicate stems from the left production of the nonterminal which will not be able to produce the now selected predicate $A_{15}[1]$.
We continue the next recursive call of $\synof{\hoaretriplet{\set{i}}{\code{body}}{\set{i}}}{i}$, i.e. 
$\synof{\code{T4;TN}}{A_{15}[1]}$.
It invokes $\synof{\code{TN}}{A_{15}[1]}$
This call first checks whether the left branch can produce the target predicate.
This is not the case, so 
eventually $(\fail{}, -)$ is returned, eliminating the left branch.
The right branch is now considered, where $\synof{\code{T12}}{A_{15}[1]}$ resolves to $(A_5[0], \code{y = y*x})$.
Therefore, the next call is 
$\synof{\code{T4}}{A_{5}[0]}$ which resolves to the pair $(i, \code{assume (x>0)})$.
Getting back to the recursive call of the loop, $\synof{\hoaretriplet{\apredset}{\kleeneof{\code{body}}}{I}}{i}$
yields the pair $(\apred, \kleeneof{(\code{4;12;17})})$
because $\demonicAssertionStrongerOf{\apred}{i}$
with $\set{\apred} = \apredset$.
In total, we get $(\apred, \aprog) = (\apred, \concatof{\kleeneof{(4;12;17)}}{\code{22}})$ as output.
We know that 
$\demonicHoareTripletHoldsSemantically{\apred}{\aprog}{\apredp}$
holds.
This means, the program indeed calculates $42$ factorial.

%% file: sections-application-main.tex
\section{Application: Memory Management in Lock-Free Data Structures}\label{ch:implementation}
\input{sections-application-implementation.tex}

%% file: sections-application-implementation.tex
We show how to employ realizability and realization logic to automatically generate code for the memory management in lock-free data structures. 
This will also be the setting for our experiments. 
There are several aspects that make this setting interesting for benchmarks: the states are type assignments and the semantics of commands is an abstract one. 
So \semgus\ \cite{dantoni:semgus} is an ideal choice, while \sygus\ \cite{sygus} and solver-aided languages\ \cite{sketch:phd,torlak:rosette:onward} need an encoding to capture the synthesis tasks.   
The programs contain loops, which is difficult for \sygus\ solvers.  
Finally, there are many synthesis options, and the right ones may be far apart. 
This creates a large space of potential solutions \cite{SearchBased18} that our compositional method manages to explore in a matter of seconds. 
\paragraph{Safe Memory Reclamation}
The asynchronous nature of lock-free data structures makes manual memory management difficult: 
if the threads do not synchronize, how does a thread know that it is safe to free a memory cell, or safe to access it? 
The solution is to add to the data structure a safe memory reclamation algorithm (SMR) which acts as a central instance that is informed about the intentions of all threads. 
If a thread wants to access a memory cell, it asks the reclamation algorithm to protect the cell.
If it wants to free a memory cell, it asks the SMR to do so. 
As the reclamation algorithm manages the protections, it can defer the free until it is safe. 
Garbage collection is an SMR, but there are more efficient solution.  
Epochs are a simple model~\cite{EBRPhd,HarrisList}, we consider here the hazard pointers that have been proposed for addition to the next C++ standard~\cite{HPProposal}. 
To be precise, we support the elaborate version in which the first hazard pointer is stronger than the second. 

Reclamation algorithms tend to implement lock-free data structures.
This means protection calls are non-atomic, and may be successful or fail due to interference. 
The reclamation algorithm will not be able to determine success.
Instead, the thread that issued the protection call will have to find out by checking global invariants:
if, for example, a sentinal node has not changed when returning from the call, then the call was successful.
Unfortunately, finding the right combination of protection calls and checks has turned out error-prone.

\looseness=-1
\Citet{POPL2019,POPL2020} gave a type system that checks memory safety. 
It can be instantiated to different reclamation algorithms and applies to a variety of data structures.
We report on experiments with the queues MS and DGLM, and the sets ORVYY, Michael, VechevCAS, and VechevDCAS.
The type system is control-flow sensitive, which means the typing can be written as a proof outline in Hoare logic.
A type check on a program $\aprog$ is successful, if the triple $\sebHoareTripletHolds{\neg \fail{}}{\aprog}{\neg \fail{}}$ can be derived: starting from no assumptions about the protection of pointers, the program will not fail. 
This means all pointers were protected before they were dereferenced.
\Cref{Section:BackgroundSMR} provides more background on \cite{POPL2019,POPL2020}.  
An important detail is that they use code annotations to inform the type system about protections that can otherwise otherwise only be inferred with a shape analysis. 
These annotations have to be discharged separately. 
We omit these checks in our experiments, and therefore our synthesis is only valid relative to the validity of the annotations.
\input{sections-application-applyshort.tex}
\paragraph*{Synthesis}\label{sec:implementation:generalApproach} 
Given a lock-free data structure and a reclamation algorithm, our goal is to synthesize in the data structure calls to the reclamation algorithm and code annotations that, together, guarantee memory safety, more precisely, make the above type check go through. 

First, we enrich every line of code with a non-terminal from which calls to the SMR algorithm and invariant annotations may be generated. In the case of a single hazard pointer, for example, the \\[0.05cm]
    \begin{minipage}{0.46\textwidth}
\begin{lstlisting}
AC ::= skip; $\color{black}\angelicChoice{}$
atomic {\@inv active(v);} $\color{black}\angelicChoice{}$
(in:protect(v);re:protect(v);)
\end{lstlisting}
    \end{minipage}
    \begin{minipage}{0.54\textwidth}
the non-terminal has the definition to the left.   
For each variable $\code{v}$, we may issue a protection, add an annotation saying that the variable has not been freed, or simply skip the annotation. 
For global invariants,
    \end{minipage}

\vspace{0.1cm}\noindent  we use heuristics of where to put them, which improves the scalability.
The result of this phase is a sketch $\asketch$. 
Next, we try to prove $\angelicHoareTripletHolds{\neg\fail{}}{\asketch}{\neg\fail{}}$ in realizability logic. 
We use the algorithm from Section~\ref{ch:vc} and check the validity of the corresponding verification conditions. 
If successful, we also get a proof outline that we feed to realization logic, implemented in the synthesis algorithm from Section~\ref{ch:adAlgo}. 
The algorithm has the guarantee to be successful, and returns a program that corresponds to the original code with suitable annotations added. 
What is interesting is that we can control the synthesis function to avoid reclamation calls and code annotations whenever possible. 
We give details on the implementation and its behavior on experiments. 

%% file: sections-application-applyshort.tex
\newcommand{\mytypes}{\mathbb{T}}
\newcommand{\atype}{\mathit{t}}
\paragraph{Instantiation}
We instantiate the parameters of our development. 
\Citet{POPL2019,POPL2020} track the protection of memory cells by assigning types from a set~$\mytypes$ to the pointers in the program $\vars{}$. 
Our states, predicates, and selections are thus 
$\states{} = \vars{} \rightarrow \mytypes$, $\demonicAssertions{} = \powersetOf{\vars{} \rightarrow \mytypes} \cup \set{\fail{}}$, and $\angelicAssertions{}=\powersetOf{\powersetOf{\vars{} \rightarrow \mytypes} \cup \set{\fail{}}}$. 

Our implementation needs an assertion language to represent predicates and selections. 
We use $\demonicAssertionsAbs{} = (\vars{} \rightarrow \mytypes) \cup \set{\fail{}}$ and 
$\angelicAssertionsAbs{}=\powersetOf{\demonicAssertionsAbs{}}$.  
The definition of the abstract predicates is due to \cite{POPL2019,POPL2020}. 
The first step is to apply a Cartesian abstraction on the predicates, which yields $\vars{}\rightarrow \powersetOf{\mytypes}$. 
A Cartesian predicate $\set{\xvar:\set{\atype_1, \atype_2}}\times \set{\yvar:\set{\atype}}$ represents the concrete predicate $\set{\xvar:\atype_1\wedge \yvar:\atype, \xvar:\atype_2\wedge\yvar:\atype}$. 
The second step is to exploit the fact that the types form a lattice, and represent 
$\xvar:\set{\atype_1, \atype_2}$ by $\xvar:\atype_1\sqcup \atype_2$. 
The details are in \Cref{Section:AbsInt}. 

The commands are the usual ones in \textsf{C}. 
The semantics is from~\cite{POPL2020} and defined in a way that works with the above abstraction.

%% file: sections-application-evaluation-main.tex
\section{Evaluation}\label{sec:evaluation}
We implemented our synthesis algorithm in a \textsf{C++} program. 
As we do not yet have an assertion language for selections, 
our implementation targets the memory reclamation problem discussed above.
It can handle different reclamation algorithms and challenging data structures, though. 
We give details on the implementation and elaborate on the performance, which was surprisingly good.

\input{sections-application-evaluation-implementationdetails.tex}
\input{sections-application-evaluation-results.tex}

%% file: sections-application-evaluation-implementationdetails.tex
\paragraph*{Implementation}
The most important technique we implemented to improve the scalability is an optimistic variant of our approach. 
The pessimistic variant is the one from \Cref{ch:implementation} that first validates the proof outline and then derives a program. 
The optimistic variants derives the program without having validated the proof outline. 
The point is that function $\syn$ checks the verification conditions it needs to construct the program anyhow, and therefore the output will be sound.  
The verification conditions for proof outlines have to compare selections, and each such comparison may create quadratically many solver calls.  

We moreover implemented the following application domain specific heuristics. 
We assume local pointers do not need to be protected. 
We only insert annotations $\code{@inv active(v)}$ for sentinel nodes. 
The point is that sentinel nodes will never be freed.
As discussed above, this would have to be confirmed by a check that we omit.
We automatically compute loop invariants from the predicates that are available when entering a loop. 
There is no saturation nor widening. 
Still, the approach worked for all data structures we examined. 
For resolving non-terminals, we prefer $\code{skip}$ over SMR calls and SMR calls over annotations.

%% file: sections-application-evaluation-results.tex
\paragraph{Results}
We experimented with the seven lock-free data structures and two hazard pointer variants discussed in \Cref{ch:implementation}. 
The environment is an Apple M2 with 8GB RAM.  
We do not have a solver backend but hand the bitvectors we need within our implementation. 
The results are given in \Cref{table:resultsPessimistic}.
In some methods we used hints that can be deduced by a lightweight shape analysis.

\input{sections-application-evaluation-tablepreamble.tex}
\input{sections-application-evaluation-tablepessimistic.tex}

The first column shows the maximum, average, and median size of a selection.
The next shows the number of invocations of the comparison function on predicates per comparison of selection. 
For example, $\angelicAssertionStrongerOf{\set{a}}{\set{b_1, b_2, b_3}}$ may yield three comparisons $\demonicAssertionStrongerOf{a}{b_i}$. 
This only applies to the pessimistic approach.
The next two columns display the time for the functions $\verificationConditionsFunc{}$ and $\syn$ in the pessimistic approach. 
The last two columns are for the optimistic approach. 
It clearly outperforms the pessimistic approach, but both are quite fast. 

The key insight of our experiments is that there are two classes of synthesis options. 
The options that are handled well by our approach lead to similar predicates and thus small selections, or they fail early.
The options that are difficult for our approach lead to new predicates that fail late.
This can be seen in the benchmarks with two hazard pointers.
In the set implementations, the order among the hazard pointers matters, and if we try to synthesize a protection with the wrong order we fail early.
For DGLM, the order does not matter and the synthesis time increases dramatically.

Taking a step back, the experiments indicate that our realizability and realization logics should be combined with an outer search~\cite{SearchBased18}. 
The search would concentrate on the synthesis options that are difficult for our approach and only fix them. 
Our approach would handle the remaining, easier options.
We see this as a promising direction for future work.

%% file: sections-application-evaluation-tablepreamble.tex
\newcolumntype{a}{>{\columncolor{teal!9}}r}
\newcolumntype{s}{>{\columncolor{teal!9}}l}
\newcolumntype{t}{>{\columncolor{teal!9}}c}
\newcolumntype{u}{>{\columncolor{teal!4}}r}
\newcolumntype{v}{>{\columncolor{teal!4}}c}

%% file: sections-application-evaluation-tablepessimistic.tex
\begin{table}
    \footnotesize
    \centering
    \caption{Results of experiments of the Pessimistic and Optimistic Approach conducted on an Apple M2.}
    \label{table:resultsPessimistic}
    \begin{tabular}{ccc|ccrr|rr}
        &                                  &                              & Max/Avg/Med & Max/Avg/Med
        & \multicolumn{2}{c|}{Pessimistic}& \multicolumn{2}{c}{Optimistic} \\ 
        &                                  &                              & $\lvert\apredset\rvert$ & $\demonicAssertionStronger{}$ per $\angelicAssertionStronger{}$ & \faStopwatch{} $\verificationConditionsFunc{}$ & \faStopwatch{} $\syn$ & \faStopwatch{} $\syn$ & \faStopwatch{} $\verificationConditionsFunc{}$ \\ 
        \hline
        \multirow{4}{*}{HP1}  & \multirow{2}{*}{Treiber's \cite{treibersStack}} & Push    & 6 / 1.57 / 1       & 12 / 1.85 / 1                                                   & $ < 0.1s$                                     & $ < 0.1s$                              & $ < 0.1s$                              & $ < 0.1s$                                                          \\ 
         & & Pop &  6 / 1.43 / 1 & 18 / 1.71 / 1 & $<0.1s$ & $<0.1s$ & $<0.1s$ & $<0.1s$ \\
         & MS \cite{msqueue} & EnQ & 5 / 1.27 / 1 & 12 / 1.40 / 1 & $<0.1s$ & $<0.1s$ & $<0.1s$ & $<0.1s$ \\
         & DGLM \cite{DGLMQueue} & EnQ & 17 / 1.3 / 1 & 64 / 1.9 / 1 & $0.7s$ & $<0.1s$ & $<0.1s$ & $<0.1s$ \\
         \hline
        \multirow{10}{*}{HP2} & MS \cite{msqueue} & DeQ & 90 / 1.6 / 1 & 1K / 2.92 / 1 & $17s$ & $0.6s$ & $0.6s$ & $<0.1s$ \\
         & DGLM \cite{DGLMQueue} & DeQ & 29 / 1.74 / 1 & 5K / 18.9 / 1 & $249s$ & $6.3s$ & $5.9s$ & $<0.1s$ \\
         & \multirow{2}{*}{ORVYY \cite{ORVYY}} & Add & 158 / 3.85 / 1 & 276 / 3.6 / 1 & $1.4s$ & $<0.1s$ & $<0.1s$ & $<0.1s$ \\
         &  & Rm & 32 / 1.29 / 1 & 276 / 1.96 / 1 & $0.6s$ & $0.1s$ & $0.1s$ & $<0.1s$ \\
         & \multirow{2}{*}{Michael \cite{MichaelSet}} & Add & 127 / 3.06 / 1 & 5K / 13.4 / 1 & $3.1s$ & $2.1s$ & $2.2s$ & $<0.1s$ \\
         &  & Rm & 127 / 3.15 / 1 & 5K / 13.7 / 1 & $1.5s$ & $2.1s$ & $2.1s$ & $<0.1s$ \\
         & \multirow{2}{*}{VechevCAS \cite{VechevDCAS}} & Add & 25 / 1.7 / 1 & 148 / 3.1 / 1 & $0.6s$ & $<0.1s$ & $<0.1s$ & $<0.1s$ \\
         &  & Rm & 21 / 1.7 / 1 & 178 / 3.1 / 1 & $0.1s$ & $<0.1s$ & $<0.1s$ & $<0.1s$ \\
         & \multirow{2}{*}{VechevDCAS \cite{VechevDCAS}} & Add & 45 / 2.64 / 1 & 341 / 6.68 / 1 & $0.4s$ & $<0.1s$ & $<0.1s$ & $<0.1s$ \\
         &  & Rm & 21 / 1.7 / 1 & 145 / 2.9 / 1 & $0.1s$ & $<0.1s$ & $<0.1s$ & $<0.1s$ \\
        \end{tabular}
       \end{table}

%% file: sections-discussion-main.tex
\section{Related Work} 
We have already discussed the related work on \semgus, \sygus, and solver-aided programming in the introduction.  
There is a body of work on programming models with angelic non-determinism, dating back to~\citet{floyd}. 
The common case, however, is that the angelic choices are made at runtime, and therefore can react to the demonic choices. 
These models are related to two-player zero-sum graph games~\cite{LNCS2500} typically used in reactive synthesis, as originally proposed by \citet*{church} and later developed by \citet*{pnueliRosner89}. 
In \semgus, the angelic choices have to be made up-front, which adds a form of imperfect information~\cite{reif} that the common case does not have to deal with. 
This imperfect information shines through in Rules~\ruleLabel{DEM} and \ruleLabel{LOOP}, where the angelic player cannot react. 
Two works, however, are closer to our development. 

\citet*{mamouras:lmcs:2017} presents a Hoare logic to reason about the correctness of programs with (runtime) angelic non-determinism. 
His assertions are classical predicates, which means the proof construction is forced to decide on a synthesis option early on, may later fail due to incompatible decisions, and in this case has to backtrack. 
To avoid precisely this backtracking, we proposed selections as assertions in realizability logic.  
As a consequence, the resulting theories are largely different.

\citet{celiku-von-wright} consider the same class of programs as Mamouras, but try to refine angelic choices to conditionals and loops.
The refinement is guided by a proof outline, and may be compared to our realization logic. 
An important difference is that they work with ordinary predicates, and try to synthesize appropriate ones. 
We work with the new selections, instead, in which we have gathered the relevant predicates.

Also related is the recent~\cite{dillig:popl:burst} that uses angelic non-determinism to synthesize recursive functions.  
During recursive calls, they select an output value for the function and later check whether the choice can be justified. 
The work \cite{SynthesisConditions10,OracleGuided10} infers a program from verification conditions that contain unknown program parts which, in the case of~\cite{OracleGuided10}, should be filled by library components. 
The crucial difference to realization logic is that we have eagerly analyzed the program space, and  the task of $\syn$ amounts to a concretization rather than a search. 
We have not seen a backtracking freedom guarantee in related work. 

The deductive approach decomposes the task of synthesizing a program into synthesis tasks for subprograms~\cite{MW80}. 
A popular idea is to factorize the subprograms along input-output equivalence so that the computation can proceed with equivalence classes~\cite{FlashFill11,AGK13,Transit13,FlashMETA15,Dillig15,OseraZ15}. 
What is different in our work is that we try to be eager in the bottom-up construction, and consider a set of synthesis options whose behavior does not have to match.
Related to our approach are~\cite{dillig:synthAbstractionRefinement,DC17} that weaken the input-output equivalence by considering an abstract domain. 
A solution for the abstract domain is then a restricted class of programs that can be searched more efficiently.
Realizability and realization logic may prove useful for this task.

We introduced our eager approach to reduce the number of costly verification queries, a  
prominent endeavor in search-based synthesis~\cite{SearchBased18}. 
An important development is to inform the solver about semantic properties of the program sought~\cite{exampleclosure20}. 
In the context of example-based specifications, a semantic property is the invariance to perturbations of the example set, which can be incorporated by enriching the example set.

%% file: sections-appendix-main.tex
\newpage
\appendix
\input{sections-appendix-proofs-main.tex}
\section{Details on the Application}
\input{sections-appendix-application-preliminaries.tex}
\input{sections-application-assertionlanguage.tex}
\input{sections-application-vcexample.tex}

%% file: sections-appendix-proofs-main.tex
\section{Proofs}\label{ch:proofs}
We present omitted proofs and details of this paper.
\input{proofsChapter2.tex}
\input{proofsChapter3.tex}
\input{proofsChapter4.tex}
\input{proofsChapter5.tex}
\input{proofsChapter7.tex}

%% file: proofsChapter2.tex
\subsection{Proofs for \Cref{ch:pls}}
\begin{lemma}\label{lem:comMonotonic}
    The function $\progSemFunc{\code{com}}$ is monotonic for all commands $\code{com}$, i.e.\ the following equation holds:
    \begin{equation*}
        \demonicAssertionStrongerOf{\apred}{\apredp} 
        \ \implies \ 
        \demonicAssertionStrongerOf{\progSemFuncOf{\code{com}}{\apred}}{\progSemFuncOf{\code{com}}{\apredp}}  \ .
    \end{equation*}
\end{lemma}
\begin{proof}
    Let $\apred$ and $\apredp$ be elements of $\demonicAssertions{}$ with $\demonicAssertionStrongerOf{\apred}{\apredp}$.

Case 1: 
    $\progSemFuncOf{\code{com}}{\apredp} = \fail{}$. 
    Then $\demonicAssertionStrongerOf{\progSemFuncOf{\code{com}}{\apred}}{\progSemFuncOf{\code{com}}{\apredp}}$.

Case 2: 
    $\progSemFuncOf{\code{com}}{\apredp} \neq \fail{}$. 
Then, $\apredp \neq \fail{}$ and there is no $p$ in $\apredp$ with $\progSemFuncOrigOf{\code{com}}{p} = \fail{}$.
Because $\apredp \neq \fail{}$, we have $\apred \neq \fail{}$ and $\apred \subseteq s$.
Therefore, there also is no $p$ in $\apred$ with $\progSemFuncOrigOf{\code{com}}{p} = \fail{}$.
    Now, let $p'$ be an element of $\progSemFuncOf{\code{com}}{\apred}$.
Thus, there exists a $p$ in $\apred$ with $p' \in \progSemFuncOrigOf{\code{com}}{p}$.
    Since $\apred \subseteq s$, this $p$ also is in $\apredp$ and therefore, $p'$ also is in $\progSemFuncOf{\code{com}}{\apredp}$, making $\progSemFuncOf{\code{com}}{\apred}$ a subset of $\progSemFuncOf{\code{com}}{\apredp}$.
    In total, $\demonicAssertionStrongerOf{\progSemFuncOf{\code{com}}{\apred}}{\progSemFuncOf{\code{com}}{\apredp}}$ holds.
\end{proof}

\begin{lemma}\label{lem:demonsPartialOrder}
    The relation $\demonicAssertionStronger{}$ is a partial order relation on $\demonicAssertions{}$.
    This means, the order is reflexive, transitive and antisymmetric, i.e.\ the following equations hold:
    \begin{align*}
        \forall \, \apred \in \demonicAssertions{} & \dv \demonicAssertionStrongerOf{\apred}{\apred} \\
        \forall \, \apred, \apredp, \apredpp \in \demonicAssertions{} & \dv \demonicAssertionStrongerOf{\apred}{\apredp} \wedge \demonicAssertionStrongerOf{\apredp}{\apredpp} \implies \demonicAssertionStrongerOf{\apred}{\apredpp} \\
        \forall \, \apred, \apredp \in \demonicAssertions{} & \dv 
        \demonicAssertionStrongerOf{\apred}{\apredp} \wedge \demonicAssertionStrongerOf{\apredp}{\apred} \implies \apred = \apredp
    \end{align*}
\end{lemma}
\begin{proof}
    Let $\apred, s, t$ be elements of $\demonicAssertions{}$.

    Reflexivity ($\demonicAssertionStrongerOf{\apred}{\apred}$): 

    Case 1: $\apred = \fail{}$. We get $\demonicAssertionStrongerOf{\apred}{\apred}$ immediately.
    Case 2: $\apred \neq \fail{}$. So $\apred$ is a set of states. 
    Then, $\apred \subseteq r$ holds.
    
    Transitivity ($\demonicAssertionStrongerOf{\apred}{\apredp} \wedge \demonicAssertionStrongerOf{\apredp}{t} \implies \demonicAssertionStrongerOf{\apred}{t})$:

    Case 1: $\apredpp = \fail{}$. Then $\demonicAssertionStrongerOf{\apred}{\apredpp}$ holds.
    Case 2: $\apredpp \neq \fail{}$. Then $\apredp \neq \fail{}$ and $\apred \neq \fail{}$.
    By transitivity of the relation$\subseteq$, we get $\demonicAssertionStrongerOf{\apred}{\apredpp}$.
    
    Anti-symmetry ($\demonicAssertionStrongerOf{\apred}{\apredp} \wedge \demonicAssertionStrongerOf{\apredp}{\apred} \implies \apred = \apredp$):

    If either one of $\apred$ and $\apredp$ is $\fail{}$ the other one also has to be fail for both inequalities to hold.
    If neither one is $\fail{}$, anti-symmetry follows from the anti-symmetry of $\subseteq$.
\end{proof}

\begin{lemma} \label{lem:demonsLattice}
    The partial order $(\demonicAssertions{}, \demonicAssertionStronger{})$ is a complete lattice.
    That means for any subset $\apredset$ of $\demonicAssertions{}$ there exists a join and a meet.
    In fact, the join and meet can be computed by the following equations: 
    Let $\apredset$ be a subset of $\demonicAssertions{}$.
    \begin{align*}
        \joinOf{\apredset} &=
        \begin{cases}
            \fail{} \hphantom{\qquad\qquad}&, \fail{} \in \apredset \\
            \bigcup_{\apred \in \apredset} \apred &, \text{else}
        \end{cases} 
        \\
        \meetOf{\apredset} &=
        \begin{cases}
            \fail{} &, \apredset = \set{\fail{}} \\
            \bigcap_{\apred \in (\apredset \setminus \set{\fail{}})} r &, \text{else}
        \end{cases}
    \end{align*}
\end{lemma}
\begin{proof}
    Let $\apredset \subseteq \demonicAssertions{}$
    
    We show that $\joinOf{\apredset}$ is an upper bound of $\apredset$:
    Let $\apred$ be in $\apredset$.
    Case 1: Assume $\fail{} \in \apredset$.
    Then $\joinOf{\apredset} = \fail{}$.
    Thus, $\demonicAssertionStrongerOf{\apred}{\joinOf{\apredset}}$.
    Case 2: Assume $\fail{}$ is not in $\apredset$.
    Then $\joinOf{\apredset} = \bigcup_{\apred \in \apredset} \apred$.
    Thus, $\apred \subseteq \joinOf{\apredset}$ and $\demonicAssertionStrongerOf{\apred}{\joinOf{\apredset}}$.

    We show that $\joinOf{\apredset}$ is the least upper bound of $\apredset$:
    Let $u$ be an upper bound of $\apredset$.
    Case 1: Assume $\fail{} \in \apredset$.
    Then $\demonicAssertionStrongerOf{\fail{}}{u}$.
    Thus, $\demonicAssertionStrongerOf{\joinOf{\apredset}}{u}$.
    Case 2: Assume $\fail{}$ is not in $\apredset$.
    Then $\joinOf{\apredset} = \bigcup_{\apred \in \apredset} \apred$.
    Since $u$ is an upper bound, $\forall \apred \in \apredset: \apred \subseteq u$.
    Thus, $(\bigcup_{\apred \in \apredset} \apred) \subseteq u$.
    So, $\demonicAssertionStrongerOf{\joinOf{\apredset}}{u}$.
    
    We show that $\meetOf{\apredset}$ is a lower bound of $\apredset$:
    Let $\apred$ be in $\apredset$. 
    Case 1: Assume $\apredset = \set{\fail{}}$. Then $\apred = \fail{} = \meetOf{\apredset}$.
    Thus $\demonicAssertionStrongerOf{\meetOf{\apredset}}{\apred}$.
    Case 2: Assume $\apredset \neq \set{\fail{}}$.
    If $\apred = \fail{}$, we get $\demonicAssertionStrongerOf{\meetOf{\apredset}}{\apred}$.
    Otherwise, $\meetOf{\apredset} = \cap_{\apred \in \apredset \setminus \set{\fail{}}} \apred$. 
    Therefore, $\meetOf{\apredset} \subseteq \apred$ and thus $\demonicAssertionStrongerOf{\meetOf{\apredset}}{\apred}$.
    
    We show that $\meetOf{\apredset}$ is the greatest lower bound of $\apredset$:
    Let $l$ be a lower bound of $\apredset$.
    Case 1: Assume $\apredset = \set{\fail{}}$. Then $\meetOf{\apredset} = \fail{}$.
    Thus $\demonicAssertionStrongerOf{l}{\meetOf{\apredset}}$.
    Case 2: Assume $\apredset \neq \set{\fail{}}$. Let $\apred$ be in $\apredset \setminus \set{\fail{}}$.
    Then $l \subseteq \apred$.
    This holds for every $\apred$.
    Thus, $l \subseteq \meetOf{\apredset}$.
    Therefore, $\demonicAssertionStrongerOf{l}{\meetOf{\apredset}}$.  
\end{proof}

\subsection*{The Missing EMPTY Rule:}
For completeness, we need to add the following rule to the calculus:
\begin{mathpar}
\makeBlue{
\inferrule[\ruleLabelSmall{EMPTY}]{
		}{
\angelicHoareTripletHolds{\apredset}{\asketch}{\emptyset}
}
}
\end{mathpar}

\subsubsection*{Proof of \Cref{th:angelicHoareSoundnessCompleteness}}
\begin{proof}[\unskip\nopunct]
We prove soundness by structural induction over the proof tree.

\baseCase{} We have $\angelicHoareTripletHolds{\set{\apred}}{\code{com}}{\set{\apredp}}$ through the rule \ruleLabel{COM}.
Thus, we have $\demonicAssertionStrongerOf{\progSemFuncOf{\code{com}}{\apred}}{\apredp}$.
Therefore, the realizability triple $\angelicHoareTripletHoldsSemantically{\set{\apred}}{\code{com}}{\set{\apredp}}$ holds.

\inductionStep{}
In the first case, we have $\angelicHoareTripletHolds{\apredset}{\concatof{\asketch_1}{\asketch_2}}{\apredsetpp}$ through the rule \ruleLabel{SEQ}.
    Thus, we have the realizability triples $\angelicHoareTripletHolds{\apredset}{\asketch_1}{\apredsetp}$ and $\angelicHoareTripletHolds{\apredsetp}{\asketch_2}{\apredsetpp}$.
    By applying the induction hypothesis, we get $\angelicHoareTripletHoldsSemantically{\apredset}{\asketch_1}{\apredsetp}$ and $\angelicHoareTripletHoldsSemantically{\apredsetp}{\asketch_2}{\apredsetpp}$.
Let $\apredpp$ be an element of $\apredsetpp$.
Then, there exists a predicate $\apredp \in \apredsetp$ and a program $\aprog_2 \in \angelConcFuncOf{\asketch_2}$ with $\demonicHoareTripletHoldsSemantically{\apredp}{\aprog_2}{\apredpp}$.
And, there exists $\apred \in \apredset$ and $\aprog_1 \in \angelConcFuncOf{\asketch_1}$ with $\demonicHoareTripletHoldsSemantically{\apred}{\aprog_1}{\apredp}$.
By completeness of Hoare logic, we get 
$\demonicHoareTripletHolds{\apred}{\aprog_1}{\apredp}$ and 
$\demonicHoareTripletHolds{\apredp}{\aprog_1}{\apredpp}$.
Using rule \ruleLabel{SEQ}, we get 
$\demonicHoareTripletHolds{\apred}{\concatof{\aprog_1}{\aprog_2}}{\apredpp}$.

    In the next case we have the triple $\angelicHoareTripletHolds{\apredset}{\asketch}{\apredsetp}$ through the rule \ruleLabel{CSQ}.
    Thus we have $\angelicHoareTripletHolds{\apredset'}{\asketch}{\apredsetp'}$ with $\angelicAssertionStrongerOf{\apredset}{\apredset'}$ and $\angelicAssertionStrongerOf{\apredsetp'}{\apredsetp}$.
Through the induction hypothesis, we get $\angelicHoareTripletHoldsSemantically{\apredset'}{\asketch}{\apredsetp'}$.
Now, let $\apredp$ be an element of $\apredsetp$.
    Since $\angelicAssertionStrongerOf{\apredsetp'}{\apredsetp}$, there is an $\apredp'$ in $\apredsetp'$ with $\demonicAssertionStrongerOf{\apredp'}{\apredp}$.
From $\angelicHoareTripletHoldsSemantically{\apredset'}{\asketch}{\apredsetp'}$ we know, that there is an $\apred'$ in $\apredset'$ and a program $\aprog \in \angelConcFuncOf{\asketch}$ such that $\demonicHoareTripletHoldsSemantically{\apred'}{\aprog}{\apredp'}$ holds.
Through transitivity of the relation $\demonicAssertionStronger{}$, we also get that $\demonicHoareTripletHoldsSemantically{\apred'}{\aprog}{\apredp}$ holds.
Due to monotonicity of executions, and since $\angelicAssertionStrongerOf{\apredset}{\apredset'}$ holds, there is an $\apred$ in $\apredset$ for which $\demonicHoareTripletHoldsSemantically{\apred}{\aprog}{\apredp}$ is true.
That means, $\angelicHoareTripletHoldsSemantically{\apredset}{\asketch}{\apredsetp}$ holds.

In the next case, we have $\angelicHoareTripletHolds{\apredset}{\choiceOf{\asketch_1}{\asketch_2}}{\apredsetp}$ through the rule \ruleLabel{DEM}.
Therefore, we know $\apredset = \set{\apred}$ and the realizability triples $\angelicHoareTripletHolds{\apredset}{\asketch_1}{\apredsetp}$ and $\angelicHoareTripletHolds{\apredset}{\asketch_2}{\apredsetp}$ hold.
From the induction hypothesis we know that $\angelicHoareTripletHoldsSemantically{\apredset}{\asketch_1}{\apredsetp}$ and $\angelicHoareTripletHoldsSemantically{\apredset}{\asketch_2}{\apredsetp}$ hold.
Now, let $\apredp$ be an element of $\apredsetp$. 
Then, there exists $\aprog_1 \in \angelConcFuncOf{\asketch_1}$ for which $\demonicHoareTripletHoldsSemantically{\apred}{\aprog_1}{\apredp}$ is true.
    Analogously, there exists $\aprog_2 \in \angelConcFuncOf{\asketch_2}$ for which $\demonicHoareTripletHoldsSemantically{\apred}{\aprog_2}{\apredp}$ is true.
We now show that $\demonicHoareTripletHoldsSemantically{\apred}{\choiceOf{\aprog_1}{\aprog_2}}{\apredp}$ is true.
This directly implies that
    $\angelicHoareTripletHoldsSemantically{\apredset}{\choiceOf{\asketch_1}{\asketch_2}}{\apredsetp}$
holds.
Let $\anexec$ be an execution of $\choiceOf{\aprog_1}{\aprog_2}$.
Thus, it is either an execution of $\aprog_1$ or $\aprog_2$.
W.l.o.g.\ let $\anexec$ be in $\demonConcFuncOf{\aprog_1}$.
Since $\demonicHoareTripletHoldsSemantically{\apred}{\aprog_1}{\apredp}$ is true, we also get that $\progSemFuncOf{\code{ex}}{\apred}$ is more precise than $\apredp$.
Therefore, the Hoare triple $\demonicHoareTripletHoldsSemantically{\apred}{\choiceOf{\aprog_1}{\aprog_2}}{\apredp}$ is true and thus 
$\angelicHoareTripletHoldsSemantically{\apredset}{\choiceOf{\asketch_1}{\asketch_2}}{\apredsetp}$ holds.

In the next case, we have $\angelicHoareTripletHolds{\apredset}{\asketch^*}{\apredset}$ through the rule \ruleLabel{LOOP}.
Thus we know the realizability triple $\angelicHoareTripletHolds{\apredset}{\asketch}{\apredset}$ holds and $\apredset = \set{\apred}$.
The induction hypothesis yields $\angelicHoareTripletHoldsSemantically{\apredset}{\asketch}{\apredset}$.
From this we know that there is a $\aprog \in \angelConcFuncOf{\asketch}$ with $\demonicHoareTripletHoldsSemantically{\apred}{\aprog}{\apred}$.
We now show that $\demonicHoareTripletHoldsSemantically{\apred}{\aprog^*}{\apred}$ is true.
For this, we inductively show $\demonicHoareTripletHoldsSemantically{\apred}{\aprog^i}{\apred}$ for every $i \in \mathbb{N}$.
Since $\aprog^0$ is $\code{skip}$, $\demonicHoareTripletHoldsSemantically{\apred}{\aprog^0}{\apred}$ immediately follows.
In the induction step, we show $\demonicHoareTripletHoldsSemantically{\apred}{\aprog^i\code{;prog}}{\apred}$.
From the hypothesis, we know that $\demonicHoareTripletHoldsSemantically{\apred}{\aprog^i}{\apred}$ and $\demonicHoareTripletHoldsSemantically{\apred}{\aprog}{\apred}$ hold.
Referring to the case where we have proven sequences in this proof, this implies that $\demonicHoareTripletHoldsSemantically{\apred}{\aprog^{i+1}}{\apred}$ is true.
Thus, we know that $\demonicHoareTripletHoldsSemantically{\apred}{\aprog^*}{\apred}$ holds.
Therefore, $\angelicHoareTripletHoldsSemantically{\apredset}{\asketch^*}{\apredset}$ is true.

In the next case, we have $\angelicHoareTripletHolds{\apredset}{\anonterm}{\apredsetp}$ through the rule \ruleLabel{ANG}.
From this, we know $\angelicHoareTripletHolds{\apredset}{\code{rhs}}{\apredsetp}$ holds with $\anonterm \code{::= rhs} \mid \dots$.
From the induction hypothesis we get $\angelicHoareTripletHoldsSemantically{\apredset}{\code{rhs}}{\apredsetp}$ \ .
Since $\angelConcFuncOf{\code{rhs}}$ is a subset of $\angelConcFuncOf{\anonterm}$, we get $\angelicHoareTripletHoldsSemantically{\apredset}{\anonterm}{\apredsetp}$ immediately.

In the last case, we have $\angelicHoareTripletHolds{\apredset_1 \cup \apredset_2}{\asketch}{\apredsetp_1 \cup \apredsetp_2}$ through the rule \ruleLabel{GATHER}.
From this, we know $\angelicHoareTripletHolds{\apredset_1}{\asketch}{\apredsetp_1}$ and $\angelicHoareTripletHolds{\apredset_2}{\asketch}{\apredsetp_2}$ hold.
Let $\apredp$ be an element of $\apredsetp_1 \cup \apredsetp_2$.
W.l.o.g.\ let $\apredp$ be an element of $\apredsetp_1$.
From the induction hypothesis we know that there is an $\apred$ in $\apredset_1 \subseteq \apredset_1 \cup \apredset_2$ and a program $\aprog \in \angelConcFuncOf{\asketch}$ with $\demonicHoareTripletHoldsSemantically{\apred}{\aprog}{\apredp}$.
Thus, we get $\angelicHoareTripletHoldsSemantically{\apredset_1 \cup \apredset_2}{\asketch}{\apredsetp_1 \cup \apredsetp_2}$.
This concludes the soundness proof.

We proceed with the proof for completeness.
In the first case, the postcondition is empty.
We have $\angelicHoareTripletHoldsSemantically{\apredset}{\asketch}{\emptyset}$.
Using the rule \ruleLabel{EMPTY}, we get 
$\angelicHoareTripletHolds{\apredset}{\asketch}{\emptyset}$.

In the second case, the postcondition is not empty.
We have $\angelicHoareTripletHoldsSemantically{\apredset}{\asketch}{\apredsetp}$.
Thus, for every $\apredp \in \apredsetp$ there is a program $\aprog \in \angelConcFuncOf{\asketch}$ and a predicate $\apred \in \apredset$ for which
$\demonicHoareTripletHoldsSemantically{\apred}{\aprog}{\apredp}$.
Since Hoare logic is complete, we have
$\demonicHoareTripletHolds{\apred}{\aprog}{\apredp}$.
We can mimic this proof in realizability logic for $\asketch$ by using rule \ruleLabel{ANG} mimicking the derivation of $\aprog$ from $\asketch$.
Thus, we have $\angelicHoareTripletHolds{\apred}{\asketch}{\apredp}$.
Since $\apredsetp$ is finite, we only need to do the proof for finitely many programs.
These proofs are then combined using the rule \ruleLabel{GATHER} multiple times to derive 
$\angelicHoareTripletHolds{\apredset'}{\asketch}{\apredsetp}$ for the subset $\apredset'$ of $\apredset$ containing the used predicates $\apred$ of $\apredset$.
Then, using rule \ruleLabel{CSQ}, we get 
$\angelicHoareTripletHolds{\apredset}{\asketch}{\apredsetp}$.
\end{proof}

%% file: proofsChapter3.tex
\newcommand{\posound}{\vdash_p}
\subsection{Proofs for \Cref{ch:proofsystem}}
To properly define the notation of a valid proof outline, we state a 
a proof system for that ranges over proof outlines in \Cref{def:plsProofs}.
When we write $\posound \anonterm(\apo)$ we mean 
$\posound \hoaretriplet{\apredset}{\anonterm(\apo)}{\apredsetp}$ where $\apredset$ resp.\ $\apredsetp$ are the union of all pre- resp.\ postconditions of the subproofs in $\apo$.
The function $\proofToProgFunc{}$ extracts the program sketch out of the proof outline:

\begin{figure}
    \caption{Definition of the Program Logic for Synthesis over proof outlines. Extensions of traditional Hoare logic are \makeBlue{blue}.}
    \label{def:plsProofs}
	\begin{mathpar}
        \footnotesize
		\inferrule[\ruleLabelSmall{PCOM}]{
\demonicAssertionStrongerOf{\progSemFuncOf{\code{com}}{r}}{s}
		}{
\angelicProofHoareTripletHolds{\set{r}}{\code{com}}{\set{s}}
}
\and
		\inferrule[\ruleLabelSmall{PSEQ}]{
\angelicProofHoareTripletHolds{\apredset}{\apo_1}{\apredsetp}
\\
\angelicProofHoareTripletHolds{\apredsetp}{\apo_2}{\apredsetpp}
		}{
            \posound
\hoaretriplet{\apredset}{\apo_1}{\apredsetp}
\code{;}
\hoaretriplet{\apredsetp}{\apo_2}{\apredsetpp}
}
\and
\inferrule[\ruleLabelSmall{PLOOP}]{
\angelicProofHoareTripletHolds{\apredset}{\apo}{\apredset}
\\
\makeBlue{\apredset = \set{\apred}}
		}{
\angelicProofHoareTripletHolds{\apredset}{(\hoaretriplet{\apredset}{\apo}{\apredset})^*}{\apredset}
}
\and
\inferrule[\ruleLabelSmall{PDEM}]{
\angelicProofHoareTripletHolds{\apredset}{\apo_1}{\apredsetp}
\\
\angelicProofHoareTripletHolds{\apredset}{\apo_2}{\apredsetp}
\\
\makeBlue{\apredset = \set{r}}
		}{
\vdash_p \choiceOf{\hoaretriplet{\apredset}{\apo_1}{\apredsetp}}{\hoaretriplet{\apredset}{\apo_2}{\apredsetp}}
}
\and
\inferrule[\ruleLabelSmall{PCSQ}]{
\angelicProofHoareTripletHolds{\angelicAssertionFont{R'}}{\apo}{\angelicAssertionFont{S'}}
\\
\angelicAssertionStrongerOf{\apredset}{\apredset'}
\\
\angelicAssertionStrongerOf{\apredsetp'}{\apredsetp}
		}{
\angelicProofHoareTripletHolds{\apredset}{\apo}{\apredsetp}
}
\and
\makeBlue{
\inferrule[\ruleLabelSmall{PANG}]{
\angelicProofHoareTripletHolds{\apredset}{\code{rhs}}{\apredsetp}
\\
\anonterm ::= \angelicChoiceOf{\code{rhs}}{\dots}
		}{
\posound \anonterm(\hoaretriplet{\apredset}{\code{rhs}}{\apredsetp})}
}
\and
\makeBlue{
\inferrule[\ruleLabelSmall{PEMPTY}]{
		}{
\angelicProofHoareTripletHolds{\apredset}{\apo}{\emptyset}
}
}
\and
\makeBlue{
\inferrule[\ruleLabelSmall{PGATHER}]{
\angelicProofHoareTripletHolds{\apredset_1}{\apo_1}{\apredsetp_1}
\\
\angelicProofHoareTripletHolds{\apredset_2}{\apo_2}{\apredsetp_2}
\\
\proofToProgFuncOf{\apo_1} = \proofToProgFuncOf{\apo_2}
		}{
\angelicProofHoareTripletHolds{\apredset_1 \cup \apredset_2}{\gatherFuncOf{\apo_1}{\apo_2}}{\apredsetp_1 \cup \apredsetp_2}
}
}
    \end{mathpar}
\end{figure}

\begin{definition}
    The function $\gatherFunc{}$ combines two proofs over the same sketch.
    \begin{equation*}
        \funcDef{\gatherFunc{}}{(\angelProgramProofLang{} \times \angelProgramProofLang{}) \partialFuncArrow{} \angelProgramProofLang{}} 
    \end{equation*}
    \begin{align*}
        \gatherFuncOf{\hoaretriplet{\apredset_1}{\code{com}}{\apredsetp_1}}{\hoaretriplet{\apredset_2}{\code{com}}{\apredsetp_2}} = &
        \, \hoaretriplet{\apredset_1 \cup \apredset_2}{\code{com}}{\apredsetp_1 \cup \apredsetp_2} \\
        \gatherFuncOf{\hoaretriplet{\apredset_1}{\code{p1}^*}{\apredsetp_1}}
        {\hoaretriplet{\apredset_2}{\code{p2}^*}{\apredsetp_2}} = &
        \, \hoaretriplet{\apredset_1 \cup \apredset_2}{
            \gatherFuncOf{\code{p1}}{\code{p2}}^*
        }{\apredsetp_1 \cup \apredsetp_2} \\
        \gatherFuncOf{\apo_1\code{;}\apo_2}
        {\apop_1\code{;}\apop_2} = &
        \, \gatherFuncOf{\apo_1}{\apop_1}\code{;}
        \gatherFuncOf{\apo_2}{\apop_2} \\
        \gatherFuncOf{\choiceOf{\apo_1}{\apo_2}}
        {\choiceOf{\apop_1}{\apop_2}} = &
        \, \choiceOf{\gatherFuncOf{\apo_1}{\apop_1}}{\gatherFuncOf{\apo_2}{\apop_2}}
        \\
        \gatherFuncOf{\apo_1 \bnf \apo_2}{\apop_1 \bnf \apop_2} = & \,
        \gatherFuncOf{\apo_1}{\apop_1} \bnf \gatherFuncOf{\apo_2}{\apop_2}
        \\
        \gatherFuncOf{\anonterm(\apo_1 \bnf \apo_2)}{\anonterm(\apop_1 \bnf \apop_2)} = & \,
        \anonterm(\gatherFuncOf{\apo_1}{\apop_1} \bnf \apo_2 \bnf \apop_2)
    \end{align*}
    In the last case, $\apo_1$ and $\apop_1$ are the productions where the right hand side in both proof outlines are the same.
    The proof outlines $\apo_2$ and $\apop_2$ reason over different sketches and their proof outlines therefore cannot be combined.

    In order for the function $\gatherFuncOf{\apo_1}{\apo_2}$ to be defined, it is required that $\apo_1$ and $\apo_2$ are proofs over the same sketch, i.e.\ $\proofToProgFuncOf{\apo_1} = \proofToProgFuncOf{\apo_2}$.
\end{definition}

We need the following theorem:
\begin{theorem}[Equivalence]\label{th:proofsystemEquivalence}
    The following implication holds for any proof $\apo$ of $\angelProgramProofLang{}$ and any angels $\apredset$ and $\apredsetp$:
    \begin{equation*}
        \angelicProofHoareTripletHolds{\apredset}{\apo}{\apredsetp} \implies 
        \angelicHoareTripletHolds{\apredset}{\proofToProgFuncOf{\apo}}{\apredsetp} 
        \ .
    \end{equation*}

    The following implication holds for any sketch $\asketch$ of $\angelicProgramLang{}$ and any angels $\apredset$ and $\apredsetp$:
    \begin{equation*}
        \angelicHoareTripletHolds{\apredset}{\asketch}{\apredsetp}
        \implies
        \exists \, \apo \in \angelProgramProofLang{}: 
        \proofToProgFuncOf{\apo} = \asketch \,
        \wedge \angelicProofHoareTripletHolds{\apredset}{\apo}{\apredsetp}
        \ .
    \end{equation*}
\end{theorem}

\subsubsection*{Proof of \Cref{th:proofsystemEquivalence}}
\begin{proof}[\unskip\nopunct]
We start with the first implication.
We show the implication by structural induction over the proof tree.

\baseCase{}: 
We have $\angelicProofHoareTripletHolds{\set{\apred}}{\code{com}}{\set{\apredp}}$ through \ruleLabel{PCOM}.
Thus, we know $\demonicAssertionStrongerOf{\progSemFuncOf{\code{com}}{\apred}}{\apredp}$ is true.
With rule \ruleLabel{COM} we directly get $\angelicHoareTripletHolds{\set{\apred}}{\code{com}}{\set{\apredp}}$.

\inductionStep{}
Since all induction steps are only applying the induction hypothesis and rebuilding the proof in the other proof system, we demonstrate it on one case only.
We have the proof outline $\angelicProofHoareTripletHolds{\apredset}{\apo_1 \code{;} \apo_2}{\apredsetpp}$ through \ruleLabel{PSEQ}.
With the preconditions and the induction hypothesis we get 
$\angelicHoareTripletHolds{R}{\asketch_1}{S}$ and 
$\angelicHoareTripletHolds{S}{\asketch_2}{T}$ where the sketches are of the corresponding proof.
Using the rule \ruleLabel{SEQ}, we get
$\angelicHoareTripletHolds{R}{\asketch_1\code{;}\asketch_2}{T}$

We continue with the second implication.
Again, we show the implication by structural induction over the proof tree.

\baseCase{}:
We have $\angelicHoareTripletHolds{\set{\apred}}{\code{com}}{\set{\apredp}}$ through \ruleLabel{COM}.
Thus, we know $\demonicAssertionStrongerOf{\progSemFuncOf{\code{com}}{\apred}}{\apredp}$ is true.
With rule \ruleLabel{PCOM} we directly get $\angelicProofHoareTripletHolds{\set{\apred}}{\code{com}}{\set{\apredp}}$.

\inductionStep{}
Since all induction steps are only applying the induction hypothesis and rebuilding the proof in the other proof system, we demonstrate it on one case only.
We have the realizability triple $\angelicHoareTripletHolds{\apredset}{\asketch_1\code{;}\asketch_2}{\apredsetpp}$ through \ruleLabel{SEQ}.
With the preconditions and the induction hypothesis we get two proofs $\apo_1$ and $\apo_2$ for which the following holds:
$\angelicProofHoareTripletHolds{\apredset}{\apo_1}{\apredsetp}$ and $\angelicProofHoareTripletHolds{\apredsetp}{\apo_2}{\apredsetpp}$.
Furthermore, $\proofToProgFuncOf{\apo_i} = \asketch_i$.
With rule \ruleLabel{PSEQ} we get $\angelicProofHoareTripletHolds{R}{\apo_1\code{;}\apo_2}{T}$.
\end{proof}

We show several properties of a valid proof outline:
When we write $\posound \apo$ with $\apo = \apo_1 \bnf \ldots \bnf \apo_n$ we mean $\posound \apo_1 \wedge \ldots \wedge \posound \apo_n$.
\begin{lemma}\label{lem:proofProperties}
    In a sketch proof outline, all realizability triples hold, i.e.\
    the following implications are true:
    \begin{equation*}
        \posound \apo_1\code{;}\apo_2
        \implies 
        \angelicProofHoareTripletHolds{\apredset}{\apo_1}{\apredsetp} \wedge
        \angelicProofHoareTripletHolds{\apredsetp'}{\apo_2}{\apredsetpp}
        \wedge \apredsetp = \apredsetp'
    \end{equation*}
    \begin{equation*}
        \posound \choiceOf{\apo_1}{\apo_2}
        \implies
         \angelicProofHoareTripletHolds{\apredset}{\apo_1}{\apredsetp} \wedge
        \angelicProofHoareTripletHolds{\apredset'}{\apo_2}{\apredsetp'}
        \wedge 
        \apredset = \apredset'
        \wedge
        \apredsetp = \apredsetp'
    \end{equation*}
    \begin{equation*}
        \posound \anonterm(\apo)
        \implies 
        \posound \apo
    \end{equation*}
    \begin{equation*}
        \angelicProofHoareTripletHolds{\apredset}{(\hoaretriplet{I}{\apo}{I'})^*}{\apredsetp}
        \implies 
        \angelicProofHoareTripletHolds{I}{\apo}{I'} 
        \wedge I = I' \,
        \wedge
        \angelicAssertionStrongerOf{\apredset}{I}
        \wedge \angelicAssertionStrongerOf{I}{\apredsetp}
        \ .
    \end{equation*}
\end{lemma}

\begin{proof}
We use structural induction to prove the implications.
We proof one implication at a time.

For the first one, the only rules that need to be considered are \ruleLabel{PSEQ}, \ruleLabel{PCSQ} and \ruleLabel{PGATHER}.
The base case holds trivially.
Rule \ruleLabel{PSEQ} immediately yields the sought after result.
Through transitivity of $\angelicAssertionStronger{}$, \ruleLabel{PCSQ} also yields the implication after applying the induction hypothesis.
When rule \ruleLabel{PGATHER} is applied, the proof has the following shape:
\begin{equation*}
    \posound 
    \hoaretriplet{\apredset \cup \apredset'}{\gatherFuncOf{\apo_1}{\apop_1}}{\apredsetp \cup \apredsetp'} 
    \code{;}
    \hoaretriplet{\apredsetp \cup \apredsetp'}{\gatherFuncOf{\apo_2}{\apop_2}}{\apredsetpp \cup \apredsetpp'}
\end{equation*}
We can use the gather rule to show $\posound \gatherFuncOf{\apo_i}{\apo_i'}$ because the individual proofs hold according to the induction hypothesis.
That fact that the intermediary assertions match is a result of the induction hypothesis.

The other cases all follow the same pattern:
The proofs for the rule corresponding to the program (here \ruleLabel{PSEQ} for a program sequence) and rule \ruleLabel{PCSQ} hold trivially.
For rule \ruleLabel{PGATHER}, the rule itself has to be applied to the result of the induction hypothesis, as seen here.
\end{proof}

\subsubsection*{Proof of \Cref{th:decisionCalculusSoundenss}}
\begin{proof}[\unskip\nopunct]
With the newly introduced notation, the theorem is the following:
\begin{theorem}[Soundness]
$\posound \apo$ and $\angelicProofRewriteHolds{\apo}{\apop}$ together imply 
$\posound \apop$. 
\end{theorem}
We conduct the proof using structural induction over the proof tree in realization logic.

\baseCase{}
We know that $\angelicProofHoareTripletHolds{R}{\code{com}}{S}$ and $\angelicProofRewriteHolds{\hoaretriplet{R}{\code{com}}{S}}{\hoaretriplet{\set{r}}{\code{com}}{\set{s}}}$
hold.
From the rule \ruleLabel{RCOM} of the realization logic we know that $\demonicAssertionStrongerOf{\progSemFuncOf{\code{com}}{r}}{s}$ is true.
This directly implies that $\angelicProofHoareTripletHolds{\set{r}}{\code{com}}{\set{s}}$ holds.

\inductionStep{}
Consider the rule  \ruleLabel{RSELECT}.
The implication immediately holds because the nested proofs are sound per \Cref{lem:proofProperties}.

Consider the rule \ruleLabel{RAC}.
Again, since the nested proofs are sound, we get 
$\posound \apo_1$ and $\posound \apo$ and $\posound \apo_2$.
Applying the induction hypothesis on the precondition of the rule yields the triple
$\posound \apop$
Applying rule \ruleLabel{PANG} on each subproof, yields
$\posound N(\apo_1)$ and $\posound N(\apop)$ and $\posound N(\apo_2)$.
Applying rule \ruleLabel{PGATHER} multiple times yields the sought after result:
    $\posound N(\apo_1 \bnf \apop \bnf \apo_2)$ .

Consider rule \ruleLabel{RSEQL}. 
Since the nested proofs are sound, we get $\angelicProofHoareTripletHolds{\apredset}{\apo_1}{\apredsetp}$ and also $\posound \apo_2$.
The induction hypothesis yields the realizability triple
$\angelicProofHoareTripletHolds{\weakerof{\apredset}}{\apop_1}{\apredsetp}$.
Using rule \ruleLabel{PSEQ}, we get 
$\posound \hoaretriplet{\weakerof{\apredset}}{\apop_1}{\apredsetp} \code{;}\apo_2$.

Consider rule \ruleLabel{RSEQR}.
Again, by soundness of the subproofs we get $\angelicProofHoareTripletHolds{\apredset}{\apo_1}{\apredsetp}$ and also 
$\angelicProofHoareTripletHolds{\apredsetp}{\apo_2}{\apredsetpp}$.
Applying the induction hypothesis we get $\angelicProofHoareTripletHolds{\weakerof{\apredsetp}}{\apop_2}{\weakerof{\apredsetpp}}$.
Since $\angelicAssertionStrongerOf{\apredsetp}{\weakerof{\apredsetp}}$, the infer rule can be applied to the first realizability triple resulting in 
$\angelicProofHoareTripletHolds{\apredset}{\apo_1}{\weakerof{\apredsetp}}$
Then applying rule \ruleLabel{PSEQ} yields: 
$\angelicProofHoareTripletHolds{\apredset}{\apo_1\code{;}\apo_2}{\weakerof{\apredsetpp}}$.

Consider rule \ruleLabel{RDEM}.
Again, through the soundness of the sub proofs (\Cref{lem:proofProperties}), we get $\angelicProofHoareTripletHolds{\apredset}{\apo_1}{\apredsetp}$ and $\angelicProofHoareTripletHolds{\apredset}{\apo_2}{\apredset_p}$.
The induction hypothesis yields  $\angelicProofHoareTripletHolds{\weakerof{\apredset}}{\apop_i}{\weakerof{\apredsetp}}$ with $\weakerof{\apredset} = \set{\apred}$.
Applying rule \ruleLabel{PDEM} results in 
\begin{equation*}
    \angelicProofHoareTripletHolds{\weakerof{\apredset}}{\choiceOf{\apo_1}{\apo_2}}{\weakerof{\apredsetp}}
    \ .
\end{equation*}

Consider rule \ruleLabel{RLOOP}.
Through the soundness of the sub proofs (\Cref{lem:proofProperties}), we get the realizability triple $\angelicProofHoareTripletHolds{I}{\apo}{I}$ .
Applying the induction hypothesis, we get the realizability triple $\angelicProofHoareTripletHolds{\weakerof{I}}{\apop}{\weakerof{I}}$ with $\weakerof{I} = \set{i}$.
Applying rule \ruleLabel{PLOOP} yields
the realizability triple
\begin{equation*}
    \angelicProofHoareTripletHolds{\weakerof{I}}{\kleeneof{\apop}}{\weakerof{I}}
    \ .
\end{equation*}

Consider rule \ruleLabel{RTRANS}.
The induction hypothesis yields soundness of $\posound \apo_2$.
Applying the induction hypothesis again yields the sought after result: $\posound \apo_3$.

Consider rule \ruleLabel{RCSQ}.
The induction hypothesis yields
$\angelicProofHoareTripletHolds{\wweakerof{\apredset}}{\apop}{\weakerof{\apredsetp}}$.
Using rule \ruleLabel{PCSQ}, we get 
$\angelicProofHoareTripletHolds{\weakerof{\apredset}}{\apop}{\wweakerof{\apredsetp}}$.

Lastly, consider rule \ruleLabel{RGATHER}.
Applying the induction hypothesis yields $\angelicProofHoareTripletHolds{\apredset_1}{\apop_1}{\apredsetp_1}$ and $\angelicProofHoareTripletHolds{\apredset_2}{\apop_2}{\apredsetp_2}$ with $\proofToProgFuncOf{\apop_1} = \proofToProgFuncOf{\apop_2}$.
Applying the rule \ruleLabel{PGATHER}, we get 
\begin{equation*}
    \angelicProofHoareTripletHolds{\apredset_1 \cup \apredset_2}{\gatherFuncOf{\apop_1}{\apop_2}}{\apredsetp_1 \cup \apredsetp_2}
\end{equation*}
\end{proof}

\subsubsection*{Proof of \Cref{th:angelicDecisionCalcComplete}}
Again, we state theorem with the new notation:
\begin{theorem}
$\posound \apo$ and $\posound \apop$ and $\proofStrongerOf{\apo}{\apop}$ together imply $\angelicProofRewriteHolds{\apo}{\apop}$. 
\end{theorem}    
\begin{proof}
For this proof, we conduct a structural induction over the syntactic structure of the proof $\apo$.
Then, some cases require their own additional structural induction over the proof tree of $\apop$ in the proof outline calculus.
In the following, we also have $\angelicAssertionStrongerOf{X}{\weakerof{X}}$ for any selections $X$ and $\weakerof{X}$.

\baseCase{}
$\apo = \code{com}$. 
Then, $\apop = \code{com}$.
We show $\angelicProofRewriteHolds{\hoaretriplet{\apredset}{\apo}{\apredsetp}}{\hoaretriplet{\weakerof{\apredset}}{\apop}{\weakerof{\apredsetp}}}$ by induction over the proof tree of $\angelicProofHoareTripletHolds{\weakerof{\apredset}}{\apop}{\weakerof{\apredsetp}}$.
In the base case, we have the realizability triple $\angelicProofHoareTripletHolds{\apred}{\code{com}}{\apredp}$.
Thus, we know that $\demonicAssertionStrongerOf{\progSemFuncOf{\code{com}}{\apred}}{\apredp}$
holds.
Therefore, $\angelicProofRewriteHolds{\hoaretriplet{\apredset}{\code{com}}{\apredsetp}}{\hoaretriplet{\set{\apred}}{\code{com}}{\set{\apredp}}}$ holds.
In the induction step, there are two cases. Either the rule \ruleLabel{PGATHER} was used, or the rule \ruleLabel{PCSQ} was used.
If \ruleLabel{PGATHER} was used, we have $\angelicProofHoareTripletHolds{\apredset_1 \cup \apredset_2}{\code{com}}{\apredsetp_1 \cup \apredsetp_2}$.
From this, we know that $\angelicProofHoareTripletHolds{\apredset_i}{\code{com}}{\apredsetp_i}$ hold.
Applying the induction hypothesis yields $\angelicProofRewriteHolds{\hoaretriplet{\apredset}{\code{com}}{\apredsetp}}{\hoaretriplet{\apredset_i}{\code{com}}{\apredsetp_i}}$.
Using the rule \ruleLabel{RGA}, we get $\angelicProofRewriteHolds{\hoaretriplet{\apredset}{\code{com}}{\apredsetp}}{\hoaretriplet{\cup_i \apredset_i}{\code{com}}{\cup_i \apredset_i}}$.
If \ruleLabel{PCSQ} was used, we have $\angelicProofHoareTripletHolds{\weakerof{\apredset}}{\code{com}}{\wweakerof{\apredsetp}}$.
From the precondition, we know $\angelicProofHoareTripletHolds{\wweakerof{\apredset}}{\code{com}}{\weakerof{\apredsetp}}$ holds.
Using the induction hypothesis, we get $\angelicProofRewriteHolds{\hoaretriplet{\apredset}{\code{com}}{\apredsetp}}{\hoaretriplet{\wweakerof{\apredset}}{\code{com}}{\weakerof{\apredsetp}}}$.
Applying rule \ruleLabel{RCSQ} yields $\angelicProofRewriteHolds{\hoaretriplet{\apredset}{\code{com}}{\apredsetp}}{\hoaretriplet{\weakerof{\apredset}}{\code{com}}{\wweakerof{\apredsetp}}}$.

\inductionStep{}

If the proof is a sequence, i.e.\
$
    \hoaretriplet{\apredset}{\apo}{\apredsetpp} = 
        \hoaretriplet{\apredset}{\apo_1}{\apredsetp}
        \code{;}
        \hoaretriplet{\apredsetp}{\apo_2}{\apredsetpp}
    \ ,
$
then the other proof is
$
    \hoaretriplet{\weakerof{\apredset}}{\apop}{\weakerof{\apredsetpp}} = 
        \hoaretriplet{\weakerof{\apredset}}{\apop_1}{\weakerof{\apredsetp}}
        \code{;}
        \hoaretriplet{\weakerof{\apredsetp}}{\apop_2}{\weakerof{\apredsetpp}}
    \ .
$

Thus we know that, the realizability triples
$\angelicProofHoareTripletHolds{\apredsetp}{\apo_2}{\apredsetpp}$ and $\angelicProofHoareTripletHolds{\weakerof{\apredsetp}}{\apop_2}{\weakerof{\apredsetpp}}$ hold.
Applying the induction hypothesis, we get $\angelicProofRewriteHolds{\apo_2}{\apop_2}$.
Applying the rule \ruleLabel{RSEQR} yields
\begin{equation*}
    \angelicProofRewriteHolds{
            \hoaretriplet{\apredset}{\apo_1}{\apredsetp} 
            \code{;}
            \hoaretriplet{\apredsetp}{\apo_2}{\apredsetpp}
    }
    {       
            \hoaretriplet{\apredset}{\apo_1}{\weakerof{\apredsetp}} 
            \code{;}
            \hoaretriplet{\weakerof{\apredsetp}}{\apop_2}{\weakerof{\apredsetpp}}
    }
            \ .
\end{equation*}

Because $\angelicProofHoareTripletHolds{\apredset}{\apo_1}{\apredsetp}$ holds, the realizability triple $\angelicProofHoareTripletHolds{\apredset}{\apo_1}{\weakerof{\apredsetp}}$ also holds due to the consequence rule.
We also have that $\angelicProofHoareTripletHolds{\weakerof{\apredset}}{\apop_1}{\weakerof{\apredsetp}}$ is true.
Applying the induction hypothesis, we get $\angelicProofRewriteHolds{\hoaretriplet{\apredset}{\apo_1}{\weakerof{\apredsetp}}}{\hoaretriplet{\weakerof{\apredset}}{\apop_1}{\weakerof{\apredsetp}}}$.
Using rule \ruleLabel{RSEQL} yields 
\begin{equation*}
    \angelicProofRewriteHolds{
            \hoaretriplet{\apredset}{\apop_1}{\weakerof{\apredsetp}} 
            \code{;}
            \hoaretriplet{\weakerof{\apredsetp}}{\apop_2}{\weakerof{\apredsetpp}}
    }
    {
            \hoaretriplet{\weakerof{\apredset}}{\apop_1}{\weakerof{\apredsetp}} 
            \code{;}
            \hoaretriplet{\weakerof{\apredsetp}}{\apop_2}{\weakerof{\apredsetpp}}
    }
            \ .
\end{equation*}
Combined with the rewrite from before and rule \ruleLabel{RTRANS}, we have 
\begin{equation*}
    \angelicProofRewriteHolds{
            \hoaretriplet{\apredset}{\apo_1}{\apredsetp} 
            \code{;}
            \hoaretriplet{\apredsetp}{\apo_2}{\apredsetpp}
            }
    {
            \hoaretriplet{\weakerof{\apredset}}{\apop_1}{\weakerof{\apredsetp}} 
            \code{;}
            \hoaretriplet{\weakerof{\apredsetp}}{\apop_2}{\weakerof{\apredsetpp}}
    }
            \ .
\end{equation*}

In the next case, the proof is a demonic choice, i.e.\ 
\begin{equation*}
    \hoaretriplet{\apredset}{\apo}{\apredsetp} = 
    \choiceOf{
        \hoaretriplet{\apredset}{\apo_1}{\apredsetp} 
    }{
        \hoaretriplet{\apredset}{\apo_2}{\apredsetp} 
    }
    \ .
\end{equation*}
Then, the other proof is 
\begin{equation*}
    \hoaretriplet{\weakerof{\apredset}}{\apop}{\weakerof{\apredsetp}} = 
    \choiceOf{
        \hoaretriplet{\weakerof{\apredset}}{\apop_1}{\weakerof{\apredsetp}} 
    }{
        \hoaretriplet{\weakerof{\apredset}}{\apop_2}{\weakerof{\apredsetp}} 
    }
    \ .
\end{equation*}
We proceed by induction over the proof tree of 
\begin{equation*}
    \choiceOf{
        \hoaretriplet{\weakerof{\apredset}}{\apop_1}{\weakerof{\apredsetp}} 
    }{
        \hoaretriplet{\weakerof{\apredset}}{\apop_2}{\weakerof{\apredsetp}} 
    }
    \ .
\end{equation*}
If the proof was done with the rule \ruleLabel{PDEM}, it has the following shape: 
\begin{equation*}
    \choiceOf{
        \hoaretriplet{\weakerof{\apredset}}{\apop_1}{\weakerof{\apredsetp}} 
    }{
        \hoaretriplet{\weakerof{\apredset}}{\apop_2}{\weakerof{\apredsetp}} 
    }
\end{equation*}
with $\weakerof{\apredset} = \set{\apred}$.
With the overall induction hypothesis and because the subproofs are sound, we know that 
\begin{equation*}
\angelicProofRewriteHolds{\hoaretriplet{\apredset}{\apo_1}{\apredsetp}}{\hoaretriplet{\weakerof{\apredset}}{\apop_1}{\weakerof{\apredsetp}}}
\text{ and } 
\angelicProofRewriteHolds{\hoaretriplet{\apredset}{\apo_2}{\apredsetp}}{\hoaretriplet{\weakerof{\apredset}}{\apop_2}{\weakerof{\apredsetp}}}
\end{equation*}
hold.
Using the rule \ruleLabel{RDEM}, we get
\begin{equation*}
    \angelicProofRewriteHolds{
    \choiceOf{
        \hoaretriplet{\apredset}{\apo_1}{\apredsetp} 
    }{
        \hoaretriplet{\apredset}{\apo_2}{\apredsetp} 
    }
    }{
    \choiceOf{
        \hoaretriplet{\weakerof{\apredset}}{\apop_1}{\weakerof{\apredsetp}} 
    }{
        \hoaretriplet{\weakerof{\apredset}}{\apop_2}{\weakerof{\apredsetp}} 
    }
    }
    \ .
\end{equation*}
Next, we, again, have two possibilities.
The proof was conducted using \ruleLabel{PGATHER} or \ruleLabel{PCSQ}.
Starting with rule \ruleLabel{PGATHER}, the proof has the following shape:
\begin{equation*}
    \posound
        \choiceOf{
            \hoaretriplet{\weakerof{\apredset_1} \cup \weakerof{\apredset_2}}{\apop_1}{\weakerof{\apredsetp_1} \cup \weakerof{\apredsetp_2}}
        }{
            \hoaretriplet{\weakerof{\apredset_1} \cup \weakerof{\apredset_2}}{\apop_2}{\weakerof{\apredsetp_1} \cup \weakerof{\apredsetp_2}}
        }
\end{equation*}
From preconditions of the gather rule we get 
\begin{equation*}
    \posound
        \choiceOf{
            \hoaretriplet{\weakerof{\apredset_1}}{\apop_1}{\weakerof{\apredsetp_1}}
        }{
            \hoaretriplet{\weakerof{\apredset_1}}{\apop_2}{\weakerof{\apredsetp_1}}
        }
\end{equation*}
\begin{equation*}
    \posound
        \choiceOf{
            \hoaretriplet{\weakerof{\apredset_2}}{\apo_1''}{\weakerof{\apredsetp_2}}
        }{
            \hoaretriplet{\weakerof{\apredset_2}}{\apo_2''}{\weakerof{\apredsetp_2}}
        }
\end{equation*}
and 
Applying the induction hypothesis then yields
\begin{equation*}
    \angelicProofRewriteHolds{
    \choiceOf{
        \hoaretriplet{\apredset}{\apo_1}{\apredsetp} 
    }{
        \hoaretriplet{\apredset}{\apo_2}{\apredsetp} 
    }
    }{
        \choiceOf{
            \hoaretriplet{\weakerof{\apredset_i}}{\apop_1}{\weakerof{\apredsetp_i}}
        }{
            \hoaretriplet{\weakerof{\apredset_i}}{\apop_2}{\weakerof{\apredsetp_i}}
        }
    }
    \ .   
\end{equation*}
Similarly for $\apo_1''$ and $\apo_2''$.
Using the rule \ruleLabel{RGATHER} we get 
\begin{equation*}
    \angelicProofRewriteHolds{
        \choiceOf{
        \hoaretriplet{\apredset}{\apo_1}{\apredsetp} 
    }{
        \hoaretriplet{\apredset}{\apo_2}{\apredsetp} 
    }    
    }{
        \choiceOf{
            \hoaretriplet{\cup_i \weakerof{\apredset_i}}{\gatherFuncOf{\apop_1}{\apo_1''}}{\cup_i \weakerof{\apredsetp_i}}
        }{
            \hoaretriplet{\cup_i \weakerof{\apredset_i}}{\gatherFuncOf{\apop_1}{\apo_1''}}{\cup_i \weakerof{\apredsetp_i}}
        }
    }
    \ .   
\end{equation*}
Closing with the rule \ruleLabel{PCSQ}, the proof has the following shape:
\begin{equation*}
        \choiceOf{
            \hoaretriplet{\weakerof{\apredset}}{\apop_1}{\wweakerof{\apredsetp}}
        }{
            \hoaretriplet{\weakerof{\apredset}}{\apop_2}{\wweakerof{\apredsetp}}
        }
    \ .
\end{equation*}
From the precondition of the infer rule, we have
\begin{equation*}
        \choiceOf{
            \hoaretriplet{\wweakerof{\apredset}}{\apop_1}{\weakerof{\apredsetp}}
        }{
            \hoaretriplet{\wweakerof{\apredset}}{\apop_2}{\weakerof{\apredsetp}}
        }
    \ .
\end{equation*}
Using the induction hypothesis, we get the following rewrite:
\begin{equation*}
    \angelicProofRewriteHolds{
        \choiceOf{
        \hoaretriplet{\apredset}{\apo_1}{\apredsetp} 
    }{
        \hoaretriplet{\apredset}{\apo_2}{\apredsetp} 
    }    
    }{
        \choiceOf{
            \hoaretriplet{\wweakerof{\apredset}}{\apop_1}{\weakerof{\apredsetp}}
        }{
            \hoaretriplet{\wweakerof{\apredset}}{\apop_2}{\weakerof{\apredsetp}}
        }
    }
    \ .
\end{equation*}
Now, the rule \ruleLabel{RCSQ} can be applied to get the sought after result:
\begin{equation*}
    \angelicProofRewriteHolds{
        \choiceOf{
        \hoaretriplet{\apredset}{\apo_1}{\apredsetp} 
    }{
        \hoaretriplet{\apredset}{\apo_2}{\apredsetp} 
    }    
    }{
        \choiceOf{
            \hoaretriplet{\weakerof{\apredset}}{\apop_1}{\wweakerof{\apredsetp}}
        }{
            \hoaretriplet{\weakerof{\apredset}}{\apop_2}{\wweakerof{\apredsetp}}
        }
    }
    \ .
\end{equation*}
This concludes the proof for demonic choices.

In the next case, the proof is a loop.
This case is analogous to demonic choices.

Finally, consider the proof is over a nonterminal.
So, the first proof is
\begin{equation*}
    \angelicProofHoareTripletHolds{\apredset}{
        \code{N}(
            \hoaretriplet{\apredset_1}{\apo_1}{\apredsetp_1}
            \bnf
            \ldots
            \bnf
            \hoaretriplet{\apredset_n}{\apo_n}{\apredsetp_n}
            \bnf
            \apo_r
        )
    }{\apredsetp}
    \ .
\end{equation*}
Then, the other proof can either also be a weaker proof over a nonterminal or it is only one production of the nonterminal.
We start with the first case.
The second proof is then 
\begin{equation*}
    \angelicProofHoareTripletHolds{\weakerof{\apredset}}{
        \code{N}(
            \hoaretriplet{\weakerof{\apredset_1}}{\apop_1}{\weakerof{\apredsetp_1}}
            \bnf
            \ldots
            \bnf
            \hoaretriplet{\weakerof{\apredset_n}}{\apop_n}{\weakerof{\apredsetp_n}}
        )
    }{\weakerof{\apredsetp}}
    \ .
\end{equation*}
We know that all the subproofs hold.
Applying the induction hypothesis then yields the following rewrites for all $i$:
\begin{equation*}
    \angelicProofRewriteHolds{
        \hoaretriplet{\apredset_i}{\apo_i}{\apredsetp_i}
    }{
        \hoaretriplet{\weakerof{\apredset_i}}{\apop_i}{\weakerof{\apredsetp_i}}
    }
    \ .
\end{equation*}
We apply the rule \ruleLabel{RANG} on the rewritten part several times.
Together with rule \ruleLabel{TRANS}, we combine the rewrites and get:
\begin{multline*}
    \angelicProofRewriteHolds{
    \hoaretriplet{\apredset}{
        \code{N}(
            \hoaretriplet{\apredset_1}{\apo_1}{\apredsetp_1}
            \angelicChoice{}
            \ldots
            \angelicChoice{}
            \hoaretriplet{\apredset_n}{\apo_n}{\apredsetp_n}
            \bnf
            \apo_r
        )
    }{\apredsetp}
    }
    {
        \\
    \hoaretriplet{\cup_i \weakerof{\apredset_i}}{
        \code{N}(
            \hoaretriplet{\weakerof{\apredset_1}}{\apop_1}{\weakerof{\apredsetp_1}}
            \angelicChoice{}
            \ldots
            \angelicChoice{}
            \hoaretriplet{\weakerof{\apredset_n}}{\apop_n}{\weakerof{\apredsetp_n}}
        )
    }{\cup_i \weakerof{\apredsetp_i}}
    }
    \ .
\end{multline*}

Moving on to the case where one production is chosen.
Without loss of generality, let this be the first production.
The proof therefore is
$\angelicProofHoareTripletHolds{\weakerof{\apredset_1}}{\apop_1}{\weakerof{\apredsetp_1}}$.
Since we know that 
$\angelicProofHoareTripletHolds{\apredset_1}{\apo_1}{\apredsetp_1}$ holds, we can apply the induction hypothesis and get the rewrite 
$\angelicProofRewriteHolds{\hoaretriplet{\apredset_1}{\apo_1}{\apredsetp_1}}{\hoaretriplet{\weakerof{\apredset}}{\apop_1}{\weakerof{\apredsetp_1}}}$.
Then, we apply rule \ruleLabel{RAC} and get the following:
\begin{multline*}
     \angelicProofRewriteHolds{
    \hoaretriplet{\apredset}{
        \code{N}(
            \hoaretriplet{\apredset_1}{\apo_1}{\apredsetp_1}
            \angelicChoice{}
            \ldots
            \angelicChoice{}
            \hoaretriplet{\apredset_n}{\apo_n}{\apredsetp_n}
        )
    }{\apredsetp}
    }
    {
    \\
    \hoaretriplet{\weakerof{\apredset_1} \cup (\cup_i \apredset_i)}{
        \code{N}(
            \hoaretriplet{\weakerof{\apredset_1}}{\apop_1}{\weakerof{\apredsetp_1}}
            \angelicChoice{}
            \ldots
            \angelicChoice{}
            \hoaretriplet{\apredset_n}{\apo_n}{\apredsetp_n}
        )
    }{\weakerof{\apredsetp_1} \cup (\cup_i \apredsetp_i)}
    }
    \ . 
\end{multline*}
Furthermore, the rule \ruleLabel{RSELECT} can be applied to get 
\begin{equation*}
     \angelicProofRewriteHolds
{
    \hoaretriplet{\weakerof{\apredset_1} \cup (\cup_i \apredset_i)}{
        \code{N}(
            \hoaretriplet{\weakerof{\apredset_1}}{\apop_1}{\weakerof{\apredsetp_1}}
            \angelicChoice{}
            \ldots
            \angelicChoice{}
            \hoaretriplet{\apredset_n}{\apo_n}{\apredsetp_n}
        )
    }{\weakerof{\apredsetp_1} \cup (\cup_i \apredsetp_i)}
    }
    {
        \\
        \hoaretriplet{\weakerof{\apredset_1}}{\apop_1}{\weakerof{\apredsetp_1}}
    }
    \ . 
\end{equation*}
Lastly, applying the rule \ruleLabel{RTRANS} yields the sought after result:
\begin{equation*}
     \angelicProofRewriteHolds{
    \hoaretriplet{\apredset}{
        \code{N}(
            \hoaretriplet{\apredset_1}{\apo_1}{\apredsetp_1}
            \angelicChoice{}
            \ldots
            \angelicChoice{}
            \hoaretriplet{\apredset_n}{\apo_n}{\apredsetp_n}
        )
    }{\apredsetp}
    }
    {
        \\
        \hoaretriplet{\weakerof{\apredset_1}}{\apop_1}{\weakerof{\apredsetp_1}}
    }
    \ . 
\end{equation*}
This concludes the proof.
\end{proof}

We proof the following lemma:
\begin{lemma}\label{lem:rewriteDemProgSingletons}
    The following implication holds:
    \begin{multline*}
        \posound \choiceOf{\hoaretriplet{\apredset}{\apo_1}{\apredsetp}}{\hoaretriplet{\apredset}{\apo_2}{\apredsetp}}
        \wedge
        \angelicAssertionStrongerOf{\apredsetp}{\set{\apredp}}
        \implies 
        \exists \, \apred, \apop_1, \apop_2 \dv 
        \apred \in \apredset 
        \wedge \angelicProofHoareTripletHolds{\set{\apred}}{\apop_1}{\set{\apredp}}
        \wedge
        \\
        \angelicProofHoareTripletHolds{\set{\apred}}{\apop_2}{\set{\apredp}}
        \wedge \proofStrongerOf{\apo_1}{\apop_1}
        \wedge \proofStrongerOf{\apo_2}{\apop_2} \,
        \wedge
        \\
        \proofToProgFuncOf{\apo_1} = \proofToProgFuncOf{\apop_1}
        \wedge \proofToProgFuncOf{\apo_2} = \proofToProgFuncOf{\apop_2}
        \ .
    \end{multline*}
\end{lemma}
\begin{proof}
We conduct the proof by induction over the proof tree of 
\begin{equation*}
        \posound \choiceOf{
            \hoaretriplet{\apredset}{\apo_1}{\apredsetp}
        }{
            \hoaretriplet{\apredset}{\apo_2}{\apredsetp}
        }
\end{equation*}
The base case holds trivially.

In the first case of the induction step, 
the proof was done using the rule \ruleLabel{PDEM}.
Then, we have
$\angelicProofHoareTripletHolds{\apredset}{\apo_1}{\apredsetp}$ and also
$\angelicProofHoareTripletHolds{\apredset}{\apo_2}{\apredsetp}$ with $\apredset = \set{\apred}$.
Thus, we know $\angelicProofHoareTripletHolds{\set{\apred}}{\apo_i}{\apredsetp}$.
Using the rule \ruleLabel{PCSQ} on both proofs, we get the realizability triples
$\angelicProofHoareTripletHolds{\set{\apred}}{\apo_i}{\set{\apredp}}$.

We have two remaining cases, either the proof was finished using the rule \ruleLabel{PGATHER} or rule \ruleLabel{PCSQ}.
In the first case, we have 
\begin{equation*}
    \posound
    \choiceOf{
        \hoaretriplet{\cup_i \apredset_i}{\apo_1}{\cup_i \apredsetp_i}
    }{
        \hoaretriplet{\cup_i \apredset_i}{\apo_2}{\cup_i \apredsetp_i}
    }
\end{equation*}
Since $\angelicAssertionStrongerOf{\cup_i \apredsetp_i}{\set{s}}$, one assertion $\apredsetp_i$ is stronger than $\set{\apredp}$.
Without loss of generality, let this be $\apredsetp_1$.
From the above realizability triple, we know that the realizability triple
\begin{equation*}
\posound \choiceOf{\hoaretriplet{\apredset_1}{\apop_1}{\apredsetp_1}}{\hoaretriplet{\apredset_1}{\apop_2}{\apredsetp_1}}
\end{equation*}
holds. Applying the induction hypothesis, we get an $\apred$ and the proofs $\apo_1''$ and $\apo_2''$ for which the realizability triples
$\angelicProofHoareTripletHolds{\set{\apred}}{\apo_1''}{\set{\apredp}}$
and
$\angelicProofHoareTripletHolds{\set{\apred}}{\apo_2''}{\set{\apredp}}$
hold.

In the second case, we have 
\begin{equation*}
    \posound \choiceOf{
        \hoaretriplet{\apredset}{\apo_1}{\weakerof{\apredsetp}}
    }{
        \hoaretriplet{\apredset}{\apo_2}{\weakerof{\apredsetp}}
    }
\end{equation*}
The precondition of the rule \ruleLabel{PCSQ} requires
\begin{equation*}
    \posound \choiceOf{
        \hoaretriplet{\weakerof{\apredset}}{\apo_1}{\apredsetp}
    }{
        \hoaretriplet{\weakerof{\apredset}}{\apo_2}{\apredsetp}
    }
\end{equation*}
Applying the induction hypothesis yields an $\apred \in \weakerof{\apredset}$ and 
$\posound \hoaretriplet{\set{\apred}}{\apop_i}{\set{\apredp}}$.
Since $\angelicAssertionStrongerOf{\apredset}{\weakerof{\apredset}}$ there is an $\apred'$ of $\apredset$ with $\demonicAssertionStrongerOf{\apred'}{\apred}$.
Thus we can strengthen the precondition of the realizability triples to get 
$\posound \hoaretriplet{\set{\apred'}}{\apop_i}{\set{\apredp}}$.
\end{proof}

We proof the following lemma:
\begin{lemma}\label{lem:rewriteLoopProgSingletons}
    The following implication holds:
    \begin{multline*}
        \angelicProofHoareTripletHolds{\apredset}{(\hoaretriplet{I}{\apo}{I})^*}{\apredsetp} \wedge
        i \in \angelicAssertionFont{I}
        \ \implies \  
        \exists \, \apop \dv \,
        \angelicProofHoareTripletHolds{\set{i}}{\apop}{\set{i}}
        \wedge 
        \\
        \proofStrongerOf{\apo}{\apop}
        \wedge \proofToProgFuncOf{\apo} = \proofToProgFuncOf{\apop}
        \ .
    \end{multline*}
\end{lemma}
\begin{proof}
    The proof is analogous to one of the previous lemma.
\end{proof}

\subsubsection*{Proof of \Cref{th:angelicChoiceElimination}}
The proof is done by an induction over the structure of the $\code{prf}$.
Again, we first state the theorem using the new notation:
\begin{theorem}[Backtracking Freedom]
    Let $\angelicProofHoareTripletHolds{\apredset}{\apo}{\apredsetp}$
    and $\apredp \in \apredsetp$. 
    Then there are $\apred$ and $\apop$ so that 
    $\angelicProofRewriteHolds{
        \hoaretriplet{\apredset}{\apo}{\apredsetp}
    }{
        \hoaretriplet{\set{\apred}}{\apop}{\set{\apredp}}
    }$, 
    $\apred \in \apredset$, 
    and $\proofToProgFuncOf{\apop} \in \angelConcFuncOf{\proofToProgFuncOf{\apo}}$.
\end{theorem}

\begin{proof}
\baseCase{}
The proof is $\angelicProofHoareTripletHolds{\apredset}{\code{com}}{\apredsetp}$.
Applying soundness yields a predicate $\apred$ with $\demonicAssertionStrongerOf{\progSemFuncOf{\code{com}}{\apred}}{\apredp}$ and $\apred \in \apredset$.
Therefore, we can apply rule \ruleLabel{RCOM} and get $\angelicProofRewriteHolds{\hoaretriplet{\apredset}{\code{com}}{\apredsetp}}{\hoaretriplet{\set{\apred}}{\code{com}}{\set{\apredp}}}$.

\inductionStep{}
If the proof is a sequence, we have 
\begin{equation*}
    \hoaretriplet{\apredset}{\apo_1}{\apredsetp}
    \code{;}
    \hoaretriplet{\apredsetp}{\apo_2}{\apredsetpp}
    \ .
\end{equation*}
Then, we know that the subproofs are true so $\apo_1$ and $\apo_2$ hold.
Applying the induction hypothesis on the second realizability triple from above, we get an $\apredp$ with $\apredp \in \apredsetp$ and a proof $\apop_2$ for which the following holds:
\begin{equation*}
\angelicProofRewriteHolds{\hoaretriplet{\apredsetp}{\apo_2}{\apredsetpp}}{\hoaretriplet{\set{\apredp}}{\apop_2}{\set{\apredpp}}}
\ .
\end{equation*}
Next, we apply the induction hypothesis on the other sequence and get a predicate $\apred$ with $\apred \in \apredset$ and a proof $\apop_1$ for which the following holds:
\begin{equation*}
\angelicProofRewriteHolds{\hoaretriplet{\apredset}{\apo_1}{\apredsetp}}{\hoaretriplet{\set{\apred}}{\apop_1}{\set{\apredp}}}
\ .
\end{equation*}
Applying rule \ruleLabel{RSEQR}, we get the following rewrite:
\begin{equation*}
    \angelicProofRewriteHolds{
            \hoaretriplet{\apredset}{\apo_1}{\apredsetp}
            \code{;}
            \hoaretriplet{\apredsetp}{\apo_2}{\apredsetpp}
    }{
            \hoaretriplet{\apredset}{\apo_1}{\set{\apredp}}
            \code{;}
            \hoaretriplet{\set{\apredp}}{\apop_2}{\set{\apredpp}}
    }
    \ .
\end{equation*}
Then, applying rule \ruleLabel{RSEQL}, we get the following rewrite:
\begin{equation*}
        \angelicProofRewriteHolds{
            \hoaretriplet{\apredset}{\apo_1}{\set{\apredp}}
            \code{;}
            \hoaretriplet{\set{\apredp}}{\apop_2}{\set{\apredpp}}
    }{
            \hoaretriplet{\set{\apred}}{\apop_1}{\set{\apredp}}
            \code{;}
            \hoaretriplet{\set{\apredp}}{\apop_2}{\set{\apredpp}}
    }
    \ .
\end{equation*}
Finally, with rule \ruleLabel{RTRANS}, we get the sought after result: 
\begin{equation*}
    \angelicProofRewriteHolds{
            \hoaretriplet{\apredset}{\apo_1}{\apredsetp}
            \code{;}
            \hoaretriplet{\apredsetp}{\apo_2}{\apredsetpp}
    }{
            \hoaretriplet{\set{\apred}}{\apop_1}{\set{\apredp}}
            \code{;}
            \hoaretriplet{\set{\apredp}}{\apop_2}{\set{\apredpp}}
    }
    \ .
\end{equation*}
If the proof is a demonic choice, so we have 
\begin{equation*}
    \choiceOf{
        \hoaretriplet{\apredset}{\apo_1}{\apredsetp}
    }{
        \hoaretriplet{\apredset}{\apo_2}{\apredsetp}
    }
    \ ,
\end{equation*}
and a predicate $\apredp$ with $\apredp \in \apredsetp$.
Then we can apply \Cref{lem:rewriteDemProgSingletons} to get 
a predicate $\apred$ and the proofs
$\angelicProofHoareTripletHolds{\set{\apred}}{\apop_1}{\set{\apredp}}$
and
$\angelicProofHoareTripletHolds{\set{\apred}}{\apop_2}{\set{\apredp}}$
with $\proofStrongerOf{\apo_i}{\apop_i}$.
.
Using \Cref{th:angelicDecisionCalcComplete}, the following two rewrites are sound:
\begin{equation*}
    \angelicProofRewriteHolds{
        \hoaretriplet{\apredset}{\apo_1}{\apredsetp}
    }{
        \hoaretriplet{\set{\apred}}{\apop_1}{\set{\apredp}}
    }
    \text{ and }
    \angelicProofRewriteHolds{
        \hoaretriplet{\apredset}{\apo_2}{\apredsetp}
    }{
        \hoaretriplet{\set{\apred}}{\apop_2}{\set{\apredp}}
    }
    \ .
\end{equation*}

An application of the rule \ruleLabel{RDEM} then yields the sought after result: 
\begin{equation*}
    \angelicProofRewriteHolds{
            \choiceOf{
                \hoaretriplet{\apredset}{\apo_1}{\apredsetp}
            }{
                \hoaretriplet{\apredset}{\apo_2}{\apredsetp}
            }
        }{
            \choiceOf{
                \hoaretriplet{\set{\apred}}{\apop_1}{\set{\apredp}}
            }{
                \hoaretriplet{\set{\apred}}{\apop_2}{\set{\apredp}}
            }
        }
        \ .
\end{equation*}

Moving on with loops.
We have the proof
$\angelicProofHoareTripletHolds{\apredset}{(\hoaretriplet{I}{\apo}{I})^*}{\apredsetp}$ and a predicate $\apredp$ with $\apredp \in \apredsetp$.
Since $\angelicAssertionStrongerOf{I}{\apredsetp}$, there also is an $i$ of $I$ with $\demonicAssertionStrongerOf{i}{\apredp}$.
Applying \Cref{lem:rewriteLoopProgSingletons}, we get a proof $\apop$ for which 
$\angelicProofHoareTripletHolds{\set{i}}{\apop}{\set{i}}$
holds and $\proofStrongerOf{\apo}{\apo'}$. 
\Cref{th:angelicDecisionCalcComplete} then yields 
\begin{equation*}
    \angelicProofRewriteHolds{
        \hoaretriplet{I}{\apo}{I}
    }{
        \hoaretriplet{\set{i}}{\apop}{\set{i}}
    }
    \ .
\end{equation*}

Applying rule \ruleLabel{RLOOP} yields 
\begin{equation*}
    \angelicProofRewriteHolds{
        \hoaretriplet{\apredset}{(\hoaretriplet{I}{\apo}{I})^*}{\apredsetp}
    }{
        \hoaretriplet{\set{i}}{(\hoaretriplet{\set{i}}{\apop}{\set{i}})^*}{\set{i}}
    }
    \ .
\end{equation*}
Since $\angelicAssertionStrongerOf{\apredset}{I}$, we know there is an $\apred$ in $\apredset$ for which 
$\demonicAssertionAbsStrongerOf{\apred}{i}$ holds.
Applying rule \ruleLabel{RCSQ}, we get the sought after result:
\begin{equation*}
    \angelicProofRewriteHolds{
        \hoaretriplet{\apredset}{(\hoaretriplet{I}{\apo}{I})^*}{\apredsetp}
    }{
        \hoaretriplet{\set{r}}{(\hoaretriplet{\set{i}}{\apop}{\set{i}})^*}{\set{s}}
    }
    \ .
\end{equation*}

In the last case, we consider proofs over nonterminals.
We have a predicate $\apredp$ and the proof 
\begin{equation*}
\angelicProofHoareTripletHolds{\apredset}{
    \code{N}(
        \hoaretriplet{\apredset_1}{\apo_1}{\apredsetp_1}
        \angelicChoice{}
        \ldots
        \angelicChoice{}
        \hoaretriplet{\apredset_n}{\apo_n}{\apredsetp_n}
    )
    }{\apredsetp}
\end{equation*}
with $\apredp \in \apredsetp$.
Since $(\apredsetp_1 \cup \ldots \cup \apredsetp_n) = \apredsetp$, 
there is a selection $\apredsetp_i$ with $
\apredp \in \apredsetp_i
$.
Without loss of generality, let this $i$ be $1$.
We know the realizability triple $\angelicProofHoareTripletHolds{\apredset_1}{\apo_1}{\apredsetp_1}$ holds.
Using the induction hypothesis, we get the following for an $\apred_1$ of $\apredset_1$:
\begin{equation*}
    \angelicProofRewriteHolds{
        \hoaretriplet{\apredset_1}{\apo_1}{\apredsetp_1}
    }{
        \hoaretriplet{\set{\apred_1}}{\apop_1}{\set{\apredp}}
    }
    \ .
\end{equation*}
Using the rule \ruleLabel{RANG}, we get the following rewrite:
\begin{equation*}
    \angelicProofRewriteHolds{
    \code{N}(
        \hoaretriplet{\apredset_1}{\apo_1}{\apredsetp_1}
        \angelicChoice{}
        \ldots
        \angelicChoice{}
        \hoaretriplet{\apredset_n}{\apo_n}{\apredsetp_n}
    )
    }{
        \\
            \code{N}(
                \hoaretriplet{\set{\apred_1}}{\apop_1}{\set{\apredp}}
                \angelicChoice{}
                \ldots
                \angelicChoice{}
                \hoaretriplet{\apredset_n}{\apo_n}{\apredsetp_n}
            )
    }
    \ .
\end{equation*}
Then, with rule \ruleLabel{RSELECT}, we get
\begin{equation*}
    \angelicProofRewriteHolds{
            \code{N}(
                \hoaretriplet{\set{\apred_1}}{\apop_1}{\set{\apredp}}
                \angelicChoice{}
                \ldots
                \angelicChoice{}
                \hoaretriplet{\apredset_n}{\apo_n}{\apredsetp_n}
            )
    }
    {
        \\
        \hoaretriplet{\set{\apred_1}}{\apop_1}{\set{\apredp}}
    }
    \ .
\end{equation*}
Finally, applying rule \ruleLabel{RTRANS} twice, we get
\begin{equation*}
    \angelicProofRewriteHolds{
    \hoaretriplet{\apredset}{
    \code{N}(
        \hoaretriplet{\apredset_1}{\apo_1}{\apredsetp_1}
        \angelicChoice{}
        \ldots
        \angelicChoice{}
        \hoaretriplet{\apredset_n}{\apo_n}{\apredsetp_n}
    )
    }{\apredsetp}
    }
    {
        \hoaretriplet{\set{\apred_1}}{\apop_1}{\set{\apredp}}
    }
    \ .
\end{equation*}
Because $\apredset_1 \subseteq \apredset$, we have $\apred_1 \in \apredset$.
This concludes the proof.
\end{proof}

%% file: proofsChapter4.tex
\subsection{Proofs for \Cref{ch:vc}}
We show the following lemma:
\begin{lemma}\label{lem:spMonotonic}
    The function $\strongestPostSem{}$ is monotonic (if the specifications of the nonterminals are monotonic):
    \begin{equation*}
        \angelicAssertionStrongerOf{\apredset}{\apredset'} 
        \ \implies \ 
        \angelicAssertionStrongerOf{\strongestPostSemOf{\apredset}{\asketch}}{\strongestPostSemOf{\apredset'}{\asketch}}
        \ .
    \end{equation*}
\end{lemma}

\begin{proof}
We proof this lemma by induction over the structure of sketch.

\baseCase{}
    We show $\angelicAssertionStrongerOf{\apredset}{\apredset'}$ implies $\angelicAssertionStrongerOf{\strongestPostSemOf{\apredset}{\code{com}}}{\strongestPostSemOf{\apredset'}{\code{com}}}$.
Let the predicate $s$ be of $\strongestPostSemOf{\apredset'}{\code{com}}$.
Then, $s = \progSemFuncOf{\code{com}}{\apred'}$ for some predicate $\apred'$ of $\apredset'$.
    Since we know $\angelicAssertionStrongerOf{\apredset}{\apredset'}$, there is a predicate $\apred$ in $R$ with $\demonicAssertionStrongerOf{r}{\apred'}$.
Due to the interpretation of commands being monotonic, we get that $\demonicAssertionStrongerOf{\progSemFuncOf{\code{com}}{r}}{\progSemFuncOf{\code{com}}{\apred'}} = s$.
    Because $\progSemFuncOf{\code{com}}{r}$ is in $\strongestPostSemOf{\apredset}{\code{com}}$, 
    $\angelicAssertionStrongerOf{\strongestPostSemOf{\apredset}{\code{com}}}{\strongestPostSemOf{\apredset'}{\code{com}}}$ follows.

\inductionStep{}
First, we consider the case where the sketch is a sequence.
    We show $\angelicAssertionStrongerOf{\apredset}{\apredset'}$ implies $\angelicAssertionStrongerOf{\strongestPostSemOf{\apredset}{\concatof{\asketch_1}{\asketch_2}}}{\strongestPostSemOf{\apredset'}{\concatof{\asketch_1}{\asketch_2}}}$.
    From the induction hypothesis, we know that $\angelicAssertionStrongerOf{\strongestPostSemOf{\apredset}{\asketch_1}}{\strongestPostSemOf{\apredset'}{\asketch_1}}$ holds.
Applying the induction hypothesis again yields 
\begin{equation*}
    \angelicAssertionStrongerOf{\strongestPostSemOf{\strongestPostSemOf{\apredset}{\asketch_1}}{\asketch_2}}{\strongestPostSemOf{\strongestPostSemOf{\apredset'}{\asketch_1}}{\asketch_2}}
\ .
\end{equation*}
Inserting the definition of the strongest post function, we get 
\begin{equation*}
    \angelicAssertionStrongerOf{\strongestPostSemOf{\apredset}{\concatof{\asketch_1}{\asketch_2}}}{\strongestPostSemOf{\apredset'}{\concatof{\asketch_1}{\asketch_2}}}
\ .
\end{equation*}

Next, we consider the case where the sketch is a choice.
We show 
\begin{equation*}
    \angelicAssertionStrongerOf{\apredset}{\apredset'} 
\implies 
    \angelicAssertionStrongerOf{\strongestPostSemOf{\apredset}{\choiceOf{\asketch_1}{\asketch_2}}}{\strongestPostSemOf{\apredset'}{\choiceOf{\asketch_1}{\asketch_2}}}
\ .
\end{equation*}
Let the predicate $s_j$ be of $\strongestPostSemOf{\apredset'}{\choiceOf{\asketch_1}{\asketch_2}}$.
Thus, there is an $\apred'$ of $\apredset'$ with $s_j' = s_1' \join{} s_2'$ and
$s_1'$ is of $\strongestPostSemOf{\set{\apred'}}{\asketch_1}$ and 
    $s_2'$ is of $\strongestPostSemOf{\set{\apred'}}{\asketch_2}$.
    Since $\angelicAssertionStrongerOf{\apredset}{\apredset'}$, there is an $\apred$ of $R$ with $\angelicAssertionStrongerOf{\set{\apred}}{\set{\apred'}}$.
Applying the induction hypothesis yields 
\begin{equation*}
    \angelicAssertionStrongerOf{\strongestPostSemOf{\set{\apred}}{\asketch_1}}{\strongestPostSemOf{\set{\apred'}}{\asketch_1}} 
    \text{ and }
    \angelicAssertionStrongerOf{\strongestPostSemOf{\set{\apred}}{\asketch_2}}{\strongestPostSemOf{\set{\apred'}}{\asketch_2}}
\end{equation*}
    Therefore, there is an $s_1 \in \strongestPostSemOf{\set{\apred}}{\asketch_1}$ and an $s_2 \in \strongestPostSemOf{\set{\apred}}{\asketch_2}$ with $\angelicAssertionStrongerOf{\set{s_1}}{\set{s_1'}}$ and $\angelicAssertionStrongerOf{\set{s_2}}{\set{s_2'}}$.
Therefore, $s_j = \demonicAssertionStrongerOf{s_1 \join{} s_2}{s_1' \join{} s_2'} = s_j'$.
    Since $s_j $ is of $\strongestPostSemOf{\apredset}{
    \choiceOf{\asketch_1}{\asketch_2}
}$, we have 
\begin{equation*}
    \angelicAssertionStrongerOf{\strongestPostSemOf{\apredset}{\choiceOf{\asketch_1}{\asketch_2}}}{\strongestPostSemOf{\apredset'}{\choiceOf{\asketch_1}{\asketch_2}}}   
 \ .
\end{equation*}

Next, consider the sketch is a loop.
We show 
\begin{equation*}
    \angelicAssertionStrongerOf{\apredset}{\apredset'} 
\implies 
    \angelicAssertionStrongerOf{\strongestPostSemOf{\apredset}{\asketch^*[I]}}
{\strongestPostSemOf{\apredset'}{\asketch^*[I]}}
\ .
\end{equation*}
Let $i$ be of $\strongestPostSemOf{\apredset'}{\asketch^*[I]}$.
Thus, $\angelicAssertionStrongerOf{\apredset'}{\set{i}}$.
By transitivity we have $\angelicAssertionStrongerOf{\apredset}{\set{i}}$ and thus $i$ is also in $\strongestPostSemOf{\apredset}{\asketch^*[I]}$.
The inequality $\demonicAssertionStrongerOf{i}{i}$ holds trivially.

Lastly, consider the sketch is a nonterminal.
Monotonicity follows directly from the required monotonicity of the specification.

This concludes the proof.
\end{proof}

We show the following lemma:
\begin{lemma}[Soundness]\label{th:spSoundess}
    For any $\apredset$ of $\angelicAssertions{}$ and any sketch $\asketch$ of $\angelicProgramLang{}$, the following realizability triple holds (if the specifications of the nonterminals and loops are sound): 
    \begin{equation*}
        \angelicHoareTripletHolds{\apredset}{\asketch}{\strongestPostSemOf{\apredset}{\asketch}}
        \ .
    \end{equation*}
\end{lemma}

\begin{proof}
We proof this lemma by induction over the structure of the sketch. 

\baseCase{}
We show 
$\angelicHoareTripletHolds{\apredset}{\code{com}}{\strongestPostSemOf{\apredset}{\code{com}}}$.
For every $\apred \in \apredset$ we can can show 
the realizability triple
$\angelicHoareTripletHolds{\set{\apred}}{\code{com}}{\set{\progSemFuncOf{\code{com}}{\apred}}}$.
Using the rule \ruleLabel{GATHER} we can show 
$\angelicHoareTripletHolds{\apredset}{\code{com}}{\bigcup_{\apred \in \apredset}\set{\progSemFuncOf{\code{com}}{\apred}}}$.
By definition, this is the same as 
$\angelicHoareTripletHolds{\apredset}{\code{com}}{\strongestPostSemOf{\apredset}{\code{com}}}$.

\inductionStep{}
First, consider a sequence of sketch.
We show the realizability triple
\begin{equation*}
\angelicHoareTripletHolds{\apredset}{\asketch_1\code{;}\asketch_2}{\strongestPostSemOf{\apredset}{\asketch_1\code{;}\asketch_2}}
\ .
\end{equation*}
Applying the induction hypothesis twice yields the realizability triples:
\begin{equation*}
\angelicHoareTripletHolds{\apredset}{\asketch_1}{\strongestPostSemOf{\apredset}{\asketch_1}}
\text{ and } 
\angelicHoareTripletHolds{\strongestPostSemOf{\apredset}{\asketch_1}}{\asketch_2}{\strongestPostSemOf{\strongestPostSemOf{\apredset}{\asketch_1}}{\asketch_2}}
\ .
\end{equation*}
Using rule \ruleLabel{SEQ}, we get the realizability triple
\begin{equation*}
\angelicHoareTripletHolds{\apredset}{\asketch_1\code{;}\asketch_2}{\strongestPostSemOf{\strongestPostSemOf{\apredset}{\asketch_1}}{\asketch_2}}
\ .
\end{equation*}
This is the same as the realizability triple
    $\angelicHoareTripletHolds{\apredset}{\asketch_1\code{;}\asketch_2}{\strongestPostSemOf{\apredset}{\asketch_1\code{;}\asketch_2}}$
.

Next, consider the sketch is a choice.
We show
\begin{equation*}
\angelicHoareTripletHolds{\apredset}{\choiceOf{\asketch_1}{\asketch_2}}{\strongestPostSemOf{\apredset}{\choiceOf{\asketch_1}{\asketch_2}}}
\ .
\end{equation*}
From the definition of the strongest post, we know that it is the following:
\begin{equation*}
    \strongestPostSemOf{\apredset}{\choiceOf{\asketch_1}{\asketch_2}}
    = 
    \bigcup_{\apred \in \apredset}\setCond{\apredp_1 \join{} \apredp_2}{
        \apredp_1 \in \strongestPostSemOf{\set{\apred}}{\asketch_1}
        \wedge 
        \apredp_2 \in \strongestPostSemOf{\set{\apred}}{\asketch_2}
        }
        \ .
\end{equation*}
For every $\apred$ of $\apredset$, we know from the induction hypothesis that
$\angelicHoareTripletHolds{\set{\apred}}{\asketch_1}{\strongestPostSemOf{\set{\apred}}{\asketch_1}}$ and 
$\angelicHoareTripletHolds{\set{\apred}}{\asketch_2}{\strongestPostSemOf{\set{\apred}}{\asketch_2}}$ hold.
Using the rule \ruleLabel{CSQ}, we can show 
the realizability triples
$\angelicHoareTripletHolds{\set{\apred}}{\asketch_1}{\set{\apredp_{k1}}}$ 
for any $\apredp_{k1}$ of $\strongestPostSemOf{\set{\apred}}{\asketch_1}$ and
$\angelicHoareTripletHolds{\set{\apred}}{\asketch_2}{\set{\apredp_{k2}}}$ 
for any $\apredp_{k2}$ of $\strongestPostSemOf{\set{\apred}}{\asketch_2}$.
Let $\apredp_j$ be the join of any two $\apredp_{k1}$ and $\apredp_{k2}$.
Using the rule \ruleLabel{CSQ}, we can show
$\angelicHoareTripletHolds{\set{\apred}}{\asketch_i}{\set{\apredp_j}}$.
With the rule \ruleLabel{DEM}, we get 
$\angelicHoareTripletHolds{\set{\apred}}{\choiceOf{\asketch_1}{\asketch_2}}{\set{\apredp_j}}$.
Applying rule \ruleLabel{GATHER} yields 
\begin{equation*}
\angelicHoareTripletHolds{\set{\apred}}{\choiceOf{\asketch_1}{\asketch_2}}
{\setCond{\apredp_1 \join{} \apredp_2}{
    \apredp_i \in \strongestPostSemOf{\set{\apred}}{\asketch_i}
}}
\ .
\end{equation*}

Applying rule \ruleLabel{GATHER} once more yields 
\begin{equation*}
    \angelicHoareTripletHolds{\apredset}{\choiceOf{\asketch_1}{\asketch_2}}
{\bigcup_{\apred \in \apredset}\setCond{\apredp_1 \join{} \apredp_2}{
    \apredp_i \in \strongestPostSemOf{\set{\apred}}{\asketch_i}
}}
\ .
\end{equation*}
This is the same as 
$
\angelicHoareTripletHolds{\apredset}{\choiceOf{\asketch_1}{\asketch_2}}
{\strongestPostSemOf{\apredset}{\choiceOf{\asketch_1}{\asketch_2}}}
$.

Next, consider the sketch is a loop.
Let $i$ be of $\strongestPostSemOf{\apredset}{\kleeneof{\asketch}}$.
From soundness of the invariant annotation, we get 
$\angelicHoareTripletHolds{\set{i}}{\asketch}{\set{i}}$.
With rule loop, we get
$\angelicHoareTripletHolds{\set{i}}{\kleeneof{\asketch}}{\set{i}}$.
With gather, we get 
\begin{equation*}
\angelicHoareTripletHolds{\strongestPostSemOf{\apredset}{\kleeneof{\asketch}}}{\kleeneof{\asketch}}{\strongestPostSemOf{\apredset}{\kleeneof{\asketch}}}\ .
\end{equation*}
With \ruleLabel{CSQ}, we get 
$\angelicHoareTripletHolds{\apredset}{\kleeneof{\asketch}}{\strongestPostSemOf{\apredset}{\kleeneof{\asketch}}}$.

Finally, consider the sketch is a nonterminal.
Soundness follows directly from the required soundness of the specification of the nonterminal.

This concludes the proof.
\end{proof}
We show the following corollary:
\begin{corollary}\label{cor:spHoareTriple}
    For any $\apredset, \angelicAssertionFont{\apredsetp}$ of $\angelicAssertions{}$ and any sketch $\asketch$ of $\angelicProgramLang{}$ with sound annotations, the following implication holds: 
   \begin{equation*}
       \angelicAssertionStrongerOf{\strongestPostSemOf{\apredset}{\asketch}}{\angelicAssertionFont{\apredsetp}} 
       \ \implies \ 
       \angelicHoareTripletHolds{\apredset}{\asketch}{\angelicAssertionFont{\apredsetp}} 
       \ .
   \end{equation*}   
\end{corollary}

\begin{proof}
    \Cref{th:spSoundess} yields $\angelicHoareTripletHolds{\apredset}{\asketch}{\strongestPostSemOf{\apredset}{\asketch}}$.
    Applying rule \ruleLabel{CSQ} gives us the sought after result 
    $\angelicHoareTripletHolds{\apredset}{\asketch}{\apredsetp}$
\end{proof}

\subsubsection*{Proof of \Cref{th:spComplete}}
\begin{proof}[\unskip\nopunct]
We have already shown soundness above.
    We prove completeness by an induction over the proof tree of $\angelicHoareTripletHolds{\apredset}{\asketch}{\apredsetp}$.

\baseCase{}
We show 
    $\angelicHoareTripletHolds{\set{\apred}}{\code{com}}{\set{\apredp}} \implies \angelicAssertionStrongerOf{\strongestPostSemOf{\set{\apred}}{\code{com}}}{\set{\apredp}}$.
From the precondition of rule \ruleLabel{COM}, we know $\demonicAssertionStrongerOf{\progSemFuncOf{\code{com}}{\apred}}{\apredp}$.
    Therefore $\strongestPostSemOf{\set{\apred}}{\code{com}} = 
    \angelicAssertionStrongerOf{\set{\progSemFuncOf{\code{com}}{\apred}}}{\set{\apredp}}$.

\inductionStep{}
Consider the rule \ruleLabel{SEQ}.
We show 
that the validity of
$\angelicHoareTripletHolds{\apredset}{\concatof{\asketch_1}{\asketch_2}}{T}$ implies
    $\angelicAssertionStrongerOf{\strongestPostSemOf{\apredset}{\concatof{\asketch_1}{\asketch_2}}}{T}$.
From the preconditions we know that both 
    $\angelicHoareTripletHolds{\apredset}{\asketch_1}{\apredsetp}$ and
    $\angelicHoareTripletHolds{\apredsetp}{\asketch_2}{T}$ hold.
Applying the induction hypothesis yields
    $\angelicAssertionStrongerOf{\strongestPostSemOf{\apredset}{\asketch_1}}{\apredsetp}$ and 
    $\angelicAssertionStrongerOf{\strongestPostSemOf{\apredsetp}{\asketch_2}}{T}$.
Because the strongest post function is monotonic we get the sought after inequality
    $\angelicAssertionStrongerOf{\strongestPostSemOf{\apredset}{\concatof{\asketch_1}{\asketch_2}}}{T}$. 

Consider the rule \ruleLabel{DEM}.
We show 
\begin{equation*}
    \angelicHoareTripletHolds{\set{\apred}}{\choiceOf{\asketch_1}{\asketch_2}}{\apredsetp} \implies \angelicAssertionStrongerOf{\strongestPostSemOf{\set{\apred}}{\choiceOf{\asketch_1}{\asketch_2}}}{\apredsetp}
\ .
\end{equation*}
From the preconditions we know 
that 
    $\angelicHoareTripletHolds{\set{\apred}}{\asketch_1}{\apredsetp}$
and
    $\angelicHoareTripletHolds{\set{\apred}}{\asketch_2}{\apredsetp}$
hold.
Applying the induction hypothesis yields the inequalities
    $\angelicAssertionStrongerOf{\strongestPostSemOf{\set{\apred}}{\asketch_1}}{\apredsetp}$
and 
    $\angelicAssertionStrongerOf{\strongestPostSemOf{\set{\apred}}{\asketch_2}}{\apredsetp}$.
Let the predicate $s$ be an element of $S$.
Because the strongest posts are more versatile than $S$, there is an $s_i$ of each strongest post with $\demonicAssertionStrongerOf{s_i}{s}$.
Therefore, $s$ is an upper bound of $s_1$ and $s_2$ and thus 
$\demonicAssertionStrongerOf{s_1 \join{} s_2}{s}$ holds.
To conclude this case, see that $s_1 \join{} s_2 \in 
    \strongestPostSemOf{\set{\apred}}{\choiceOf{\asketch_1}{\asketch_2}}$,
and therefore 
    $\angelicAssertionStrongerOf{\strongestPostSemOf{\set{\apred}}{\choiceOf{\asketch_1}{\asketch_2}}}{\apredsetp}$
holds.

Consider the rule \ruleLabel{ANG}.
Completeness follows directly from the requirement that the specification of the nonterminals is complete.

Consider the rule \ruleLabel{CSQ}.
We show
\begin{equation*}
\angelicHoareTripletHolds{\apredset}{\asketch}{\apredsetp'} 
\implies 
\angelicAssertionStrongerOf{\strongestPostSemOf{\apredset}{\asketch}}{\apredsetp'}
\ .
\end{equation*}
From the preconditions, we know that the realizability triple 
$\angelicHoareTripletHolds{\apredset'}{\asketch}{\apredsetp}$ holds with
$\angelicAssertionStrongerOf{\apredset}{\apredset'}$ and 
$\angelicAssertionStrongerOf{\apredsetp}{\apredsetp'}$.
From the induction hypothesis we know that 
$\angelicAssertionStrongerOf{\strongestPostSemOf{\apredset'}{\asketch}}{\apredsetp}$.
The sought after result immediately follows from monotonicity and because the relation $\angelicAssertionStronger{}$ is transitive:
$\angelicAssertionStrongerOf{\strongestPostSemOf{\apredset}{\asketch}}{\apredsetp'}$.

Consider the rule \ruleLabel{LOOP}.
We show $\angelicHoareTripletHolds{\set{i}}{\kleeneof{\asketch}[I]}{\set{i}}$ implies 
$\angelicAssertionStrongerOf{\strongestPostSemOf{\set{i}}{\kleeneof{\asketch}[I]}}{\set{i}}$.
Due to the precondition of rule \ruleLabel{LOOP} and completeness of the annotation, we have $i \in I$.
Thus, $i \in \strongestPostSemOf{\set{i}}{\kleeneof{\asketch}[I]}$ and therefore the inequality trivially holds.

Lastly, consider the rule \ruleLabel{GATHER}.
We show 
\begin{equation*}
\angelicHoareTripletHolds{\cup_i \apredset_i}{\asketch}{\cup_i \apredsetp_i} 
\implies 
    \angelicAssertionStrongerOf{\strongestPostSemOf{\cup_i \apredset_i}{\asketch}}{(\cup_i \apredsetp_i)}
\ .
\end{equation*}

From the preconditions we know that
$\angelicHoareTripletHolds{\apredset_i}{\asketch}{\apredsetp_i}$ hold.
Applying the induction hypothesis yields
$\angelicAssertionStrongerOf{\strongestPostSemOf{\apredset_i}{\asketch}}{\apredsetp_i}$.
Using monotonicity of the strongest post function, we get 
$\angelicAssertionStrongerOf{\strongestPostSemOf{\cup_i \apredset_i}{\asketch}}{\apredsetp_i}$ for every $i$.
Together, this means 
$\angelicAssertionStrongerOf{\strongestPostSemOf{\cup_i \apredset_i}{\asketch}}{(\cup_i \apredsetp_i)}$ 
This concludes the proof.
\end{proof}

\subsubsection*{Proof of \Cref{th:vcSoundness}}
\begin{proof}[\unskip\nopunct]
We prove soundness by induction over the structure of the sketch.

\baseCase{}
    We show $\models \verificationConditionsFuncOf{\hoaretriplet{\apredset}{\code{com}}{\apredsetp}} \implies \angelicHoareTripletHolds{\apredset}{\code{com}}{\apredsetp}$.
    Because the verification conditions hold, we know that the inequality $\angelicAssertionStrongerOf{\vcStrongestPostSemOf{\apredset}{\code{com}}}{\apredsetp}$ is true.
    The realizability triple $\angelicHoareTripletHolds{\apredset}{\code{com}}{\apredsetp}$ immediately follows from \Cref{cor:spHoareTriple}.

\inductionStep{}
In the first case, the sketch is a sequence.
    We know the following verification conditions hold: $\models \verificationConditionsFuncOf{\hoaretriplet{\apredset}{\concatof{\asketch_1}{\asketch_2}}{\apredsetp}}$.
Thus, we know that the verification conditions 
$\models \verificationConditionsFuncOf{\hoaretriplet{\apredset}{\asketch_2}{\vcStrongestPostSemOf{\apredset}{\asketch_1}}}$
and 
    $\models \verificationConditionsFuncOf{\hoaretriplet{\vcStrongestPostSemOf{\apredset}{\asketch_1}}{\asketch_2}{\apredsetp}}$
also hold.
Applying the induction hypothesis yields 
\begin{equation*}
\angelicHoareTripletHolds{\apredset}{\asketch_1}{\vcStrongestPostSemOf{\apredset}{\asketch_1}}
\text{ and } 
    \angelicHoareTripletHolds{\vcStrongestPostSemOf{\apredset}{\asketch_1}}{\asketch_2}{\apredsetp}
\ .
\end{equation*}
Using the rule \ruleLabel{SEQ}, we get 
    $\angelicHoareTripletHolds{\apredset}{\concatof{\asketch_1}{\asketch_2}}{\apredsetp}$.

If the sketch is a choice, we have 
    $\angelicAssertionStrongerOf{\vcStrongestPostSemOf{\apredset}{\choiceOf{\asketch_1}{\asketch_2}}}{\apredsetp}$
and for every $r \in \apredset$ we have 
    $\models \verificationConditionsFuncOf{\hoaretriplet{\set{\apred}}{\asketch_1}{\vcStrongestPostSemOf{\set{\apred}}{\asketch_1}}}$
and
    $\models \verificationConditionsFuncOf{\hoaretriplet{\set{\apred}}{\asketch_2}{\vcStrongestPostSemOf{\set{\apred}}{\asketch_2}}}$.
Applying the induction hypothesis yields
the two realizability triples
    $\angelicHoareTripletHolds{\set{\apred}}{\asketch_1}{\vcStrongestPostSemOf{\set{\apred}}{\asketch_1}}$
and 
    $\angelicHoareTripletHolds{\set{\apred}}{\asketch_2}{\vcStrongestPostSemOf{\set{\apred}}{\asketch_2}}$
for every $r \in \apredset$.
Using rule \ruleLabel{CSQ}, we can show the realizability triple
    $\angelicHoareTripletHolds{\set{\apred}}{\asketch_1}{\vcStrongestPostSemOf{\set{\apred}}{\choiceOf{\asketch_1}{\asketch_2}}}$
and
    $\angelicHoareTripletHolds{\set{\apred}}{\asketch_2}{\vcStrongestPostSemOf{\set{\apred}}{\choiceOf{\asketch_1}{\asketch_2}}}$.
Using the rule \ruleLabel{DEM}, we get 
    $\angelicHoareTripletHolds{\set{\apred}}{\choiceOf{\asketch_1}{\asketch_2}}{\vcStrongestPostSemOf{\set{\apred}}{\choiceOf{\asketch_1}{\asketch_2}}}$
for every $r \in \apredset$.
Using the gather rule, we get 
$\angelicHoareTripletHolds{\apredset}{\choiceOf{\asketch_1}{\asketch_2}}{\vcStrongestPostSemOf{\apredset}{\choiceOf{\asketch_1}{\asketch_2}}}$.
Applying the rule \ruleLabel{CSQ}, we get
    $\angelicHoareTripletHolds{\apredset}{\choiceOf{\asketch_1}{\asketch_2}}{\apredsetp}$.

If the sketch is a loop, we have the verification conditions
$\models (\bigcup_{i \in I} \verificationConditionsFuncOf{\hoaretriplet{\set{i}}{\code{s}}{\set{i}}})$ from the check annotation function
and the inequality
    $\angelicAssertionStrongerOf{\vcStrongestPostSemOf{\apredset}{\code{s}^*[I]}}{\apredsetp}$.
The induction hypothesis yields the realizability triple
$\angelicHoareTripletHolds{\set{i}}{\code{s}}{\set{i}}$ for every $i \in I$.
Applying the rule \ruleLabel{LOOP} yields 
$\angelicHoareTripletHolds{\set{i}}{\code{s}^*}{\set{i}}$ for every $i \in I$.
Since $\vcStrongestPostSemOf{\apredset}{\code{s}^*[I]}$ is a subset of $I$ we can use multiple applications of the rule \ruleLabel{GATHER} to get 
$\angelicHoareTripletHolds{\vcStrongestPostSemOf{\apredset}{\code{s}^*[I]}}{\code{s}^*}{\vcStrongestPostSemOf{\apredset}{\code{s}^*[I]}}$.
Per definition, we have $\angelicAssertionStrongerOf{\apredset}{\vcStrongestPostSemOf{\apredset}{\code{s}^*[I]}}$.
Using the rule \ruleLabel{CSQ}, we get that $\angelicHoareTripletHolds{\apredset}{\code{s}^*}{\apredsetp}$ holds.

If the sketch is a nonterminal,
we know that $\angelicAssertionStrongerOf{\strongestPostSemOf{\apredset}{\code{\anonterm}[\Gamma]}}{\apredsetp}$ holds and from the check annotations function we know there is a $j$ with $\models \verificationConditionsFuncOf{\hoaretriplet{\apredset}{\aprog}{\strongestPostSemOf{\apredset}{\aprog}}}$ for every $\aprog$ of the oracle function $\oracleFuncOf{\anonterm}{j}$.
Also, we have $\angelicAssertionStrongerOf{(\bigcup_{\aprog \in \oracleFuncOf{\anonterm}{j}} \strongestPostSemOf{\apredset}{\aprog})}{\Gamma(\apredset)}$.
Using the induction hypothesis and rule \ruleLabel{ANG}, we get 
$\angelicHoareTripletHolds{\apredset}{\anonterm}{\strongestPostSemOf{\apredset}{\aprog}}$ for every $\aprog$ of $\oracleFuncOf{\anonterm}{j}$.
Using the gather rule, we get 
$\angelicHoareTripletHolds{\apredset}{\anonterm}{
 \bigcup_{\aprog \in \oracleFuncOf{\anonterm}{j}}   
\strongestPostSemOf{\apredset}{\aprog}}$.
Using the rule \ruleLabel{CSQ}, we get 
$\angelicHoareTripletHolds{\apredset}{\anonterm}{\Gamma(\apredset)}$.
Since $\Gamma(\apredset)$ is equal to $\strongestPostSemOf{\apredset}{\anonterm[\Gamma]}$, we can apply \ruleLabel{CSQ} one more time to get 
$\angelicHoareTripletHolds{\apredset}{\anonterm}{\apredsetp}$.

Next, we show completeness.
First, see that if $\angelicAssertionStrongerOf{\apredsetp}{\weakerof{\apredsetp}}$ and $\models \verificationConditionsFuncOf{\hoaretriplet{\apredset}{\asketch}{\apredsetp}}$ then we also have
$\models \verificationConditionsFuncOf{\hoaretriplet{\apredset}{\asketch}{\weakerof{\apredsetp}}}$.
We proceed by induction over the structure of sketches.

\baseCase{}
We have $\asketch = \code{com}$.
Completeness results from completeness of $\vcStrongestPostSem{}$.

\inductionStep{}
We have $\angelicHoareTripletHoldsSemantically{\apredset}{\concatof{\asketch_1}{\asketch_2}}{\apredsetp}$.
From soundness of the strongest post and the induction hypothesis, we get $\models \verificationConditionsFuncOf{\hoaretriplet{\apredset}{\asketch_1}{\vcStrongestPostSemOf{\apredset}{\asketch_1}}}$ and we also get the verification conditions of
$\models \verificationConditionsFuncOf{\hoaretriplet{\vcStrongestPostSemOf{\apredset}{\asketch_1}}{\asketch_2}{\vcStrongestPostSemOf{\apredset}{\concatof{\asketch_1}{\asketch_2}}}}$.
Because of completeness of the strongest post and our remark that the post condition in verification conditions may be weakened, we get 
$\models \verificationConditionsFuncOf{\hoaretriplet{\vcStrongestPostSemOf{\apredset}{\asketch_1}}{\asketch_2}{\apredsetp}}$.

The proof for choices directly follows from soundness and completeness of the strongest post function.
Then, only an application of the induction hypothesis is left.

In the case of loops we can discharge the first inequality by soundness and completeness from strongest post.
The inequalities from the check annotation function are true because the annotations are sound and thus the induction hypothesis can be applied.

The case for non-terminals is more cumbersome.
The first inequality follows from completeness of the strongest post function.
We now discuss why the inequalities in the check annotations function hold:
From $\angelicHoareTripletHoldsSemantically{\apredset}{\anonterm[\Gamma]}{\apredsetp}$, we get for every $\apredp \in \apredsetp$ a program $\aprog$ and an $\apred \in \apredsetp$ for which 
$\demonicHoareTripletHoldsSemantically{\apred}{\aprog}{\apredp}$ holds.
This implies that 
$\angelicHoareTripletHoldsSemantically{\set{\apred}}{\aprog}{\set{\apredp}}$ holds.
From completeness of the strongest post, we get
$\angelicAssertionStrongerOf{\strongestPostSemOf{\set{\apred}}{\aprog}}{\set{\apredp}}$.
By monotonicity of the strongest post function, we get 
$\angelicAssertionStrongerOf{\strongestPostSemOf{\apredset}{\aprog}}{\set{\apredp}}$.
We collect the programs returned from completeness in the set $\aprogset$.
We get that $\angelicAssertionStrongerOf{(\cup_{\aprog \in \aprogset}\strongestPostSemOf{\apredset}{\aprog})}{\apredsetp}$.
Since $S$ is finite, $\aprogset$ is finite.
Thus, there is a $j$ for which $\aprogset \subseteq \oracleFuncOf{\anonterm}{j}$ holds.
With this, we also have
$\angelicAssertionStrongerOf{(\cup_{\aprog \in \oracleFuncOf{\anonterm}{j}}\strongestPostSemOf{\apredset}{\aprog})}{\apredsetp}$.
The verification conditions on the individual programs hold because of the soundness of the strongest post function and the induction hypothesis.
\end{proof}

%% file: proofsChapter5.tex
\subsection{Proofs for \Cref{ch:adAlgo}}
\subsection*{Proof of \Cref{Theorem:syn}}
First, we state the theorem using the new notation:
\begin{theorem}[$\syn$-Sound-And-Complete]
Consider $\posound \hoaretriplet{\apredset}{\apo}{\apredsetp}$ and $\apredp\in\apredsetp$ with $\apredp\neq\fail{}$. 
Then $\synof{\hoaretriplet{\apredset}{\apo}{\apredsetp}}{\apredp}=(\apred, \aprog)$ 
with $\apred\in\apredset$, $\apred\neq\fail{}$, $\aprog\in\angelConcFuncOf{\proofToProgFuncOf{\apo}}$, and $\demonicHoareTripletHoldsSemantically{\apred}{\aprog}{\apredp}$. 
The number of SMT solver calls is at most $\sizeof{\apo}$. 
\end{theorem}
\begin{proof}
We proof this theorem by induction over the shape of $\hoaretriplet{\apredset}{\apo}{\apredsetp}$.

\baseCase{}
We have the valid realizability triple $\hoaretriplet{\apredset}{\code{com}}{\apredsetp}$.
Thus, we know for every $\apredp' \in \apredsetp$ there is an $\apred \in \apredset$ with $\demonicAssertionStrongerOf{\progSemFuncOf{\code{com}}{\apred}}{\apredp'}$.
Since $\apredp$ is of $\apredsetp$, the call $\synof{\hoaretriplet{\apredset}{\code{com}}{\apredsetp}}{\apredp}$ eventually terminates returning $(\apred, \code{com})$.
The predicate $\apred$ cannot be $\fail{}$ because $\apredp$ is not $\fail{}$ and $\code{com}$ preserves $\fail{}$.
Also $\code{com} \in \angelConcFuncOf{\code{com}}$ holds trivially.
Moreover, since 
$\demonicAssertionStrongerOf{\progSemFuncOf{\code{com}}{\apred}}{\apredp}$, the Hoare triple $\demonicHoareTripletHoldsSemantically{\apred}{\code{com}}{\apredp}$ holds.
Since we only have to go through $\apredset$ once to find a suitable $\apred$, the number of SMT solver calls is at most $\sizeof{\apredset} \leq \sizeof{\apo}$.

\inductionStep{}
In the first case,
we have $\hoaretriplet{\apredset}{\concatof{\apo_1}{\apo_2}}{\apredsetp}$.
Thus, the call is $\synof{\hoaretriplet{\apredset}{\concatof{\apo_1}{\apo_2}}{\apredsetp}}{\apredp}$.
Let $\apredsetpp$ be the intermediary selection.
Applying the induction hypothesis, the call $\synof{\apo_2}{\apredp}$ returns $(\apredpp, \aprog_2)$ with $\apredpp \in \apredsetpp$, $\apredpp \neq \fail{}$, $\aprog_2 \in \angelConcFuncOf{\proofToProgFuncOf{\apo_2}}$, and $\demonicHoareTripletHoldsSemantically{\apredpp}{\aprog_2}{\apredp}$.
Also the number of SMT calls is at most $\sizeof{\apo_2}$.
Then we call $\synof{\apo_1}{\apredpp}$ and by the induction hypothesis we get 
$(\apred, \aprog_1)$ with $\apred \in \apredset$, $\apred \neq \fail{}$, $\aprog_1 \in \angelConcFuncOf{\proofToProgFuncOf{\apo_1}}$, and $\demonicHoareTripletHoldsSemantically{\apred}{\aprog_1}{\apredpp}$.
Also the number of SMT calls is at most $\sizeof{\apo_1}$.
Put together, the function returns $(\apred, \concatof{\aprog_1}{\aprog_2})$.
Since both Hoare triples hold, we get 
$\demonicHoareTripletHoldsSemantically{\apred}{\concatof{\aprog_1}{\aprog_2}}{\apredp}$.
Also, we have $\concatof{\aprog_1}{\aprog_2} \in \angelConcFuncOf{\proofToProgFuncOf{\apo}}$.
And the number of SMT calls is at most $\sizeof{\apo_1} + \sizeof{\apo_2} \leq \sizeof{\apo}$.

In the next case, we have
$\hoaretriplet{\apredset}{\anonterm(\apo)}{\apredsetp}$.
We have two sub-cases:
First, $\apo$ is only one proof outline as opposed to many separated by the $\bnf$ symbol.
Then the pre and post condition of $\apo$ and $\anonterm(\apo)$ match.
We directly get the sought after result by applying the induction hypothesis and seeing that $\angelConcFuncOf{\proofToProgFuncOf{\apo}} \subseteq \angelConcFuncOf{\proofToProgFuncOf{\anonterm(\apo)}}$.
In the second sub-case, we know that $\apo = \apo_1 \bnf \apo_2$.
Since $\apredsetp$ is the union of all posts of $\apo$, we know that eventually we try out a subproof with $\apredp$ in its post.
Then, the induction hypothesis is applied again directly yielding the required results.

In the next case, we have
$\choiceOf{\hoaretriplet{\apredset}{\apo_1}{\apredsetp}}{\hoaretriplet{\apredset}{\apo_2}{\apredsetp}}$.
Because the proof outline is valid, we know that for every $\apredp \in \apredsetp$ there is an $\apred$ of $\apredset$ for which there is a program $\choiceOf{\aprog_1}{\aprog_2}$ for which
$\demonicHoareTripletHoldsSemantically{\apred}{\choiceOf{\aprog_1}{\aprog_2}}{\apredp}$ holds.
This implies that 
$\demonicHoareTripletHoldsSemantically{\apred}{\aprog_i}{\apredp}$ must hold.
This also means, 
the realizability triple
$\angelicHoareTripletHoldsSemantically{\apred}{\aprog_i}{\apredp}$ holds which in turn implies that 
$\angelicHoareTripletHoldsSemantically{\apred}{\angelConcFuncOf{\apo_i}}{\apredp}$ holds.
By completeness, the outlines $\apop_i$ can be built when supplied with the correct $\apred \in \apredset$.
Then, we use the induction hypothesis 
to get $(\apred, \aprog_i)$ from the recursive calls with $\apred \in \set{\apred} \subseteq \apredset$, and the Hoare triple 
$\demonicHoareTripletHoldsSemantically{\apred}{\aprog_i}{\apredp}$, and $\aprog_i \in \angelConcFuncOf{\proofToProgFuncOf{\apop_i}}$, 
and the number of SMT solver calls for each recursion is at most $\sizeof{\apop_i}$.
Since $\proofToProgFuncOf{\apo_i} = \proofToProgFuncOf{\apop_i}$ we also have that 
$\choiceOf{\aprog_1}{\aprog_2}$ is of $\proofToProgFuncOf{\choiceOf{\apo_1}{\apo_2}}$.
Because $\demonicHoareTripletHoldsSemantically{\apred}{\aprog_i}{\apredp}$, we also have 
$\demonicHoareTripletHoldsSemantically{\apred}{\choiceOf{\aprog_1}{\aprog_2}}{\apredp}$.
When the proof outlines $\apop_i$ are properly looked up, there is no need for additional SMT calls.
Thus, the SMT calls are at most $\sizeof{\choiceOf{\apop_1}{\apop_2}} \leq \sizeof{\choiceOf{\apo_1}{\apo_2}}$.

In the last case consider a loop 
$\hoaretriplet{\apredset}{\kleeneof{\hoaretriplet{I}{\apo}{I}}}{\apredsetp}$.
In the case the realizability triple of the loop was weakened, we keep the original invariant selections.
Thus $\apredset$ resp. $\apredsetp$ are stronger resp. weaker than $I$.
We have $\angelicAssertionStrongerOf{I}{\apredsetp}$.
Therefore, there is an $i$ of $I$ with $\demonicAssertionStrongerOf{i}{\apredp}$.
This $i$ will eventually be found by the $\syn$ function.
Since the proof is sound, $I$ is a sound invariant for $\proofToProgFuncOf{\apo}$.
Thus, for every $i$ of $I$, the realizability triple
$\angelicHoareTripletHolds{\set{i}}{\proofToProgFuncOf{\apo}}{\set{i}}$ is valid.
Therefore, the proof outline $\apop$ can be constructed. 
By the induction hypothesis, $\synof{\apop}{i}$ returns $(i, \aprog)$ with 
$\aprog \in \angelConcFuncOf{\proofToProgFuncOf{\apop}}$, and 
$\demonicHoareTripletHoldsSemantically{i}{\aprog}{i}$, and the number of SMT solver calls is at most $\sizeof{\apop}$.
Because 
$\demonicHoareTripletHoldsSemantically{i}{\aprog}{i}$
holds, we also have 
$\demonicHoareTripletHoldsSemantically{i}{\kleeneof{\aprog}}{i}$.
Since $\angelicAssertionStrongerOf{\apredset}{I}$ there is an $\apred$ of $\apredset$ with $\demonicAssertionStrongerOf{\apred}{i}$.
Since $\demonicAssertionStrongerOf{i}{\apredp}$, we can weaken to
$\demonicHoareTripletHoldsSemantically{\apred}{\kleeneof{\aprog}}{\apredp}$ and return $\apred \in \apredset$.
Because $\aprog \in \angelConcFuncOf{\proofToProgFuncOf{\apop}}$, we have 
$\kleeneof{\aprog} \in \angelConcFuncOf{\kleeneof{\proofToProgFuncOf{\apop}}}$.
Because $\proofToProgFuncOf{\apop} = \proofToProgFuncOf{\apo}$, we have 
$\kleeneof{\aprog} \in \angelConcFuncOf{\kleeneof{\proofToProgFuncOf{\apo}}}$.
Because the proof outline $\apop$ can be looked up and does not need recomputing,
we need at most $\sizeof{I}$ additional SMT calls.
All in all, we have at most $\sizeof{\apo}$ SMT calls.
\end{proof}

%% file: proofsChapter7.tex
\subsection{Proofs for \Cref{ch:implementation}}
\begin{definition}
    An abstract predicate is fail or maps variables to sets of states of the SMR automaton.
\begin{equation*}
\demonicAssertionsAbs{} = (\vars{} \rightarrow \powersetOf{\automaton{}}) \cup \set{\fail{}}
\end{equation*}
\end{definition}
To substitute the original more precise relation $\demonicAssertionStronger{}$ defined on $\demonicAssertions{}$, we introduce a new relation on the abstracted domain.
\begin{definition}
    We define a stronger relation $\demonicAssertionAbsStronger{}$ on $\demonicAssertionsAbs{}$.
    \begin{equation*}
        \demonicAssertionAbsStrongerOf{a}{b}  
        \ \Leftrightarrow \ 
        b = \fail{} \vee
        \forall \, v \in \vars{} \dv a(v) \subseteq b(v)
    \end{equation*}
\end{definition}
Ignoring the predicate $\fail{}$, this relation matches the order in Meyer and Wolff's paper.

\begin{definition}
We define a predicate abstraction function to abstract the original predicate and a predicate concretisation function to concretize abstract predicates.
The predicate abstraction function is a cartesian abstraction:
\begin{equation*}
    \funcDef{\demonicAbsFunc{}}{\demonicAssertions{} \rightarrow \demonicAssertionsAbs{}}
\end{equation*}
Let $v$ be a variable in $\vars{}$. Then, the function is defined as follows:
\begin{equation*}
    \begin{cases}
        \demonicAbsFuncOf{r} = \fail{} &, r = \fail{} \\ 
        \demonicAbsFuncOf{r}(v) = \bigcup_{q \in r} q(v) &, \text{else}
    \end{cases}
\end{equation*}
We continue with the demon concretisation function:
\begin{equation*}
    \funcDef{\demonicConcFunc{}}{\demonicAssertionsAbs{} \rightarrow \demonicAssertions{}}
\end{equation*}
\begin{equation*}
    \demonicConcFuncOf{a} = 
    \begin{cases}
        \fail{} &, a = \fail{} \\ 
        \setCond{q \in \states{}}{\forall v \in \vars{} \dv q(v) \subseteq a(v)} &, \text{else}
    \end{cases}
\end{equation*}
\end{definition}

\begin{lemma}
    $
    \demonicAssertionStrongerOf{r}{s} \implies \demonicAssertionAbsStrongerOf{\demonicAbsFuncOf{r}}{\demonicAbsFuncOf{s}}    
    \ .
    $
\end{lemma}
\begin{proof}
If the predicate $s$ is $\fail{}$, then $\demonicAbsFuncOf{s}$ is also $\fail{}$ so the inequality holds.
Otherwise, we have 
\begin{align*}
    r \subseteq s \implies 
    & \,
    (\bigcup_{q \in r}q(v)) \subseteq 
    (\bigcup_{q \in s}q(v))
    \\ 
    \implies & \,
    \demonicAssertionAbsStrongerOf{\demonicAbsFuncOf{r}}{\demonicAbsFuncOf{s}}
    \ .
\end{align*}
Here, $v$ is any variable of $\vars{}$.
\end{proof}

\begin{lemma}
    $
    \demonicAssertionAbsStrongerOf{a}{b} \implies \demonicAssertionStrongerOf{\demonicConcFuncOf{a}}{\demonicConcFuncOf{b}}    
    \ .
    $
\end{lemma}
\begin{proof}
If the abstract predicate $b$ is $\fail{}$, then $\demonicConcFuncOf{b}$ is also $\fail{}$ so the inequality holds.
Otherwise, we have 
\begin{align*}
    a(v) \subseteq b(v) \implies 
    & \,
    \setCond{p \in \states{}}{\forall v' \in \vars{} \dv p(v) \subseteq a(v)} \subseteq 
    \\ & \qquad
    \setCond{p \in \states{}}{\forall v' \in \vars{} \dv p(v) \subseteq b(v)}
    \\ 
    \implies & \,
    \demonicAssertionStrongerOf{\demonicConcFuncOf{a}}{\demonicConcFuncOf{b}}
    \ .
\end{align*}
Here, $v$ is any variable of $\vars{}$.
\end{proof}

\begin{lemma}\label{lem:demonGalois}
    The pair $(\demonicAbsFunc{}, \demonicConcFunc{})$ is a Galois connection between $\demonicAssertions{}$ and $\demonicAssertionsAbs{}$:
    \begin{equation*}
        (\demonicAssertions{}, \demonicAssertionStronger{})
        \galois{\demonicAbsFunc{}}{\demonicConcFunc{}}
        (\demonicAssertionsAbs{}, \demonicAssertionAbsStronger{})
    \end{equation*}
    This means, the following inequalities hold:
    \begin{align*}
        \forall r \in \demonicAssertions{} &\dv  \demonicAssertionStrongerOf{r}{
        \ \demonicConcFuncOf{\demonicAbsFuncOf{r}}}\\
        \forall a \in \demonicAssertionsAbs{} & \dv 
         \demonicAssertionAbsStrongerOf{\demonicAbsFuncOf{\demonicConcFuncOf{a}}}{
        \ a}
    \end{align*}
\end{lemma}

\begin{proof}
We start the proof with the first equation.
Let $r$ be of $\demonicAssertions{}$.
If $r = \fail{}$ then $\demonicConcFuncOf{\demonicAbsFuncOf{r}} = \fail{}$.
Thus, the inequality $\demonicAssertionStrongerOf{r}{\demonicConcFuncOf{\demonicAbsFuncOf{r}} = \fail{}}$ holds.
If $r$ is not $\fail{}$, then     
\begin{equation*}
\demonicConcFuncOf{\demonicAbsFuncOf{r}} = 
\setCond{p \in \states{}}{\forall v \in \vars{} \dv p(v) \subseteq \bigcup_{q \in r} q(v)}
\ .
\end{equation*}
Let the state $q$ be of $r$.
Then $q$ is also in  
$\demonicConcFuncOf{\demonicAbsFuncOf{r}}$.
Therefore $r \subseteq 
\demonicConcFuncOf{\demonicAbsFuncOf{r}}$.
And that means $\demonicAssertionStrongerOf{r}{\demonicConcFuncOf{\demonicAbsFuncOf{r}}}$.

We continue with the second equation.
Let $a$ be of $\demonicAssertionsAbs{}$.
If $a$ is $\fail{}$, then $\demonicAbsFuncOf{\demonicConcFuncOf{a}}$ is trivially stronger.
If $a$ is not $\fail{}$, then we show that for any $v$ of $\vars{}$ the inclusion 
$\demonicAbsFuncOf{\demonicConcFuncOf{a}}(v) \subseteq a(v)$ holds.
Let $v$ be of $\vars{}$.
Then 
\begin{align*}
    \demonicAbsFuncOf{\demonicConcFuncOf{a}}(v) = & \,
    \cup_{q \in \demonicConcFuncOf{a}}q(v) 
    \\
    = & \,
    \cup_{q \in \setCond{p \in \states{}}{\forall v' \in \vars{} \dv p(v') \subseteq a(v')}} q(v)\\
    \subseteq &\,
    a(v)
    \ .
\end{align*}
This concludes the proof.
\end{proof}

\begin{lemma}\label{lem:demonsAbsPO}
    The relation $\demonicAssertionAbsStronger{}$ is a partial order relation on $\demonicAssertionsAbs{}$.
\end{lemma}
\begin{proof}
We first show reflexivity.
Let $a$ be of $\demonicAssertionsAbs{}$.
If $a = \fail{}$, then the inequality $\demonicAssertionAbsStrongerOf{a}{a}$ trivially holds.
Otherwise, let $v$ be a variable of $\vars{}$.
Then the inclusion $a(v) \subseteq a(v)$ also holds.

We continue with transitivity.
Let $a$, $b$ and $c$ be of $\demonicAssertionsAbs{}$.
We assume that the inequalities $\demonicAssertionAbsStrongerOf{a}{b}$ and 
$\demonicAssertionAbsStrongerOf{b}{c}$ are true.
If $c = \fail{}$, then the inequality $\demonicAssertionAbsStrongerOf{a}{c}$ trivially holds.
Otherwise, we get that $b \neq \fail{}$ and also $a \neq \fail{}$.
Let $v$ be of $\vars{}$.
We have $a(v) \subseteq b(v)$ and $b(v) \subseteq c(v)$.
Through transitivity of the relation $\subseteq$ we also get the inclusion $a(v) \subseteq c(v)$.

Lastly, we show antisymmetry.
Let $a$ and $b$ be of $\demonicAssertionsAbs{}$.
We have $\demonicAssertionAbsStrongerOf{a}{b}$ and $\demonicAssertionAbsStrongerOf{b}{a}$.
If $a$ is $\fail{}$ then $b$ must also be $\fail{}$ and vice versa.
If neither one is fail, let $v$ be of $\vars{}$.
We know that the inequalities $a(v) \subseteq b(v)$ and $b(v) \subseteq a(v)$ are true and thus we have $a(v) = b(v)$ for all $v$ of $\vars{}$.
That means, $a$ and $b$ are equal.
\end{proof}

\begin{lemma} \label{lem:demonsAbsLattice}
    The partial order $(\demonicAssertionsAbs{}, \demonicAssertionAbsStronger{})$ is a complete lattice.
    In fact, the join and meet can be computed by the following equations: 
    Let $A$ be a subset of $\demonicAssertionsAbs{}$.
    \begin{align*}
        &\begin{cases}
        \joinOf{A} =
            \fail{} \hphantom{\ \, \quad\qquad\qquad}&, \fail{} \in A \\
        \joinOf{A}(v) =
            \bigcup_{a \in A} a(v) &, \text{else}
        \end{cases} 
        \\
        &\begin{cases}
        \meetOf{R} =
            \fail{} &, R = \set{\fail{}} \\
        \meetOf{R} =
            \bigcap_{a \in (A \setminus \set{\fail{}})} a(v) &, \text{else}
        \end{cases}
    \end{align*}

\end{lemma}

\begin{proof}
    Let $A$ be a subset of $\demonicAssertionsAbs{}$.
    Let $a$ be an element of $A$.
    We first show that the join is an upper bound $\demonicAssertionAbsStrongerOf{a}{\joinOf{A}}$.
    If $\fail{} \in A$, then $\joinOf{A} = \fail{}$ so 
    the inequality 
    $\demonicAssertionAbsStrongerOf{a}{\joinOf{A}}$
    trivially holds.
    Otherwise, let $v$ be of $\vars{}$.
    The inclusion $a(v) \subseteq \bigcup_{a' \in A} a'(v)$ holds trivially, since $a \in A$.
    Now, to show that the join is the least upper bound, let $u$ be an upper bound of $A$.
    If $u$ is $\fail{}$
    the inequality $\demonicAssertionAbsStrongerOf{\joinOf{A}}{u}$ holds trivially.
    Otherwise, if $u$ is not $\fail{}$, we know $\fail{} \not\in A$.
    Let $v$ be of $\vars{}$.
    For any abstract demon $a$ of $A$, $a(v) \subseteq u(v)$ must hold.
    Therefore the inclusion $\joinOf{A}(v) = \bigcup_{a \in A} a(v) \subseteq u(v)$ holds.
    Thus, 
    the inequality
    $\demonicAssertionAbsStrongerOf{\joinOf{A}}{u}$ is true.
    
    Moving on with the meet.
    Let $a$ be an element of $A$.
    We first show that the meet is a lower bound: 
    $\demonicAssertionAbsStrongerOf{\meetOf{A}}{a}$.
    If $A = \set{\fail{}}$, then $a = \fail{}$ so the inequality
    $\demonicAssertionAbsStrongerOf{\meetOf{A}}{a}$ trivially holds.
    Otherwise, let $v$ be of $\vars{}$.
    Then, by definition we know that $\meetOf{A}(v) = \bigcap_{a' \in A\setminus\set{\fail{}}} a'(v)$.
    If $a$ is $\fail{}$, the inequality trivially holds.
    Otherwise, 
    $\meetOf{A}(v) \subseteq a(v)$ is also true. 
    Now, to show that the meet is the greatest lower bound, let $l$ be a lower bound of $A$.
    If $l$ is $\fail{}$, $A$ must be $\set{\fail{}}$ and thus $\meetOf{A} = \fail{}$, so $\demonicAssertionAbsStrongerOf{l}{\meetOf{A}}$ trivially holds.
    Otherwise, let $v$ be of $\vars{}$.
    Then for any $a$ of $A$ that is not $\fail{}$, $l(v) \subseteq a(v)$ must hold.
    Thus, $l(v) \subseteq \bigcap_{a' \in A\setminus\set{\fail{}}} a'(v) = \meetOf{A}(v)$ holds.
    And therefore the inequality $\demonicAssertionAbsStrongerOf{l}{\meetOf{A}}$ is true.
\end{proof}

Next, we lift the abstract predicates to abstract selections.
We remind the reader, that selections are sets of predicates, i.e.\ $\angelicAssertions{} = \powersetOf{\demonicAssertions{}}$.
\begin{definition}
    An \emph{Abstract Selection} is a set of abstract predicates.
    \begin{equation*}
        \angelicAssertionsAbs{} = \powersetOf{\demonicAssertionsAbs{}}
    \end{equation*}
\end{definition}

The more versatile relation on abstract angels is the same as on $\angelicAssertions{}$, except that the elements of the abstract selections are compared using the abstract more precise relation defined on abstract predicates.
\begin{definition}
    $
        \angelicAssertionAbsStrongerOf{A}{B} 
        \ \Leftrightarrow \ 
         \forall b \in B \dv \exists \, a \in A \dv \demonicAssertionAbsStrongerOf{a}{b}
         $
\end{definition}

\begin{definition}
    The selection abstraction function $\angelicAbsFunc{}$ relates elements of $\angelicAssertions{}$ to their abstract representation in $\angelicAssertionsAbs{}$.
    \begin{equation*}
        \funcDef{\angelicAbsFunc{}}{\angelicAssertions{} \rightarrow \angelicAssertionsAbs{}}
    \end{equation*}
    \begin{equation*}
        \angelicAbsFuncOf{\angelicAssertionFont{R}} = \setCond{\demonicAbsFuncOf{r}}{r \in \angelicAssertionFont{R}}  
    \end{equation*}

    The selection concretisation function $\angelicConcFunc{}$ relates elements of $\angelicAssertionsAbs{}$ to the elements of the original domain $\angelicAssertions{}$.
    \begin{equation*}
        \funcDef{\angelicConcFunc{}}{\angelicAssertionsAbs{} \rightarrow \angelicAssertions{}}
    \end{equation*}
    \begin{equation*}
        \angelicConcFuncOf{A} = \setCond{\demonicConcFuncOf{a}}{a \in A}  
    \end{equation*}
\end{definition}

\begin{lemma}
    $
    \angelicAssertionStrongerOf{R}{S} \implies \angelicAssertionAbsStrongerOf{\angelicAbsFuncOf{R}}{\angelicAbsFuncOf{S}}
    \ .
    $
\end{lemma}
\begin{proof}
Let $\demonicAbsFuncOf{s}$ be of $\angelicAbsFuncOf{S}$.
We know there is a predicate $r$ in $R$ more precise than $s$.
Thus, the inequality $\demonicAssertionAbsStrongerOf{\demonicAbsFuncOf{r}}{\demonicAbsFuncOf{s}}$ is true.
And therefore the inequality
$\angelicAssertionAbsStrongerOf{\angelicAbsFuncOf{R}}{\angelicAbsFuncOf{S}}$
holds.
\end{proof}

\begin{lemma}
    $
    \angelicAssertionStrongerOf{A}{B} \implies \angelicAssertionStrongerOf{\angelicConcFuncOf{A}}{\angelicConcFuncOf{B}}
    \ .
    $
\end{lemma}
\begin{proof}
Let the predicate $\demonicConcFuncOf{b}$ be of $\angelicConcFuncOf{B}$.
We know there is an abstract predicate $a$ of $A$ with $\demonicAssertionAbsStrongerOf{a}{b}$.
Thus, $\demonicConcFuncOf{a}$ is more precise than $\demonicConcFuncOf{b}$.
Since $\demonicConcFuncOf{a} \in \angelicConcFuncOf{A}$, we know the inequality
$\angelicAssertionStrongerOf{\angelicConcFuncOf{A}}{\angelicConcFuncOf{B}}$
holds.
\end{proof}

\begin{lemma}\label{lem:angelAbsGalois}
    The pair $(\angelicAbsFunc{}, \angelicConcFunc{})$ is a Galois connection between $\angelicAssertions{}$ and $\angelicAssertionsAbs{}$:
    \begin{equation*}
        (\angelicAssertions{}, \angelicAssertionStronger{})
        \galois{\angelicAbsFunc{}}{\angelicConcFunc{}}
        (\angelicAssertionsAbs{}, \angelicAssertionAbsStronger{})
    \end{equation*}
    This means, the following inequalities hold:
    \begin{align*}
        \forall R \in \angelicAssertions{} &\dv  \angelicAssertionStrongerOf{R}{
        \ \angelicConcFuncOf{\angelicAbsFuncOf{R}}}\\
        \forall \, A \in \angelicAssertionsAbs{} & \dv 
         \angelicAssertionAbsStrongerOf{\angelicAbsFuncOf{\angelicConcFuncOf{A}}}{
        \ A}
    \end{align*}
\end{lemma}

\begin{proof}
We start the proof with the first inequality.
Let $R$ be of $\angelicAssertions{}$.
Let the predicate $s$ be of $\angelicConcFuncOf{\angelicAbsFuncOf{R}}$.
Then $s \in \setCond{\demonicConcFuncOf{a}}{a \in \setCond{\demonicAbsFuncOf{r}}{r \in R}}$.
Thus, $s \in 
\setCond{\demonicConcFuncOf{\demonicAbsFuncOf{r}}}{r \in R}$.
Since the inequality $\demonicAssertionAbsStrongerOf{r}{\demonicConcFuncOf{\demonicAbsFuncOf{r}}}$ is true, we get that there is an $r$ in $R$ more precise than $s$, thus the inequality $\angelicAssertionAbsStrongerOf{R}{\angelicConcFuncOf{\angelicAbsFuncOf{R}}}$ holds.

Moving on with the second inequality, let $A$ be of $\angelicAssertionsAbs{}$.
Let the abstract predicate $a$ be of $A$.
We have $\angelicAbsFuncOf{\angelicConcFuncOf{A}} = \setCond{\demonicAbsFuncOf{r}}{r \in \setCond{\demonicConcFuncOf{a}}{a \in A}}$.
Thus, $\demonicAbsFuncOf{\demonicConcFuncOf{a}}$ is in $\angelicAbsFuncOf{\angelicConcFuncOf{A}}$.
Since we know that the inequality 
$\demonicAssertionAbsStrongerOf{\demonicAbsFuncOf{\demonicConcFuncOf{a}}}{a}$ holds, we get that the inequality  $\angelicAssertionAbsStrongerOf{\angelicAbsFuncOf{\angelicConcFuncOf{A}}}{A}$ is true.
\end{proof}

\begin{lemma}\label{lem:galoisInsertion}
    The pair $(\demonicAbsFunc{}, \demonicConcFunc{})$ is a Galois insertion, i.e.\ the following equation holds for any $a$ of $\demonicAssertionsAbs{}$:
    $
        \demonicAbsFuncOf{\demonicConcFuncOf{a}} = a 
    $.

    The pair $(\angelicAbsFunc{}, \angelicConcFunc{})$ also is a Galois insertion, i.e.\ the following equation holds true for any $A$ of $\angelicAssertionsAbs{}$:
    $
        \angelicAbsFuncOf{\angelicConcFuncOf{A}} = A
    $.
\end{lemma}

\begin{proof}
We start with the first equality.
Since we know the inequality $\demonicAssertionAbsStrongerOf{\demonicAbsFuncOf{\demonicConcFuncOf{a}}}{a}$ holds, it is sufficient to show that the inequality
$\demonicAssertionAbsStrongerOf{a}{\demonicAbsFuncOf{\demonicConcFuncOf{a}}}$ is true because the relation $\demonicAssertionAbsStronger{}$ is antisymmetric.
So let $a$ be of $\demonicAssertionsAbs{}$.
If $a = \fail{}$, then $\demonicAbsFuncOf{\demonicConcFuncOf{a}}$ is also $\fail{}$.
If $a$ is not $\fail{}$, then let $v$ be any variable of $\vars{}$.
Then 
\begin{align*}
    a(v) \subseteq &  \,
    \cup_{p \in \states{} \wedge p(v) \subseteq a(v)}p(v)
    \\
    \subseteq & \, 
    \cup_{p \in \setCond{p' \in \states{}}{\forall v' \in \vars{} \dv p'(v') \subseteq a(v')}}p(v)
    \\
    = & \, 
    \demonicAbsFuncOf{\demonicConcFuncOf{a}}
\end{align*}
The second inclusion holds true because of the abundance of available states.
Therefore, the inequality $\demonicAssertionAbsStrongerOf{a}{\demonicAbsFuncOf{\demonicConcFuncOf{a}}}$ holds.
That means $a = \demonicAbsFuncOf{\demonicConcFuncOf{a}}$ is true.

Because the relation on abstract angels is not antisymmetric, we show equality directly:
\begin{align*}
\angelicAbsFuncOf{\angelicConcFuncOf{A}} 
= & \,
\setCond{\demonicAbsFuncOf{r}}{r \in \setCond{\demonicConcFuncOf{a}}{a \in A}}
\\
= & \,
\setCond{\demonicAbsFuncOf{\demonicConcFunc{a}}}{a \in A}
\\
= & \,
\setCond{a}{a \in A}
\\
= & \,
A
\ .
\end{align*}
\end{proof}

\begin{definition}
    Meyer and Wolff also provide an interpretation of commands that either maps to another state or to their notation of failing: $\top$.
    \begin{gather*}
    \progSemFuncAbstractOf{\code{com}}{a} =
    \begin{cases}
        \fail{} &, a = \fail{} \vee \progSemFuncTypesOf{\code{com}}{a} = \top \\
        \progSemFuncTypesOf{\code{com}}{a} &, \text{ else}
    \end{cases}
    \end{gather*}
\end{definition}

\begin{lemma}
    The abstract interpretation of commands $\progSemFuncAbstract{\code{com}}$ is a safe abstraction for $\progSemFunc{\code{com}}$, i.e.\
    the following equation holds:
    \begin{equation*}
        \progSemFuncAbstractOf{\code{com}}{a} 
        = 
        \demonicAbsFuncOf{\progSemFuncOf{\code{com}}{\demonicConcFuncOf{a}}}
    \end{equation*}
\end{lemma}
\begin{proof}
The original interpretation of $\code{com}$ in Meyer and Wolff's work is monotonic.
Thus, $\progSemFuncAbstract{\code{com}}$ is monotonic.
We show equality by showing $\demonicAssertionStronger{}$ in both ways.
By antisymmetry, equality follows.

We start with
$
\demonicAssertionAbsStrongerOf{
 \progSemFuncAbstractOf{\code{com}}{a}   
}{
\demonicAbsFuncOf{\progSemFuncOf{\code{com}}{\demonicConcFuncOf{a}}}
}
$:
If the right side is $\fail{}$ we are done.
Otherwise,
see that $a \in \demonicConcFuncOf{a}$.
Thus, $\progSemFuncAbstractOf{\code{com}}{a} \in 
\progSemFuncOf{\code{com}}{\demonicConcFuncOf{a}}
$.
Let $v$ be a variable of $\vars{}$.
Then $\progSemFuncAbstractOf{\code{com}}{a}(v)$
is a subset of
$
\demonicAbsFuncOf{\progSemFuncOf{\code{com}}{\demonicConcFuncOf{a}}}(v)$.
Therefore, the inequality 
$
\demonicAssertionAbsStrongerOf{
 \progSemFuncAbstractOf{\code{com}}{a}   
}{
\demonicAbsFuncOf{\progSemFuncOf{\code{com}}{\demonicConcFuncOf{a}}}
}
$
holds.

We proceed with the other direction:
Again, if the left side is $\fail{}$ we are done.
Otherwise, see that for all $q$ of $\demonicConcFuncOf{a}$, $\demonicAssertionAbsStrongerOf{q}{a}$ holds.
Now, let $b$ be an element of $\progSemFuncOf{\code{com}}{\demonicConcFuncOf{a}}$.
Then, $\demonicAssertionAbsStrongerOf{b}{\progSemFuncAbstractOf{\code{com}}{a}}$ holds because the original interpretation of commands is monotonic.
Thus, for any variable $v$ of $\vars{}$, we have
$b(v) \subseteq \progSemFuncAbstractOf{\code{com}}{a}(v)$.
Therefore, 
$
\demonicAbsFuncOf{\progSemFuncOf{\code{com}}{\demonicConcFuncOf{a}}}(v)$
is a subset of
$\progSemFuncAbstractOf{\code{com}}{a}(v)$.
So, the inequality holds.
This concludes the proof.
\end{proof}

\begin{lemma}
    The following inequality holds:
    \begin{equation*}
        \angelicAssertionAbsStrongerOf{
            \angelicAbsFuncOf{\strongestPostSemOf{\angelicConcFuncOf{A}}{\code{com}}}
        }{
            \strongestPostSemAbsOf{A}{\code{com}}
        }
    \end{equation*}
\end{lemma}
\begin{proof}
   \begin{align*}
        \angelicAbsFuncOf{\strongestPostSemOf{\angelicConcFuncOf{A}}{\code{com}}} = & \,
        \setCond{\demonicAbsFuncOf{\progSemFuncOf{\code{com}}{r}}}{r \in \angelicConcFuncOf{A}}
        \\
        = & \,
        \setCond{\demonicAbsFuncOf{\progSemFuncOf{\code{com}}{\demonicConcFuncOf{a}}}}{a \in A}
        \\
        = & \,
        \setCond{\progSemFuncAbstractOf{\code{com}}{a}}{a \in A}
        \\
        = & \,
        \strongestPostSemAbsOf{A}{\code{com}}
   \end{align*} 
   The inequality holds because $\angelicAssertionAbsStronger{}$ is reflexive.
\end{proof}

\begin{lemma}
    The following equations hold:
    \begin{equation*}
        \angelicAbsFuncOf{R \cup S} = \angelicAbsFuncOf{R} \cup \angelicAbsFuncOf{S}
    \end{equation*}
    \begin{equation*}
        \angelicConcFuncOf{A \cup B} = \angelicConcFuncOf{A} \cup \angelicConcFuncOf{B}
        \ .
    \end{equation*}
\end{lemma}
\begin{proof}
    We start with the first equation.
    \begin{align*}
    \angelicAbsFuncOf{R \cup S} = & \,
    \setCond{\demonicAbsFuncOf{r}}{r \in R \cup S} \\
    = & \, 
    \setCond{\demonicAbsFuncOf{r}}{r \in R} \cup 
    \setCond{\demonicAbsFuncOf{s}}{s \in S} \cup 
     \\
     = & \, 
     \angelicAbsFuncOf{R} \cup \angelicAbsFuncOf{S}
    \end{align*}
    The proof for the second equation is analogous.
\end{proof}

\begin{lemma}\label{lem:forConj}
    The following equation holds
    \begin{equation*}
        \angelicAssertionAbsStrongerOf{\vcStrongestPostSemAbsOf{A}{\code{com}}}{B}
        \implies
        \angelicAssertionStrongerOf{\vcStrongestPostSemOf{\angelicConcFuncOf{A}}{\code{com}}}{\angelicConcFuncOf{B}}
    \end{equation*}
\end{lemma}
\begin{proof}
    Because the abstract strongest post is a safe abstraction, we have 
    \begin{equation*}
        \angelicAbsFuncOf{\vcStrongestPostSemOf{\angelicConcFuncOf{A}}{\code{com}}} 
        \angelicAssertionAbsStronger{} 
        \vcStrongestPostSemAbsOf{A}{\code{com}}
        \angelicAssertionAbsStronger{} 
        B
        \ .
    \end{equation*}
    Because we have a Galois connection, the following inequality holds:
    \begin{equation*}
        \angelicAssertionStrongerOf{
            \vcStrongestPostSemOf{\angelicConcFuncOf{A}}{\code{com}}
        }{
            \angelicConcFuncOf{\angelicAbsFuncOf{\vcStrongestPostSemOf{\angelicConcFuncOf{A}}{\code{com}}}}
        }
        \ .
    \end{equation*}
    Applying the monotonicity of $\angelConcFunc{}$ we get the following two inequalities:
    \begin{align*}
         \angelicAssertionStrongerOf{
            \vcStrongestPostSemOf{\angelicConcFuncOf{A}}{\code{com}}
        }{
            &
            \angelicConcFuncOf{\vcStrongestPostSemAbsOf{A}{\code{com}}}
        }
        \\
        \angelicAssertionStrongerOf{
            \vcStrongestPostSemOf{\angelicConcFuncOf{A}}{\code{com}}
        }{
            &
            \angelicConcFuncOf{B}
        }
    \end{align*}
\end{proof}

\subsection*{Proof of \Cref{th:abstractVC}}
\begin{proof}[\unskip\nopunct]
    We prove the theorem by induction over the structure of $\asketch$.
    
    \baseCase{}
    If $\asketch = \code{com}$, then we know the inequality 
    $\angelicAssertionAbsStrongerOf{\vcStrongestPostSemAbsOf{A}{\code{com}}}{B}$
    is true.
    This then implies that 
    the inequality
    $\angelicAssertionStrongerOf{\vcStrongestPostSemOf{\angelicConcFuncOf{A}}{\code{com}}}{\angelicConcFuncOf{B}}$ holds.
    Therefore, all verification conditions of $\verificationConditionsFuncOf{\hoaretriplet{\angelicConcFuncOf{A}}{\code{com}}{\angelicConcFuncOf{B}}}$ are valid.
    Thus, we can show the realizability triple $\angelicHoareTripletHolds{\angelicConcFuncOf{A}}{\code{com}}{\angelicConcFuncOf{B}}$.
    
    \inductionStep{}
    If the sketch $\asketch = \concatof{\asketch_1}{\asketch_2}$, we know that the verification conditions of the functions
    $\absVerificationConditionsFuncOf{\hoaretriplet{A}{\asketch_1}{\vcStrongestPostSemAbsOf{A}{\asketch_1}}}$ and 
    $\absVerificationConditionsFuncOf{\hoaretriplet{\vcStrongestPostSemAbsOf{\asketch_1}{A}}{\asketch_2}{B}}$
    hold.
    Applying the induction hypothesis, we get
    that
    $\angelicHoareTripletHolds{\angelicConcFuncOf{A}}{\asketch_1}{\angelicConcFuncOf{\vcStrongestPostSemAbsOf{A}{\asketch_1}}}$ and 
    $\angelicHoareTripletHolds{\angelicConcFuncOf{\vcStrongestPostSemAbsOf{A}{\asketch_1}}}{\asketch_2}{\angelicConcFuncOf{B}}$ are valid.
    Using rule \ruleLabel{SEQ}, we prove 
    the realizability triple
    $\angelicHoareTripletHolds{\angelicConcFuncOf{A}}{\concatof{\asketch_1}{\asketch_2}}{\angelicConcFuncOf{B}}$.

    The other rules are also analogous to the original proof.
\end{proof}
\subsection*{Proof of \Cref{th:connectionToSeb}}
\begin{proof}[\unskip\nopunct]
By completeness of Hoare logic, we get 
$\demonicHoareTripletHolds{\demonicConcFuncOf{a}}{\aprog}{\demonicConcFuncOf{b}}$.
Since the rules are the Same as in Meyer and Wolff's type system, we get 
$\sebHoareTripletHolds{\demonicConcFuncOf{a}}{\aprog}{\demonicConcFuncOf{b}}$.
\end{proof}

%% file: sections-appendix-application-preliminaries.tex
\subsection{Background on Safe Memory Reclamation}\label{Section:BackgroundSMR}
Our goal is to synthesize code that makes a lock-free data structure memory safe. 
We recall the basics of safe memory reclamation and outline the approach from~\cite{POPL2019,POPL2020,WolffPhd} to verify that a lock-free data structure properly protects its memory using a safe memory reclamation algorithm. 
We also demonstrate the approach on an example.
\subsection{Background} \label{sec:preliminaries:overview}
When writing programs in languages such as C++, memory management is manual.
This is opposed to other programming languages, such as Java, which provide \emph{Garbage Collection} (GC) for automatic memory management.
Not having GC at hand proves challenging in implementing (and verifying) concurrent, especially lock-free data structures 
due to use-after-free errors and the ABA problem. 
For assistance, programmers can resort to Safe Memory Reclamation (SMR) algorithms, e.g.\ Hazard Pointers (HP) \cite{HP}. 
Instead of freeing memory directly, these algorithms provide a \emph{retire} function that delays freeing memory until it is safe to do so. 
In order for the SMR algorithm to know when freeing memory is allowed, the programmer has to call SMR specific functions.
With HP, the programmer has to call a \emph{protect} function on a pointer they want to access.
Afterwards, however, the programmer has to check if the protection was successful.
In the next chapter, we automatically synthesize the code for such calls and checks.

In \cite{POPL2019,POPL2020,WolffPhd}, Meyer and Wolff proposed a type-based approach to automatically verify lock-free data structures that use SMR algorithms for memory reclamation.
We base our synthesis upon their work. 
The idea is to abstract away the implementation details of the SMR and verify the program with the help of a specification, a so-called SMR automaton.
The automata consist of states, and subsets of these states form the types, also called guarantees, in their approach. 
The states capture the effect that the SMR algorithm has on a pointer. 
Every SMR automaton comes with at least three guarantees.
The local guarantee, $\localGuarantee{}$, signals that the pointer is thread local.
The active guarantee, $\activeGuarantee{}$, signals that the pointer is published and was not yet retired.
The safe guarantee, $\safeGuarantee{}$, signals that the pointer is protected by the SMR algorithm.
Additionally, SMR automata may bring SMR specific guarantees $\ELGuarantee{}$.

Pointer dereferences are only allowed when the SMR automaton for the pointer is guaranteed to be in a state described by $\localGuarantee{}$, $\activeGuarantee{}$, or $\safeGuarantee{}$.

The type system tracks the guarantees each pointer has.
It annotates the original program commands, SMR commands, and invariant annotations that may have to be added to the code. 
The invariant annotations provide information about the protection status of a pointer that the type system can rely on, e.g.\ that a shared pointer is not retired at a certain program location.
The invariant annotations need to be verified separately. 
Importanty, this can be done assuming GC.
When the type check is successful and the invariant annotations can be discharged, the program does not unsafely dereference memory and does not suffer from ABAs.

\subsection{Example SMR Automaton: Hazard Pointer}
\input{sections-appendix-application-hpsmrfig.tex}
We demonstrate the concept of SMR automata on the SMR algorithm Hazard Pointer.
In \Cref{fig:HPsmr}, the SMR automaton $\automaton{}$ for Hazard Pointers is depicted.
It consists of eight states, one of them final.
The automaton specifies illegal behavior. 
That means, when using HP, all transitions to the final state will not occur.
The automaton takes two parameters: $z_p$ representing the pointer whose state is tracked, and $z_t$ representing the thread whose perspective is taken. 
The automaton has two extra guarantees: $\EinvGuarantee{}$ and $\EisuGuarantee{}$.
In \Cref{fig:HPsmr}, all guarantees are depicted by colored zones.
A pointer is safe to access when it possesses the guarantees $\activeGuarantee{}$, $\localGuarantee{}$, or $\safeGuarantee{}$, because all transitions for command $\code{free}$ lead to the fail state.
Therefore, when a pointer is guaranteed to be in one of the states described by the three guarantees, it cannot be freed. 
Thus, a dereference would be safe. 
Here, guarantee $\safeGuarantee{}$ represents the protection by the SMR algorithm.
Guarantee $\EinvGuarantee{}$ describes the states the automaton is guaranteed to be in when thread $z_t$ invokes a protection on pointer~$z_p$.
When that call returns, the automaton is guaranteed to be in the states described by $\EisuGuarantee{}$.
When thread $z_t$ protects another pointer, or unprotects pointer $z_p$, the aforementioned guarantees are lost.
The transitions for $\code{free}$ and $\code{retire}$ do not depend on the executing thread and can be interferences from other threads.
The guarantee $\activeGuarantee{}$ is not closed under interferences.
Local pointers cannot experience interference, thus interference can be disregarded for guarantee $\localGuarantee{}$.
When using HP, after issuing a protection, it must be checked if the protection was successful.
This is represented in the automaton by the fact that the guarantees $\EisuGuarantee{}$ and $\safeGuarantee{}$ are not the same.
After returning from the protection call, the automaton could be in state $q_5$ where it is not safe yet.
The strategy to protect a pointer therefore is to first call the method $\code{protect}$ and then compare it to a pointer that possesses the guarantee $\activeGuarantee{}$.
If both are equal, they must be in the state $q_4$ ($q_f$ is never reached) and thus both automatically get the guarantee $\safeGuarantee{}$ and are thereby safe to access.

\subsection{Example Type Check: Treiber's Stack \code{pop}}
\begin{wrapfigure}{O}{0.6\textwidth}
    \caption{Excerpt of \code{pop} in Treiber's stack demonstrating Meyer and Wolff's type system using Hazard Pointers (simplified).}
    \label{code:exampleTypeSystemDemo}
\begin{lstlisting}[multicols=2]
$\code{top}: \nothingGuarantee{}, \code{TOS}: \nothingGuarantee{}$
top = TOS;
$\code{top}: \nothingGuarantee{}, \code{TOS}: \nothingGuarantee{}$
in:protect(top);
$\code{top}: \EinvGuarantee{}, \code{TOS}: \nothingGuarantee{}$
re:protect();
$\code{top}: \EisuGuarantee{}, \code{TOS}: \nothingGuarantee{}$
atomic {
  $\code{top}: \EisuGuarantee{}, \code{TOS}: \nothingGuarantee{}$
  \@inv active(TOS);
  $\code{top}: \EisuGuarantee{}, \code{TOS}: \activeGuarantee{}$
  assume(top == TOS);
  $(\code{top}: \EisuGuarantee{} \wedge \activeGuarantee{},$
   $\code{TOS}: \activeGuarantee{} \wedge \EisuGuarantee{})$
  $(\code{top}: \EisuGuarantee{} \wedge \activeGuarantee{} \wedge \safeGuarantee{},$
   $\code{TOS}: \activeGuarantee{} \wedge \EisuGuarantee{} \wedge \safeGuarantee{})$
}
$\code{top}: \EisuGuarantee{} \wedge \safeGuarantee{}, \code{TOS}: \nothingGuarantee{}$
d = top.data;
$\code{top}: \EisuGuarantee{} \wedge \safeGuarantee{}, \code{TOS}: \nothingGuarantee{}$
\end{lstlisting}
    \end{wrapfigure}

We demonstrate the type check from~\cite{POPL2020,WolffPhd}.
\Cref{code:exampleTypeSystemDemo} shows an excerpt of the \code{pop} method from Treiber's stack using Hazard Pointers as the SMR algorithm.
There are two pointer variables:
the local variable \code{top} and the shared variable \code{TOS}.
In the beginning, both pointers possess no guarantees.
First, the current value of the shared variable is read into \code{top}.
Afterwards a hazard pointer protection for $\code{top}$ is called.
It is invoked in Line~4 and returns in Line~6.
After a protection, one has to check if it was successful.
This requires two steps. 
First, we signal the type system that the shared variable $\code{TOS}$ is not yet retired through the invariant annotation $\code{\@active(TOS)}$.
Afterwards, we compare the previously read value stored in the variable $\code{top}$ with the current value of $\code{TOS}$.
If both variables point to the same address, we combine the information we have on each variable.
Therefore, both pointers now have the guarantees $\EisuGuarantee{}$ and $\activeGuarantee{}$ and by inference also $\safeGuarantee{}$. 
Leaving the atomic block, other threads are able to interfere again, so the active guarantee is lost for $\code{top}$.
Because the pointer $\code{TOS}$ is shared and other threads can manipulate it in any way, it loses all its guarantees.
However, the pointer $\code{top}$ can now be safely dereferenced in Line~19.
We use the theory presented in this paper and leverage the type system discussed above to automatically synthesize the extra code needed to pass the type check, i.e.\ Lines 4, 6, and 10.

%% file: sections-appendix-application-hpsmrfig.tex
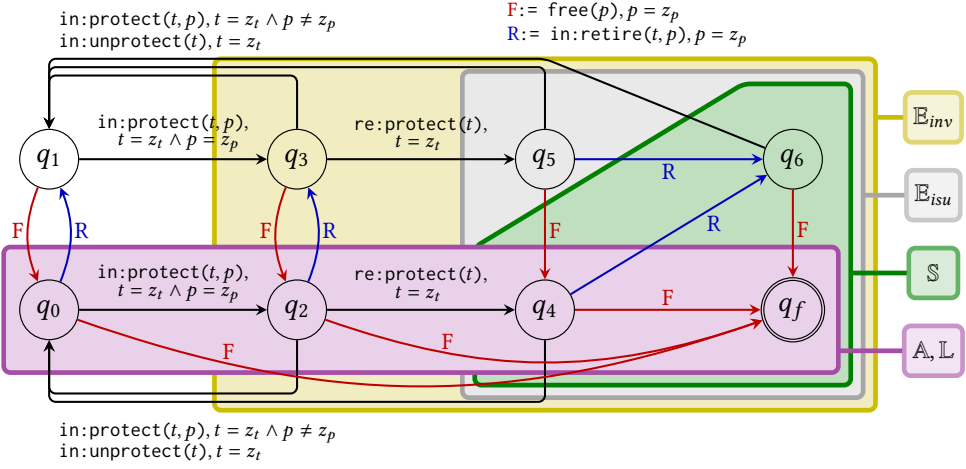
\begin{figure}
    \caption{SMR automaton for HP, adapted from \cite{WolffPhd}.}
    \label{fig:HPsmr}
    \centering
    \scalebox{1.1}{
    \begin{tikzpicture}
        \tikzset{
          sh2n1/.style={shift={(0,1.0)}},
          sh2n2/.style={shift={(0,1.1)}},
          sh2n3/.style={shift={(0,1.2)}},
          sh2s1/.style={shift={(0,-1.0)}},
          sh2s2/.style={shift={(0,-1.1)}},
          sh2s3/.style={shift={(0,-1.2)}},
          sh2e/.style={shift={(1,0)}},
          sh2w/.style={shift={(-1,0)}},
          sh2nw/.style={shift={(-1,1)}},
          sh2ne/.style={shift={(1,1)}},
          sh2sw/.style={shift={(-1,-1)}},
          sh2se/.style={shift={(1,-1)}},
          rc/.style={rounded corners=0.8mm,line width=0.7pt},
          place/.style={draw,circle},
          label/.style={draw, rounded corners=0.8mm, ultra thick, minimum width=0.7cm, minimum height=0.7cm, scale=0.85},
          acc/.style={draw,circle,double=magenta!10!violet!10, double distance=2pt},
          func/.style={scale=0.7},
        }

        \draw [ultra thick, draw=olive!50!yellow, fill=olive!50!yellow!25, opacity=0.6, rounded corners=0.8mm]
        (10,1.2) -- (2,1.2) -- (2,-3.0) -- (10,-3.0) -- cycle;

        \draw [ultra thick, draw=gray!70!white, fill=gray!70!white!25, opacity=0.6, rounded corners=0.8mm]
        (9.85,1.05) -- (5,1.05) -- (5,-2.85) -- (9.85,-2.85) -- cycle;

        \draw [ultra thick, draw=green!50!black, fill=green!50!black!25, opacity=0.6, rounded corners=0.8mm]
        (9.7,0.9) -- (8.4,0.9) -- (5.15, -1.05) -- (5.15,-2.7) -- (9.7,-2.7) -- cycle;

        \draw [ultra thick, draw=violet!70!white, fill=violet!70!white!25, opacity=0.6, rounded corners=0.8mm]
        (9.55,-1.05) -- (-0.55, -1.05) -- (-0.55,-2.55) -- (9.55,-2.55) -- cycle;

        \node[label, draw=olive!50!yellow!60, fill=olive!50!yellow!15] (Einv) at (10.7, 0.5) {$\EinvGuarantee{}$};
        \draw [ultra thick, draw=olive!50!yellow, opacity=0.6]

        (Einv) -- (10, 0.5);

        \node[label, draw=gray!70!white!60, fill=gray!70!white!15] (Eisu) at (10.7, -0.43) {$\EisuGuarantee{}$};
        \draw [ultra thick, draw=gray!70!white, opacity=0.6]

        (Eisu) -- (9.85, -0.43);

        \node[label, draw=green!50!black!60, fill=green!50!black!15] (safe) at (10.7, -1.36) {$\safeGuarantee{}$};
        \draw [ultra thick, draw=green!50!black, opacity=0.6]
        (safe) -- (9.7, -1.36);

        \node[label, draw=violet!70!white!60, fill=violet!70!white!15] (active) at (10.7, -2.29) {$\activeGuarantee{}, \localGuarantee{}$};
        \draw [ultra thick, draw=violet!70!white, opacity=0.6]
        (active) -- (9.55, -2.29);

        \node[func] (Flabel) at (7, 1.6) {$\begin{aligned}
        &{\text{\color{red!75!black}F}} \code{:= free(}p\code{)}, p = z_p 
        \\[-0.8ex]
        &{\text{\color{blue!75!black}R}} \code{:= in:retire(}t, p\code{)}, p = z_p
        \end{aligned}$};

        \node[place] (1) at (0,0) {$q_1$};
        \node[place] (3) at (3,0) {$q_3$};
        \node[place] (5) at (6,0) {$q_5$};
        \node[place] (6) at (9,0) {$q_6$};

        \node[place] (0) at (0,-1.8) {$q_0$};
        \node[place] (2) at (3,-1.8) {$q_2$};
        \node[place] (4) at (6,-1.8) {$q_4$};
        \draw [draw=black]  (9,-1.8) circle (10pt);
        \node [place] (f) at (9,-1.8) {$q_f$};
      

        \draw[-stealth,rc] (6) -- ([sh2n3]5.center) -- node[scale=0.7, pos=0.7][above, font=\linespread{0.5}\selectfont]{$\begin{aligned}
        &\code{in:protect(}t,p\code{)}, t = z_t \wedge p \neq z_p
        \\[-0.8ex]& \code{in:unprotect(}t\code{)}, t = z_t
        \end{aligned}$} ([sh2n3]1.center) -- (1);
        \draw[-stealth,rc] (5) -- ([sh2n2]5.center) -- ([sh2n2]1.center) -- (1);
        \draw[-stealth,rc] (3) -- ([sh2n1]3.center) -- ([sh2n1]1.center) -- (1);

        \draw[-stealth,rc] (4) -- ([sh2s2]4.center) -- node[scale=0.7, pos=0.7][below=0.15cm, font=\linespread{0.5}\selectfont]{$\begin{aligned}
        &\code{in:protect(}t,p\code{)}, t = z_t \wedge p \neq z_p
        \\[-0.8ex]& \code{in:unprotect(}t\code{)}, t = z_t
        \end{aligned}$} ([sh2s2]0.center) -- (0);
        \draw[-stealth,rc] (2) -- ([sh2s1]2.center) -- ([sh2s1]0.center) -- (0);

        \draw[-stealth,rc] (1) -- node[scale=0.7][above, font=\linespread{0.5}\selectfont]{$\begin{aligned}
        &\code{in:protect(}t,p\code{)},
        \\[-0.8ex]& \quad t = z_t \wedge p = z_p
        \end{aligned}$} (3);

        \draw[-stealth,rc] (0) -- node[scale=0.7][above, font=\linespread{0.5}\selectfont]{$\begin{aligned}
        &\code{in:protect(}t,p\code{)},
        \\[-0.8ex]& \quad t = z_t \wedge p = z_p
        \end{aligned}$} (2);

        \draw[-stealth,rc] (3) -- node[scale=0.7][above, font=\linespread{0.5}\selectfont]{$\begin{aligned}
        &\code{re:protect(}t\code{)},
        \\[-0.8ex]& \qquad t = z_t
        \end{aligned}$} (5);

        \draw[-stealth,rc] (2) -- node[scale=0.7][above, font=\linespread{0.5}\selectfont]{$\begin{aligned}
        &\code{re:protect(}t\code{)},
        \\[-0.8ex]& \qquad t = z_t
        \end{aligned}$} (4);

        \path[-stealth,rc, red!75!black] (0) edge[bend right=19]  node[scale=0.8,red!75!black, above=-0.05, pos =0.21] {F} (f);

        \path[-stealth,rc, red!75!black] (2) edge[bend right=19]  node[scale=0.8,red!75!black, above=-0.05, pos =0.27] {F} (f);

        \path[-stealth,rc, red!75!black] (4) edge  node[scale=0.8,red!75!black, above=-0.05] {F} (f);

        \path[-stealth,rc, red!75!black] (6) edge  node[scale=0.8,red!75!black, right=-0.05, pos=0.45] {F} (f);

        \path[-stealth,rc, red!75!black] (5) edge  node[scale=0.8,red!75!black, right=-0.05, pos=0.45] {F} (4);

        \path[-stealth,rc, red!75!black] (1) edge[bend right=22]  node[scale=0.8,red!75!black, left=-0.05, pos=0.45] {F} (0);

        \path[-stealth,rc, red!75!black] (3) edge[bend right=22]  node[scale=0.8,red!75!black, left=-0.05, pos=0.45] {F} (2);

        \path[-stealth,rc, blue!75!black] (0) edge[bend right=22]  node[scale=0.8,blue!75!black, right=-0.05, pos=0.55] {R} (1);

        \path[-stealth,rc, blue!75!black] (2) edge[bend right=22]  node[scale=0.8,blue!75!black, right=-0.05, pos=0.55] {R} (3);

        \path[-stealth,rc, blue!75!black] (5) edge  node[scale=0.8,blue!75!black, below=-0.05] {R} (6);

        \path[-stealth,rc, blue!75!black] (4) edge  node[scale=0.8,blue!75!black, below, pos=0.73] {R} (6);

    \end{tikzpicture}
    }
\end{figure}

%% file: sections-application-assertionlanguage.tex
\subsection{Assertion Language}\label{Section:AbsInt}
Our assertion language for predicates is $\demonicAssertionsAbs{} = (\vars{} \rightarrow \powersetOf{\automaton{}}) \cup \set{\fail{}}$, and we also call the elements in this domain abstract predicates. 
We annotate our development by $\#$ to indicate that we work with abstract predicates. 
Functions defined on the abstract domain that are annotated with a $\#$ symbol, e.g. $\absVerificationConditionsFunc{}$.
As an abstraction function we have $\demonicAbsFunc{}$ and $\angelicAbsFunc{}$ for predicates and selections, respectively.
Abstract selections are a set of abstract predicates.
Similarly, the concretisation functions are called $\demonicConcFunc{}$ and $\angelicConcFunc{}$.
Since it is deterministic, the original interpretation of commands $\progSemFuncTypes{\code{com}}$ serves as a safe abstract interpretation of $\progSemFunc{\code{com}}$ when the semantics for $\fail{}$ are added.
Thus, there is no need to concretize to the non~abstract domain.
We use the abstract version of verification conditions to proof realizability triples:
\begin{theorem}
    \label{th:abstractVC}
    $
        \models \absVerificationConditionsFuncOf{\hoaretriplet{A}{\code{stmt}}{B}} 
    $ implies $
        \angelicHoareTripletHolds{\angelicConcFuncOf{A}}{\code{stmt}'}{\angelicConcFuncOf{B}}
    $
        .
\end{theorem}
Proofs over programs that start and end in singletons can be replicated in the type system of Meyer and Wolff of which a successful type check is denoted by $\sebHoareTripletDash{}$:
\begin{theorem}
    \label{th:connectionToSeb}
    $
        \demonicHoareTripletHoldsSemantically{\demonicConcFuncOf{a}}{\aprog}{\demonicConcFuncOf{b}}
        $
        implies
        $
        \sebHoareTripletHolds{a}{\aprog}{b}
        $
        .
\end{theorem}

%% file: sections-application-vcexample.tex
\section{Example: \code{pop} Method in Treiber's Stack}\label{Section:Treiber}
We demonstrate our synthesis approach on the \code{pop} method of Treiber's Stack.
First, we insert nonterminals that can resolve to calls to the SMR algorithm and invariant annotations into the program code.
Then, using verification conditions, it is shown that these nonterminals are sufficient for deriving a memory safe program that passes Meyer and Wolff's type check.
Afterwards, using the synthesis function, we derive a memory safe program.  
\begin{figure}
    \caption{Treiber's Stack Pop with inserted nonterminals.
      {\color{materialRed}Red} nonterminals resolve to $\code{skip}$.
      {\color{materialPurple}Purple} nonterminals resolve to the {\color{gray} gray} code beside them. 
    }
    \label{code:exampleVCTreibers}
\begin{lstlisting}[multicols=2]
BEGIN LOOP {
top := TOS; #AC;#
((assume(top == NULL); #AC;# 
  result := EMPTY; #AC;#)
$\color{black} \choice{}$
(
  assume(top != NULL); #AC;#
  ((
    atomic { |\label{line:atomicityAbstractionStart}|
      ~AC;~$\vsim$@\@inv active(TOS);@ |\label{line:invActiveForProtection}|
      top := TOS;| \label{line:topEqualsTOS} |
      ~AC;~$\vsim$@in:protect(top);|\label{line:protections}|
             re:protect();@
      assume(top != NULL); #AC;#
    }|\label{line:atomicityAbstractionEnd}|
    #AC;# next := top.next; #AC;# | \label{line:nextEqualsTopNext} |
    atomic {
      ~AC;~$\vsim$@\@inv active(TOS);@|\label{line:invActiveBeforeCAS}|
     ((assume(CAS(TOS, top, next)); | \label{line:CAS} |
        #AC;#
        flag := true; #AC;#
      )
        $\color{black} \choice{}$
      ( flag := false; #AC;#))
    } 
    #AC;#
    ((assume(flag == true); #AC;# 
      result := top.data; #AC;#) | \label{line:resultEqualsTopData} |
    $\color{black} \choice{}$
    (assume(flag == false); #AC;
     #skip; #AC;#))
  ))
  $\color{black} \choice{}$
  ( skip; #AC;#)
))
}* END LOOP
\end{lstlisting}
\end{figure}

In \Cref{code:exampleVCTreibers}, the pop method of Treiber's stack is depicted.
Therein, we already inserted the nonterminal $\code{AC}$.
Insertions that will resolve to $\code{skip}$ are marked in red.
Insertions that will resolve to commands other than skip are marked in purple with the commands they will resolve to in gray right beside them.
The atomic block in Lines~\ref{line:atomicityAbstractionStart} to \ref{line:atomicityAbstractionEnd} serves as an atomicity abstraction for $\code{assume(top = TOS)}$.
The assumption is required to check if a protection was successful.
Without diving into details, the atomicity abstraction is required in Meyer and Wolff's type system because this assumption is at risk of a harmless ABA.
A type check does not pass on a harmless ABA.
(This is a detail we left out in the overview in \Cref{sec:preliminaries:overview}.)
The atomicity abstraction circumvents this problem.

In the type system, dereferences and checks for equality can only be performed on pointers that are safe to access, i.e.\ pointers that possess the active, local, or safe guarantee.
Since the $\code{next}$ pointer is never accessed or compared, it does not need any protection.
We therefore omit it in the following.
As mentioned before, we assume the shared pointer $\code{TOS}$ to always be active, therefore a protection is not needed because having the active guarantee $\activeGuarantee{}$ is sufficient for safe dereferences and comparisons.
However, signaling the type system that the pointer is indeed active is still required and is done using invariant annotations.
Observe that $\code{TOS}$ needs to be active at least in Line \ref{line:CAS} as it is compared to $\code{top}$ there.
The pointer $\code{top}$ needs to be protected at multiple locations.
It is accessed and compared to another pointer in Lines \ref{line:nextEqualsTopNext}, \ref{line:CAS} and \ref{line:resultEqualsTopData}.
Knowing which pointers need to be protected, we restrict nonterminal to the definition depicted in \Cref{code:insertionsHPTreibers:appendix}.
(This optimization is not required.)

\begin{figure}
    \caption{Insertions for HP in Treiber's Stack.}
    \label{code:insertionsHPTreibers:appendix}
\begin{lstlisting}
AC ::= skip; $\color{black}\angelicChoice{}$
atomic {\@inv active(TOS);} $\color{black}\angelicChoice{}$
(in:protect(top); re:protect(top);)
\end{lstlisting}
    \end{figure}

We call the resulting program sketch $\asketch$.
Using verification conditions, the realizability triple
$
    \angelicHoareTripletHolds{\demonicConcFuncOf{(\code{TOS}: \nothingGuarantee{}, \code{top}: \nothingGuarantee{})}}{\asketch}{\demonicConcFuncOf{(\code{TOS}: \nothingGuarantee{}, \code{top}: \nothingGuarantee{})}}   
$
is proven.
In order to generate verification conditions, the loop of $\asketch$ needs to be annotated with a loop invariant.
We choose the invariant $I = \set{(\code{TOS}: \nothingGuarantee{}, \code{top}: \nothingGuarantee{})}$.
In total, the verification conditions function generates 27 unique comparisons. 
The result is that every comparison holds.
Therefore, the above realizability triple is true.
Inserting the selections used for generating the verification conditions into the program yields a proof outline $\apo$ with $\proofToProgFuncOf{\apo} = \asketch$.
Having shown that there is a way to concretize this program sketch, we now use the synthesis algorithm to eliminate the nonterminals.
Thus, we call the synthesis function with the predicate  
$(\code{TOS}: \nothingGuarantee{}, \code{top}: \nothingGuarantee{})$
as the target for condition.
In Lines \ref{line:invActiveForProtection}, \ref{line:protections} and \ref{line:invActiveBeforeCAS} the nonterminals resolve to commands other than $\code{skip}$. 
In the following, we discuss the invocations of the synthesis function where the nonterminal does not resolve to $\code{skip}$.
\begin{figure}
    \caption{Recursive calls of decision function}
\smaller
\begin{gather}
\begin{multlined}
\synof{
    \hoaretriplet
    {\angelicConcFuncOf{\set{(\code{TOS}: \nothingGuarantee{}, \code{top}: \nothingGuarantee{}),
    (\code{TOS}: \nothingGuarantee{}, \code{top}: \EinvGuarantee{} \wedge \EisuGuarantee{})
    }}}
    {\code{AC}}
    {\angelicConcFuncOf{\set{
        (\code{TOS}: \activeGuarantee{}, \code{top}: \nothingGuarantee{}),
        \\
        \qquad \qquad \quad
        (\code{TOS}: \activeGuarantee{}, \code{top}: \EinvGuarantee{} \wedge \EisuGuarantee{}), 
        (\code{TOS}: \nothingGuarantee{}, \code{top}: \EinvGuarantee{} \wedge \EisuGuarantee{}),
        (\code{TOS}: \nothingGuarantee{}, \code{top}: \nothingGuarantee{})
        }}}
}
{\\ \demonicConcFuncOf{(\code{TOS}: \activeGuarantee{}, \code{top}: \nothingGuarantee{})}}
\\
=
(\demonicConcFuncOf{(\code{TOS}: \nothingGuarantee{}, \code{top}: \nothingGuarantee{})} ,\code{@inv active(TOS);})
\end{multlined}
\label{eq:invActiveForProtection}
\\
\begin{multlined}
\label{eq:protections}
 \synof
 {
    \hoaretriplet
    {
    \angelicConcFuncOf{
        \set{
            (\code{TOS}: \nothingGuarantee{}, \code{top}: \nothingGuarantee{}),
            (\code{TOS}: \activeGuarantee{}, \code{top}: \activeGuarantee{})
            }}
    }
    {
        \code{AC}
    }
    {
    \angelicConcFuncOf{\set{
        (\code{TOS}: \activeGuarantee{}, \code{top}: \activeGuarantee{}),
        \\
        \qquad \qquad \qquad \qquad \qquad \qquad \quad \:
        (\code{TOS}: \activeGuarantee{}, \code{top}: \EinvGuarantee{} \wedge \EisuGuarantee{} \wedge \activeGuarantee{} \wedge \safeGuarantee{}),
        (\code{TOS}: \activeGuarantee{}, \code{top}: \nothingGuarantee{}),
        \\
        \qquad \qquad \qquad \qquad \qquad \qquad \qquad \qquad
        (\code{TOS}: \nothingGuarantee{}, \code{top}: \nothingGuarantee{}),
        (\code{TOS}: \nothingGuarantee{}, \code{top}: \EinvGuarantee{} \wedge \EisuGuarantee{})
        }}
    }
 }
 {
   \\
    \demonicConcFuncOf{(\code{TOS}: \activeGuarantee{}, \code{top}: \EinvGuarantee{} \wedge \EisuGuarantee{} \wedge \activeGuarantee{} \wedge \safeGuarantee{})}
 }
 \\
=
 (
  (\code{TOS}: \activeGuarantee{}, \code{top}: \activeGuarantee{})
 ,(\code{in:protect(top); re:protect();}))
\end{multlined}
\\
\begin{multlined}
\label{eq:invActiveBeforeCAS}
\synof{
\hoaretriplet
{
\angelicConcFuncOf{\set{
(\code{TOS}: \nothingGuarantee{}, \code{top}: \EinvGuarantee{} \wedge \EisuGuarantee{} \wedge \safeGuarantee{}),
\fail{} 
}
}}
{
\code{AC}
}
{
\angelicConcFuncOf{\set{
(\code{TOS}: \activeGuarantee{}, \code{top}: \EinvGuarantee{} \wedge \EisuGuarantee{} \wedge \safeGuarantee{}),
\\
\qquad \qquad \qquad \qquad \qquad \qquad \qquad \qquad \qquad \qquad \quad
(\code{TOS}: \nothingGuarantee{}, \code{top}: \EinvGuarantee{} \wedge \EisuGuarantee{} \wedge \safeGuarantee{}),
\fail{}
}}
}
}
{
\\
\demonicConcFuncOf{(\code{TOS}: \activeGuarantee{}, \code{top}: \EinvGuarantee{} \wedge \EisuGuarantee{} \wedge \safeGuarantee{})}
}
\\
=
( 
(\code{TOS}: \nothingGuarantee{}, \code{top}: \EinvGuarantee{} \wedge \EisuGuarantee{} \wedge \safeGuarantee{})
,\code{@inv active(TOS);})
\end{multlined}
\end{gather}
\end{figure}
In \Cref{eq:invActiveForProtection}, the active invariant annotation for $\code{TOS}$ is chosen.
This insertion can be seen in Line \ref{line:invActiveForProtection}.
Here, executing the active invariant annotation on the first and second predicate in the pre condition results in the first and second predicate of the post condition, respectively.
The third predicate of the post condition can be reached from both of the predicates in the pre condition.
From the first predicate, a protection to $\code{top}$ needs to be issued.
From the second predicate, a $\code{skip}$ is sufficient.
The last predicate of the post condition is reached through executing $\code{skip}$ on the first predicate of the pre condition.
Thus, the first, third and fourth predicates of the post condition are reachable from the target pre condition.
The first predicate of the post condition matches.
Therefore, $\code{@inv active}$ needs to be inserted and the first predicate of the precondition is returned.
Next, in \Cref{eq:protections}, the protections for $\code{top}$ are synthesized.
The result is inserted in Line \ref{line:protections}.
There are 5 predicates in the post condition of the nonterminal.
The first two predicates are results of executing the nonterminal on the second predicate of the precondition.
The first predicate of the postcondition stems from calling $\code{skip}$ or $\code{@inv active(TOS)}$.
The second predicate is the result of issuing the protection for $\code{top}$.
The third, fourth and fifth predicate are results of executing the nonterminal on the first predicates of the pre condition.
The third predicate is the product of inserting the protection for $\code{top}$.
The fourth predicate comes from a $\code{skip}$. 
The last predicate is the result of calling $\code{@inv active(TOS)}$.
The second predicate of the post condition matches target post condition.
Therefore, the protections for $\code{top}$ must be synthesized here and the second predicate of the precondition is returned.

Lastly, in \Cref{eq:invActiveBeforeCAS}, the active annotation in Line \ref{line:invActiveBeforeCAS} before the compare and swap is synthesized.
There are two predicates in the pre condition, one of which is $\fail{}$.
Executing the nonterminal on $\fail{}$ leads to $\fail{}$ in the post condition.
Using the invariant annotation on the first predicate of the pre condition yields the first predicate of the post condition.
Inserting $\code{skip}$ or a protection to $\code{top}$ leads to the second predicate.
The first predicate of the post condition matches the target post condition, thus $\code{@inv active(TOS)}$ is inserted and the first predicate from the precondition is returned.

We call the version of Treiber's stack with the resolved nonterminals from \Cref{code:exampleVCTreibers} $\aprog$.
We know that the Hoare triple
$
    \demonicHoareTripletHoldsSemantically{\demonicConcFuncOf{(\code{TOS}: \nothingGuarantee{}, \code{top}: \nothingGuarantee{})}}{\aprog}{\demonicConcFuncOf{(\code{TOS}: \nothingGuarantee{}, \code{top}: \nothingGuarantee{})}}
$
holds due to \Cref{Theorem:syn}.
Applying \Cref{th:connectionToSeb}, we conclude that the resulting program type checks:
$
\sebHoareTripletHolds{(\code{TOS}: \nothingGuarantee{}, \code{top}: \nothingGuarantee{})}{\aprog}{(\code{TOS}: \nothingGuarantee{}, \code{top}: \nothingGuarantee{})}
$
is true.
That means, we successfully synthesized a program $\aprog$ that passes Meyer and Wolff's type check and therefore is memory safe (if the invariant annotations can be discharged).